\documentclass[11pt]{article}

\usepackage[margin=1in]{geometry}
\usepackage{amsmath,amssymb,amsthm,mathtools}
\usepackage{physics}
\usepackage{enumitem}
\usepackage{microtype}
\usepackage{needspace}
\usepackage{algorithm}
\usepackage{algpseudocode}
\usepackage{authblk}
\usepackage{tikz}
\usetikzlibrary{arrows.meta,decorations.pathreplacing}
\usepackage[numbers,comma,sort&compress]{natbib}
\usepackage[colorlinks=true,allcolors=blue]{hyperref}

\newtheorem{theorem}{Theorem}[section]
\newtheorem{proposition}[theorem]{Proposition}
\newtheorem{lemma}[theorem]{Lemma}
\theoremstyle{definition}
\newtheorem{definition}[theorem]{Definition}
\newtheorem{remark}[theorem]{Remark}

\newcommand{\cH}{\mathcal H}
\newcommand{\cL}{\mathcal L}
\newcommand{\cE}{\mathcal E}
\newcommand{\cK}{\mathcal K}
\newcommand{\Id}{\mathrm{Id}}
\providecommand{\Tr}{}
\renewcommand{\Tr}{\operatorname{Tr}}
\renewcommand{\ket}[1]{\lvert #1\rangle}
\renewcommand{\bra}[1]{\langle #1\rvert}
\newcommand{\proj}[1]{\lvert #1\rangle\!\langle #1\rvert}
\renewcommand{\norm}[1]{\left\lVert #1\right\rVert}
\title{Query-Optimal and Gate-Efficient Lindbladian Simulation}
\author[1]{Boyang Chen}
\author[2,3]{Minbo Gao}
\author[4,5]{Xinzhao Wang$^{*}$}
\author[4,5]{Shuo Zhou}

\affil[1]{Department of Computer Science and Technology,
Tsinghua University, Beijing, China}
\affil[2]{Institute of Software,
Chinese Academy of Sciences, Beijing, China}
\affil[3]{University of Chinese Academy of Sciences,
Beijing, China}
\affil[4]{Center on Frontiers of Computing Studies,
Peking University, Beijing, China}
\affil[5]{School of Computer Science,
Peking University, Beijing, China}
\date{}

\begin{document}
\maketitle

\begingroup
\renewcommand{\thefootnote}{}
\renewcommand{\theHfootnote}{authors-alphabetical}
\footnotemark
\footnotetext{\begin{tabular}[t]{@{}l@{}}
Authors are listed alphabetically.\\
$^{*}$Corresponding author: \texttt{xinzhaowang3@gmail.com}
\end{tabular}}
\endgroup

\begin{abstract}
We give a quantum algorithm for Lindbladian simulation given a block encoding
of the Hamiltonian $H$ and a projected unitary encoding of the stacked jump operator
$B=\sum_{k=1}^m\ket{k}\otimes L_k$, with normalization factors
$\alpha_H$ and $\alpha_B$, respectively.  For evolution time $t$, set
$\tau=(\alpha_H+\alpha_B^2)t$.  The algorithm approximates the evolution
channel to diamond-norm error $\varepsilon$ using
\[
    O\!\left(
        \tau+
        \frac{\log(1/\varepsilon)}
        {\log\!\left(e+\log(1/\varepsilon)/\tau\right)}
    \right)
\]
oracle queries, matching the query lower bound for Hamiltonian simulation.
The number of additional one- and two-qubit gates is linear in the query
complexity up to polylogarithmic factors.
The query- and gate-complexity bounds extend to
Lipschitz-continuous time-dependent Lindbladians under coherent time-indexed
oracle access.
Our construction uses a one-query transducer that implements a product of
rational approximations to short-time evolution when supplied with a
catalyst.  We bound
the error from omitting the catalyst by exploiting orthogonality between
different sequences of Kraus labels.
The gate implementation combines a compressed Kraus-label representation,
which stores only the positions and values of the nonzero labels, with the
rotation factorization of Chen et al.~\cite{ChenEtAl2026Gate}.
\end{abstract}

\newpage
\section{Introduction}

Quantum simulation is a fundamental problem in quantum computation. While
closed quantum systems are described by Hamiltonian dynamics, Markovian
open-system dynamics is commonly described by Lindblad master
equations~\cite{GoriniKossakowskiSudarshan1976,Lindblad1976}. Lindbladian
dynamics describes dissipative processes such as decoherence, relaxation,
and thermalization~\cite{BreuerPetruccione2002}. Beyond modeling physical
dissipation, Lindblad evolution has been used for dissipative quantum
computation and ground- and thermal-state
preparation~\cite{VWI09,KlieschEtAl2011,DingChenLin2024,ChenEtAl2025Thermal}.
It also serves as an algorithmic primitive for solving linear differential
equations, nonconvex optimization, and optimal control of open quantum
systems~\cite{SGAZ25,ChenEtAl2025Langevin,ChenEtAl2025LocalMinima,HeEtAl2024Control}. Efficient
simulation of Lindbladian evolution is therefore a basic problem in the
study of open-system quantum dynamics and quantum algorithms.

A natural question is to determine the query complexity of Lindbladian simulation as a function of the evolution time and the target precision. Hamiltonian simulation provides a natural comparison: when the dissipative terms vanish, Lindbladian evolution reduces to Hamiltonian evolution, for which the optimal joint dependence on evolution time and precision is known.
For Hamiltonian simulation, let $t$ denote the evolution time and
$\tau=\alpha_H t$ the normalized evolution time, where $\alpha_H$ is the
block-encoding normalization of $H$. For $\tau\geq1$, the optimal
query complexity
is~\cite{LowChuang2019,GilyenEtAl2019}
\begin{equation}
\Theta\!\left(
\tau+
\frac{\log(1/\varepsilon)}
{\log\left(e+\log(1/\varepsilon)/\tau\right)}
\right).
\label{eq:intro-ham-bound}
\end{equation}
This lower bound also applies to Lindbladian simulation. The question is whether it can be achieved in the presence of arbitrary dissipation, or whether dissipation introduces an additional asymptotic query cost. To summarize, we ask:

\begin{quote}
    \emph{What is the optimal query complexity of simulating Lindbladian dynamics?}
\end{quote}

A sequence of works has progressively improved the complexity of Lindbladian
simulation. In the following comparison, $\tau$ denotes evolution time
normalized according to each algorithm's input model. We display only the
time and precision dependence, suppressing other model-dependent factors, and
our normalization is specified in Section~\ref{sec:model}.
Kliesch et al.~\cite{KlieschEtAl2011} established an $O(\tau^2/\varepsilon)$
simulation algorithm for local Lindbladian dynamics, and Childs and
Li~\cite{ChildsLi2017} developed improved algorithms with complexity
$O(\tau^{3/2}/\sqrt{\varepsilon})$ in the sparse setting, together with an
$\Omega(\tau)$ query lower bound. Cleve and Wang~\cite{CleveWang2017} gave the first algorithm with nearly
linear dependence on the evolution time and polylogarithmic dependence on
the inverse error. Li and
Wang~\cite{LiWang2023} subsequently obtained the query complexity
$$
O\!\left(
\tau\frac{\log(\tau/\varepsilon)}
{\log\log(\tau/\varepsilon)}
\right).
$$
For time-dependent Lindbladians, He et al.~\cite{HeEtAl2024Control} obtained
\(O\!\left(
\tau\log^2(\tau/\varepsilon)/\log^2\log(\tau/\varepsilon)
\right)\) queries.
These results established that high-precision Lindbladian simulation is
possible with nearly linear dependence on the evolution time, but they did
not attain the optimal joint time--precision dependence known for Hamiltonian
simulation. More recently, Borras and Marvian~\cite{BorrasMarvian2026}
showed that an additive dependence on time and precision in the number of
queries to the jump operators can be achieved for Lindbladians satisfying
\(\sum_k L_k^\dagger L_k\propto I\).
For general Lindbladians, however, the optimal joint dependence on time and
precision remained open.

We resolve this question by giving an algorithm for general Lindbladian
simulation that matches the Hamiltonian query lower bound in
\eqref{eq:intro-ham-bound}. In the oracle model specified below, with
normalized evolution time $\tau$, we approximate time-independent
Lindbladian evolution to diamond-norm error $\varepsilon$ using
$$
O\!\left(
\tau+
\frac{\log(1/\varepsilon)}
{\log\!\left(e+\log(1/\varepsilon)/\tau\right)}
\right)
$$
oracle queries.
The same query bound extends to Lipschitz-continuous time-dependent Lindbladians under coherent time-indexed oracle access. The number of additional one- and two-qubit gates is linear in the query complexity up to polylogarithmic factors, with only logarithmic dependence on the number of jump operators.

Our algorithm builds on the transducer framework of Belovs, Jeffery, and Yolcu~\cite{BelovsJefferyYolcu2024} and the Hamiltonian simulation algorithm of Chen et al.~\cite{ChenEtAl2026}. We extend this approach to Lindbladian simulation through a new transducer construction and a catalyst-removal error analysis that establishes the optimal query bound.

\subsection{Model and main result}
\label{sec:model}

Unless a subscript is shown, $\|\cdot\|$ denotes the Euclidean norm for
vectors and the induced operator norm for linear maps.
We write $\log$ for the natural logarithm and $\log_2$ for the base-two
logarithm used in register sizes.

Let \(\cH_{\mathsf S}\) be the system Hilbert space, let \(H\) be Hermitian, and let
\(L_1,\ldots,L_m\) be the jump operators of the Lindbladian. Introduce an
\((m+1)\)-dimensional Kraus-label register \(\mathsf E\), with distinguished state
\(\ket0_{\mathsf E}\) and a nonzero-label subspace
$$
\cH_{\rm L}=\operatorname{span}\{\ket{k}_{\mathsf E}:1\leq k\leq m\},
\qquad
\cH_{\mathsf E}=\mathbb C\ket0_{\mathsf E}\oplus\cH_{\rm L}.
$$
Collect the jump operators into
$$
B=\sum_{k=1}^m\ket{k}_{\mathsf E}\otimes L_k:
\cH_{\mathsf S}\longrightarrow\cH_{\mathsf E}\otimes\cH_{\mathsf S},
\qquad
\operatorname{ran}B\subseteq\cH_{\rm L}\otimes\cH_{\mathsf S}.
$$
Then
\begin{equation}
\cL(\rho)=-i[H,\rho]
+\Tr_{\mathsf E}(B\rho B^\dagger)
-\frac12\{B^\dagger B,\rho\}.
\label{eq:lindbladian}
\end{equation}
For a simulation time \(t\geq0\), our goal is to approximate \(e^{t\cL}\)
in diamond norm.

\begin{definition}[Block encoding~\cite{GilyenEtAl2019}]
\label{def:encoding}
For \(\alpha>0\), an \((s+a)\)-qubit unitary \(U\) is an
\emph{\((\alpha,a,0)\)-block encoding} of an \(s\)-qubit operator \(A\) if
\begin{equation}
(\bra0^{\otimes a}\otimes I)U(\ket0^{\otimes a}\otimes I)=A/\alpha.
\label{eq:block-encoding}
\end{equation}
We call this an \emph{\(\alpha\)-block encoding} when the ancilla count
is not needed.
\end{definition}

For a rectangular map \(A\), we use a projected unitary encoding specified
by input and output isometries. Let \(U\) be unitary, and let
\(\iota_{\rm in},\iota_{\rm out}\) be isometries embedding the domain
and codomain of \(A\) into the space on which \(U\) acts.
The triple \((U,\iota_{\rm in},\iota_{\rm out})\) is a
\emph{projected unitary encoding} of \(A/\alpha\) if
$$
\iota_{\rm out}^\dagger U\iota_{\rm in}=A/\alpha.
$$
Its normalization is \(\alpha\).

We assume access to a Hermitian \(\alpha_H\)-block encoding \(U_H\) of
\(H\) and a projected unitary encoding \(U_B\) of \(B/\alpha_B\), where
\(\alpha_H,\alpha_B>0\). Let \(\mathsf A_H\) and \(\mathsf A_B\) denote
their ancilla registers, which contain \(a\) qubits in total. The Hermitian block
encoding satisfies
\begin{equation}
U_H=U_H^\dagger,\qquad U_H^2=I,\qquad
(\bra0_{\mathsf A_H}\otimes I_{\mathsf S})U_H
(\ket0_{\mathsf A_H}\otimes I_{\mathsf S})=H/\alpha_H.
\label{eq:H-encoding}
\end{equation}
The Hermiticity requirement incurs only constant overhead.  Let $\mathsf F$
be a flag qubit, and let $\operatorname{Had}_{\mathsf F}$ denote the Hadamard
gate on $\mathsf F$.  If \(U\) is any \(\alpha_H\)-block encoding of the
Hermitian operator \(H\), then
$$
\widetilde U=(\operatorname{Had}_{\mathsf F}\otimes I)
\left[(\ket{0}\!\bra{1})_{\mathsf F}\otimes U
+(\ket{1}\!\bra{0})_{\mathsf F}\otimes U^\dagger\right]
(\operatorname{Had}_{\mathsf F}\otimes I)
$$
is a Hermitian \(\alpha_H\)-block encoding of \(H\) using one additional
flag qubit \(\mathsf F\).

For the jump-operator encoding, let
$$
P_{\rm L}:=\left(\sum_{k=1}^m\proj{k}_{\mathsf E}\right)\otimes I_{\mathsf S}
=(I_{\mathsf E}-\proj0_{\mathsf E})\otimes I_{\mathsf S}
$$
be the orthogonal projection onto
\(\cH_{\rm L}\otimes\cH_{\mathsf S}\). The projected unitary encoding
\(U_B\) satisfies
\begin{equation}
(\bra0_{\mathsf A_B}\otimes P_{\rm L})U_B
(\ket0_{\mathsf A_B}\ket0_{\mathsf E}\otimes I_{\mathsf S})=B/\alpha_B.
\label{eq:B-encoding}
\end{equation}

To use a single query operator, we package the oracle unitaries into
$$
\Omega=U_H\oplus U_B\oplus U_B^\dagger.
$$
We count each application of \(\Omega\), \(\Omega^\dagger\), or a controlled
version as one query. Equivalently, one may count calls to
\(U_H,U_B,U_B^\dagger\) separately, which changes the query complexity by at
most a constant factor. For the gate complexity, we count all one- and
two-qubit gates other than oracle calls. The total normalization and
normalized evolution time used below are
\begin{equation}
\alpha=\alpha_H+\alpha_B^2,\qquad \tau=\alpha t.
\label{eq:scale}
\end{equation}

\Needspace{12\baselineskip}
Our main result is the following theorem.  For \(\tau\geq1\), its
query complexity matches the Hamiltonian lower bound
\eqref{eq:intro-ham-bound}; the theorem also bounds the number of additional
one- and two-qubit gates.

\begin{theorem}[Query-optimal simulation]
\label{thm:main}
For \(0<\varepsilon\leq1/2\), there is a quantum algorithm implementing a
channel \(\widetilde{\cE}\) satisfying
$$
\norm{\widetilde{\cE}-e^{t\cL}}_\diamond\leq\varepsilon.
$$
It makes
\begin{equation}
O\!\left(
\tau+
\frac{\log(1/\varepsilon)}
{\log\!\bigl(e+\log(1/\varepsilon)/\tau\bigr)}
\right)
\label{eq:upper-complexity}
\end{equation}
queries. With \(\ell_\varepsilon=\log(e(1+\tau)/\varepsilon)\), it uses
\begin{equation}
O\!\left[
\left(\tau+
\frac{\log(1/\varepsilon)}
{\log\!\bigl(e+\log(1/\varepsilon)/\tau\bigr)}\right)
\bigl(a+\ell_\varepsilon[\log(m+1)+\ell_\varepsilon]\bigr)
\right]
\label{eq:main-gate-complexity}
\end{equation}
one- and two-qubit gates, excluding oracle calls.
\end{theorem}

The error guarantee and query bound in Theorem~\ref{thm:main} are proved in
Section~\ref{sec:oaa}, after the construction in
Sections~\ref{sec:rational}--\ref{sec:lcu}.  The gate bound
\eqref{eq:main-gate-complexity} is proved in
Section~\ref{sec:gate-complexity}.

For Lipschitz-continuous time-dependent Lindbladians on \([0,T]\),
Theorem~\ref{thm:time-dependent} achieves the query bound
\eqref{eq:upper-complexity}
under the coherent time-indexed oracle model of
Section~\ref{sec:time-dependent}. Theorem~\ref{thm:gate-complexity}
gives the corresponding gate bound.

Before describing the algorithm, we relate the stacked jump-operator
encoding \eqref{eq:B-encoding} to the individual jump-operator access model
used in earlier work.

\begin{remark}[Individual jump-operator access]
\label{rem:individual-jump-access}
Some previous Lindbladian simulation algorithms, including Li and Wang~\cite{LiWang2023},
assume a separate \(\alpha_k\)-block encoding for each jump operator \(L_k\).
Our access model also covers this setting when these block encodings can be
applied coherently conditioned on \(k\).  Preparing
\[
\frac{1}{\alpha_B}\sum_{k=1}^m \alpha_k\ket{k}_{\mathsf E},
\qquad
\alpha_B^2=\sum_{k=1}^m\alpha_k^2,
\]
and applying the corresponding block encoding of \(L_k\) conditioned on \(k\)
gives a projected unitary encoding of
\[
\frac{1}{\alpha_B}\sum_{k=1}^m\ket{k}_{\mathsf E}\otimes L_k
=\frac{B}{\alpha_B}.
\]
Chen et al.~\cite{ChenEtAl2025Thermal} instead assume a block encoding
of the stacked operator \(\sum_k\ket{k}\otimes L_k\), which is essentially the
access model used here.
\end{remark}

For $\tau\geq1$, the Hamiltonian lower bound in \eqref{eq:intro-ham-bound}
\cite{LowChuang2019,GilyenEtAl2019} applies here by setting \(B=0\).
In this Hamiltonian-only case, the jump-operator oracle can be omitted,
and the corresponding normalized evolution time is $\tau=\alpha_H t$.
This lower bound matches the upper bound \eqref{eq:upper-complexity}, establishing
query optimality in this regime.

\subsection{Algorithm overview}
\label{sec:overview}

We approximate $e^{t\cL}$ by a product of $J$ completely positive and
trace-preserving (CPTP) steps. Increasing $J$
reduces the discretization error; the main task is to implement this product
with a query count that does not grow with $J$.
We follow the simulation framework of~\cite{ChenEtAl2026}, based on the
transducers of~\cite{BelovsJefferyYolcu2024}.  The Lindbladian-specific
ingredients are the local transducer, the error analysis for omitting
the catalyst, and the implementation and gate-complexity analysis of
the Kraus-label updates.

\paragraph{A CPTP approximation.}
We first approximate the evolution over one time step $\delta=t/J$ by a CPTP
map. Define
\begin{equation}
 K=iH+\frac12B^\dagger B,
 \qquad
 R_\delta=(I+\delta K/2)^{-1},
 \label{eq:overview-resolvent}
\end{equation}
\begin{equation}
 N_\delta=(I-\delta K/2)R_\delta,
 \qquad
 J_\delta=\sqrt\delta\,BR_\delta.
 \label{eq:overview-kraus}
\end{equation}
Since $K+K^\dagger=B^\dagger B$,
$N_\delta^\dagger N_\delta+J_\delta^\dagger J_\delta=I$, and hence
\begin{equation}
 W_\delta\psi
 =\ket0_{\mathsf E}\otimes N_\delta\psi+J_\delta\psi
 \label{eq:overview-step}
\end{equation}
is a Stinespring isometry for the quantum channel
\[
\begin{aligned}
\Phi_\delta(\rho)
:=\Tr_{\mathsf E}(W_\delta\rho W_\delta^\dagger)=N_\delta\rho N_\delta^\dagger
  +\sum_{k=1}^m\delta L_kR_\delta\rho R_\delta^\dagger L_k^\dagger.
\end{aligned}
\]
Its Kraus operators are \(N_\delta\) and
\(\sqrt\delta L_kR_\delta\), \(1\leq k\leq m\). Accordingly, the register
\(\mathsf E\) stores the corresponding \emph{Kraus label}: \(0\) for
\(N_\delta\) and \(k\) for \(\sqrt\delta L_kR_\delta\).
Applying $W_\delta$ successively with fresh registers
$\mathsf E_0,\ldots,\mathsf E_{J-1}$ gives a Stinespring isometry $W_J$ for
$\Phi_\delta^J$, with output space
$\cH_J=\cH_{\mathsf E_0}\otimes\cdots\otimes
\cH_{\mathsf E_{J-1}}\otimes\cH_{\mathsf S}$.
Section~\ref{sec:rational} proves that, when $\alpha\delta\leq1/2$,
\begin{equation}
 \|\Phi_\delta^J-e^{t\cL}\|_\diamond\leq\frac{10\tau^2}{J}.
 \label{eq:overview-discretization}
\end{equation}

\paragraph{Implementing the product with a catalyst.}
To implement \(W_J\), we first construct a one-query
transducer~\cite{BelovsJefferyYolcu2024} for the single-step isometry
\(W_\delta\).  In the transducer formalism,
\(\cH_{\mathsf E}\otimes\cH_{\mathsf S}\) is the public space, and we
introduce a private space \(\cH_{\rm priv}\).  We construct a linear map
\(\gamma:\cH_{\mathsf S}\to\cH_{\rm priv}\) and an oracle-independent
unitary \(G\) on the direct sum of these two spaces.  The local transducer
\(S_\delta\) first applies
\(\Omega=U_H\oplus U_B\oplus U_B^\dagger\) to the private space and then
applies \(G\).  It satisfies
\begin{equation}
 S_\delta\bigl((\ket0_{\mathsf E}\otimes\psi)\oplus\gamma\psi\bigr)
 =W_\delta\psi\oplus\gamma\psi,
 \qquad
 \gamma^\dagger\gamma\preceq\alpha\delta I.
 \label{eq:overview-local}
\end{equation}
Here \(\ket0_{\mathsf E}\otimes\psi\) and \(W_\delta\psi\) are the public
input and output, respectively, while \(\gamma\psi\) is the \emph{catalyst} in the
private space and is returned unchanged.

The construction extends the Cayley transducer for Hamiltonian simulation in
\cite[Sec.~4.2]{ChenEtAl2026} to Lindbladian simulation.
When \(B=0\), we have \(J_\delta=0\).
In this case,
\[
 \left.W_\delta\psi\right|_{B=0}
 =\ket0_{\mathsf E}\otimes
 (I-i\delta H/2)(I+i\delta H/2)^{-1}\psi,
\]
where the operator on the right is the Cayley transform of \(H\).  For general
\(B\), the Kraus operator
\[
 N_\delta=(I-\delta K/2)(I+\delta K/2)^{-1}
\]
is obtained from this Cayley transform by replacing \(iH\) with
\(K=iH+B^\dagger B/2\).  As in the Hamiltonian Cayley transducer, the
\(U_H\) summand of \(\Omega\) supplies the \(iH\) term in \(K\).  The
\(U_B^\dagger\) summand supplies the \(B^\dagger B/2\) term, while the
\(U_B\) summand produces the jump output
\(J_\delta\psi=\sqrt\delta\,BR_\delta\psi\).
Figure~\ref{fig:local-catalyst-action} illustrates these roles.
Section~\ref{sec:local-step} gives the full construction and proves that the
catalyst is restored.

In Section~\ref{sec:full-product}, we use sequential composition of
transducers~\cite[Prop.~9.9]{BelovsJefferyYolcu2024}, as in
\cite[Sec.~4.2]{ChenEtAl2026}, to combine \(J\) copies of the local transducer
\(S_\delta\), each implementing one application of \(W_\delta\), into a
one-query transducer \(S\) implementing the product isometry \(W_J\).
The catalyst of $S$ is the direct sum of the
catalysts for all local transducers. Let $\mathsf T$ denote the register
labeling these private-space components, with basis label
$\ket{j}_{\mathsf T}$ corresponding to step $j$, $0\leq j<J$.
A single query applies $\Omega$ to
the private factor of every direct-sum component, after which the fixed
local unitaries act in time order.  Writing $\cK$ for the combined private
space, $S$ acts on the public space $\cH_J$ and the private space $\cK$,
with a catalyst map $\Gamma$ satisfying
\begin{equation}
 S(\psi\oplus\Gamma\psi)=W_J\psi\oplus\Gamma\psi,
 \qquad
 \Gamma^\dagger\Gamma\preceq\tau I.
 \label{eq:overview-global}
\end{equation}

\paragraph{Reuse without preparing the catalyst.}
We now approximate $W_J$ without preparing $\Gamma\psi$.
The reuse construction~\cite[Theorem~3.2]{BelovsJefferyYolcu2024}
distributes the amplitude of the public input uniformly over \(N\) mutually
orthogonal copies of \(\cH_J\) and applies \(S\) successively to them.  For
comparison, suppose that this circuit is supplied with the scaled catalyst
\(\Gamma\psi/\sqrt N\).  By linearity of the global transducer identity
\eqref{eq:overview-global}, each call satisfies
\[
 S\left(\frac{\psi}{\sqrt N}\oplus\frac{\Gamma\psi}{\sqrt N}\right)
 =\frac{W_J\psi}{\sqrt N}\oplus\frac{\Gamma\psi}{\sqrt N}.
\]
The private output equals the private input, so it can be passed unchanged to
the next call.  After all \(N\) calls, coherently recombining the public
outputs produces \(W_J\psi\) exactly.  The algorithm instead runs the same
circuit with the zero vector as its private input.  Let
\(P_N:\cH_{\mathsf S}\to\cH_J\) denote
the map obtained by recombining its public outputs.
Section~\ref{sec:reuse} shows that this circuit provides a projected unitary
encoding of \(P_N\).  For \(\norm\psi=1\), supplying the scaled catalyst
instead of the zero private vector changes the circuit input by
\(\Gamma\psi/\sqrt N\).
Unitarity preserves this distance, and projecting onto the public output
cannot increase it.  Hence
\[
 \|W_J-P_N\|\leq\frac{\|\Gamma\|}{\sqrt N}
 \leq\sqrt{\frac{\tau}{N}}.
\]
In the nontrivial regime $2\tau>\varepsilon$, this estimate gives
\(O(\tau/\varepsilon^2)\) calls to reach error \(\varepsilon\), with
quadratic dependence on \(1/\varepsilon\).

To improve the precision dependence, we use a more precise expression for
\(W_J-P_N\).  Define the blocks
\(S_{01}:\cK\to\cH_J\) and \(S_{11}:\cK\to\cK\) by
\[
 S(0\oplus z)=S_{01}z\oplus S_{11}z.
\]
Because \(S\) is unitary, \(\norm{S_{01}},\norm{S_{11}}\leq1\).  For the
polynomial \(g_N(\zeta):=N^{-1}\sum_{\ell=0}^{N-1}\zeta^\ell\), the reuse
identity proved in Section~\ref{sec:reuse}, following
\cite[Lemma~4]{ChenEtAl2026}, gives
\begin{equation}
 W_J-P_N=S_{01}g_N(S_{11})\Gamma.
 \label{eq:overview-reuse}
\end{equation}
Consequently, \(\|W_J-P_N\|\leq\sqrt\tau\,\|g_N(S_{11})\|\).
The norm $\|g_N(S_{11})\|$ need not decrease rapidly with $N$, so
this estimate alone does not give rapid error decay for $P_N$.
Following \cite[Sec.~7]{ChenEtAl2026}, we use linear combination of unitaries
(LCU)~\cite{GilyenEtAl2019} to construct a projected unitary encoding of a
linear combination of the \(P_N\).  Writing
\(\Lambda=\sum_N|\lambda_N|\), the resulting projected block is
\(\Lambda^{-1}\sum_N\lambda_NP_N\).
For coefficients satisfying $\sum_N\lambda_N=1$,
\begin{equation}
 W_J-\sum_N\lambda_NP_N
 =S_{01}\left[\sum_N\lambda_Ng_N(S_{11})\right]\Gamma.
 \label{eq:overview-weighted-error}
\end{equation}
The coefficients should make
$\|\sum_N\lambda_Ng_N(S_{11})\|$ small while keeping the LCU normalization
$\Lambda$ bounded, since $\Lambda$ controls the cost of the subsequent
amplitude amplification.

\paragraph{Reducing the approximation error.}
To choose coefficients $\{\lambda_N\}$ for which $\sum_N\lambda_Ng_N(S_{11})$ has small norm,
we first construct a polynomial whose norm at $S_{11}$ decays rapidly.
Let $S_{11}(0)$ be obtained from $S_{11}$ by setting $\delta=0$ in the
post-query unitary $G$ of each local transducer $S_\delta$, while leaving
$\Omega$ unchanged.
Appendix~\ref{sec:private-block-proof} shows that its square is the negative of a
projection, so the polynomial $p(\zeta):=\zeta^2(1+\zeta^2)$ satisfies
\begin{equation}
 p(S_{11}(0))=S_{11}(0)^2\bigl[I+S_{11}(0)^2\bigr]=0.
 \label{eq:overview-zero-polynomial}
\end{equation}
For an integer $q\geq1$, define
$\mathcal F_q(\zeta):=(p(\zeta)/2)^{2q}$.
Write $\Delta=S_{11}-S_{11}(0)$.  Since
\eqref{eq:overview-zero-polynomial} gives $p(S_{11}(0))=0$, substituting
$S_{11}=S_{11}(0)+\Delta$ expresses $p(S_{11})$ as a sum of products whose
factors are $S_{11}(0)$ and $\Delta$, with at least one $\Delta$ in every
product.  Expanding $\mathcal F_q(S_{11})$ therefore gives a sum of products
of $S_{11}(0)$ and $\Delta$, each containing at least $2q$ factors $\Delta$.

To bound the sum of these products, Appendix~\ref{sec:private-block-proof}
tracks the time index stored in register $\mathsf T$ after each factor.  The
operator $S_{11}(0)$ preserves this index, while $\Delta$ never decreases it.
Thus the indices in every nonzero term form a nondecreasing sequence.  Bounding
the norm of such a term factor by factor gives $O(\sqrt{\alpha\delta})$ for each
occurrence of $\Delta$, improved to $O(\alpha\delta)$ when $\Delta$ leaves
the index unchanged.  Since every product contains at least $2q$ factors
$\Delta$, this yields at least a factor $(\alpha\delta)^q$ in its norm
bound.  Monotonicity limits the number of possible
sequences.  Moreover, terms that create nonzero Kraus labels at different
positions have orthogonal output spaces, so their contributions add in
squared norm.  Combining the decay from the $\Delta$ factors with these
monotonicity and orthogonality bounds leaves $J$ only through the combination
$J\alpha\delta=\tau$ and gives
\begin{equation}
 \|\mathcal F_q(S_{11})\|
 \leq\left(\frac{C_0\tau}{q}\right)^q,
 \label{eq:overview-factorial}
\end{equation}
for a universal constant $C_0$, provided $q\geq C_0\tau$ and
$J\geq\max\{2\tau,8q\}$.

\Needspace{10\baselineskip}
\paragraph{Combining reuse lengths.}
We now turn this polynomial bound into a choice of reuse coefficients with
bounded LCU normalization.
Multiplying $\mathcal F_q$ by
$g_{12q}(\zeta):=(12q)^{-1}\sum_{\ell=0}^{12q-1}\zeta^\ell$
averages the monomial coefficients
of $\mathcal F_q$ and thereby reduces the LCU
normalization~\cite[Sec.~7]{ChenEtAl2026}.
Section~\ref{sec:lcu} proves that the resulting
polynomial has the expansion
\begin{equation}
 \mathcal Q_q(\zeta)=\mathcal F_q(\zeta)g_{12q}(\zeta)
 =\sum_{N=1}^{20q}\lambda_Ng_N(\zeta),
 \qquad \sum_N\lambda_N=1,
 \qquad
 \sum_N|\lambda_N|=2.
 \label{eq:overview-lcu}
\end{equation}
The identity $\sum_N|\lambda_N|=2$ gives LCU normalization $2$,
yielding a projected unitary encoding of
$\frac12\sum_N\lambda_NP_N$.
Since $S_{11}$ is a contraction, $\|g_{12q}(S_{11})\|\leq1$.
Combining the weighted-error identity \eqref{eq:overview-weighted-error}, the polynomial bound
\eqref{eq:overview-factorial}, and using $\|S_{01}\|\leq1$ and
$\|\Gamma\|\leq\sqrt\tau$ from \eqref{eq:overview-global}, gives
\begin{align}
 \left\|W_J-\sum_N\lambda_NP_N\right\|
 &=\|S_{01}\mathcal F_q(S_{11})g_{12q}(S_{11})\Gamma\| \notag\\
 &\leq\|S_{01}\|\,\|\mathcal F_q(S_{11})\|\,
          \|g_{12q}(S_{11})\|\,\|\Gamma\| \notag\\
 &\leq\sqrt\tau\,\|\mathcal F_q(S_{11})\|
 \leq\sqrt\tau\left(\frac{C_0\tau}{q}\right)^q.
 \label{eq:overview-final-error}
\end{align}
Since $\deg\mathcal Q_q=20q-1$, only reuse lengths $N\leq20q$ occur in
\eqref{eq:overview-lcu}.  A length-$N$ reuse makes $N$ calls to $S$.
The LCU circuit selects $N$ coherently, so $20q$ controlled calls to $S$
suffice.  Since each call to $S$ uses one oracle query, the LCU circuit uses
$20q$ queries.
Oblivious amplitude amplification for isometries~\cite{CleveWang2017}
uses two calls to the LCU circuit and one to its inverse.  Its projected action is the cubic polynomial
$\frac32\widetilde W_J-\frac12\widetilde W_J\widetilde W_J^\dagger
\widetilde W_J$, which is close to $W_J$ whenever
$\widetilde W_J=\sum_N\lambda_NP_N$ is close to $W_J$.  The OAA circuit
therefore uses at most $60q$ queries.  Tracing out its auxiliary registers
gives a channel
$\widetilde{\cE}$ satisfying
\begin{equation}
 \bigl\|\widetilde{\cE}-\Phi_\delta^J\bigr\|_\diamond
 =O\!\left(\sqrt\tau\left(\frac{C_0\tau}{q}\right)^q\right).
 \label{eq:overview-channel-error}
\end{equation}
For $2\tau>\varepsilon$, Lemma~\ref{lem:factorial-inversion} shows that $q$
can be chosen so that the error in \eqref{eq:overview-channel-error} is at
most $\varepsilon/2$ and
\[
 q=O\!\left(
 \tau+
 \frac{\log(1/\varepsilon)}
 {\log\!\bigl(e+\log(1/\varepsilon)/\tau\bigr)}
 \right).
\]
Then take
\[
 J\geq\max\{2\tau,8q,20\tau^2/\varepsilon\}.
\]
The first two lower bounds on $J$ ensure the hypotheses of
\eqref{eq:overview-factorial}, while the last and
\eqref{eq:overview-discretization} give
$\|\Phi_\delta^J-e^{t\cL}\|_\diamond\leq\varepsilon/2$.
The triangle inequality therefore gives
$\|\widetilde{\cE}-e^{t\cL}\|_\diamond\leq\varepsilon$, and the bound of
$60q$ queries gives \eqref{eq:upper-complexity}.  Increasing $J$ affects
the gate complexity of implementing $S$ but not its one-query cost.

\paragraph{Reducing the gate count.}
We finally give a more gate- and qubit-efficient implementation of
$S$:
a direct implementation applies $J$ local updates and uses $J$ registers
for the Kraus labels.
Section~\ref{sec:gate-complexity} factors each update into a
rotation, a Kraus-label exchange, and a permutation with phases within
the private space.  Reordering the factors places the private-space
permutations immediately after the query, where they add only $O(1)$
gates, and separates the label exchanges from the rotations.
The exchanges change at most one nonzero Kraus label per call to $S$
or $S^\dagger$.  Over $O(q)$ calls, it therefore suffices to use a compressed
Kraus-label representation that stores the positions and values of $O(q)$
nonzero labels.  We arrange these pairs in
decreasing order of position, so an exchange inserts or removes the first
pair using controlled shifts of the registers.

For fixed Kraus labels, the rotations act only at step indices after
the last nonzero label.  We partition this range into power-of-two
intervals and use the rotation factorization
of~\cite[Lemma~10]{ChenEtAl2026Gate} on each interval.
This gives a gate count per call with only polylogarithmic dependence
on $J$.
Dividing the evolution time into shorter intervals reduces the polynomial
degree needed on each interval and yields the gate bound in
Theorem~\ref{thm:main} while preserving the total query bound.

\paragraph{Applications.}
Our algorithm also improves the complexity of solving semi-dissipative
linear ODEs.  Shang et al.~\cite{SGAZ25} reduced the preparation of the
normalized solution of a homogeneous semi-dissipative linear ODE to
Lindbladian simulation via a non-diagonal density-matrix encoding, followed
by \(O(\eta_T^{-1})\) rounds of amplitude amplification, where
\(\eta_T=\|\mu(T)\|\) is the norm of the unnormalized ODE solution at time
$T$.  In their square-root access model, replacing the coherent
Lindbladian-simulation subroutine by ours improves the query complexity from
\(O\!\left(\frac{\tau}{\eta_T}
\frac{\log(\tau/(\varepsilon\eta_T))}
{\log(e+\log(\tau/(\varepsilon\eta_T)))}\right)\) to
\(O\!\left(\frac{1}{\eta_T}\left[
\tau+\frac{\log(1/(\varepsilon\eta_T))}
{\log(e+\log(1/(\varepsilon\eta_T))/\tau)}\right]\right)\).
Here \(\tau=\alpha_VT\).  Thus the precision-dependent term no longer
multiplies \(\tau\); in the time-dependent case, this also removes the
squared logarithmic factor in the previous bound.

Chen et al.~\cite{ChenEtAl2025Thermal} prepare Gibbs states by simulating an exact detailed-balance Lindbladian for its mixing time \(t_{\mathrm{mix}}\), with query complexity \(O\!\left(\tau\log(\tau/\varepsilon)/\log\log(\tau/\varepsilon)\right)\), where \(\tau\) is the normalized evolution time. Replacing their simulation subroutine with our optimal algorithm improves this to \(O\!\left(\tau+\frac{\log(1/\varepsilon)}{\log(e+\log(1/\varepsilon)/\tau)}\right)\), without changing the mixing time or the \(\beta\)-dependent cost of constructing the Lindbladian oracles.

Ding et al.~\cite{DingChenLin2024} prepare ground states by evolving an engineered single-jump Lindbladian for its mixing time, so the task reduces to simulating this evolution for a normalized time \(\tau_{\mathrm{mix}}\). Our optimal simulator reduces the corresponding query complexity from \(O\!\left(\tau_{\mathrm{mix}}\log(\tau_{\mathrm{mix}}/\varepsilon)/\log\log(\tau_{\mathrm{mix}}/\varepsilon)\right)\) to \(O\!\left(\tau_{\mathrm{mix}}+\frac{\log(1/\varepsilon)}{\log(e+\log(1/\varepsilon)/\tau_{\mathrm{mix}})}\right)\), without changing the problem-dependent mixing time or the cost of constructing the jump-operator oracle.

\subsection{Related work}

A concurrent manuscript by Kitchen~\cite{Kitchen2026} studies time-dependent
Lindbladian simulation using a causal query-compression framework and reports
the same asymptotic query complexity. After completing this manuscript, we
became aware of concurrent work by Wang and Ye~\cite{WangYe2026}, who
independently obtain the same asymptotic dependence of the query complexity on
normalized time and precision using a closely related transducer construction.
Their second version also achieves gate complexity comparable to ours in
its dependence on evolution time and precision.

Other approaches to Lindbladian simulation use additional structural
assumptions, target different computational tasks, or incur larger
asymptotic query costs.

\paragraph{Structure-dependent simulation.}
Shang, An, and Shao~\cite{ShangAnShao2026} obtained an
additive dependence on time and precision in the query complexity for purely
dissipative Lindbladians with unitary jump operators, and showed exponential
fast-forwarding in circuit depth for block-diagonal Pauli jump operators.
Gao, Ji, and Liu~\cite{GaoJiLiu2026} used Hamiltonian twirling to identify
fast-forwardable families of Lindbladian dynamics and proved matching lower
bounds on the time dependence for the families considered.
Wang et al.~\cite{WangEtAl2026Commutator} derived commutator-based Trotter
error bounds for Lindbladians, improving product-formula bounds for locally
interacting systems. These results use additional structure not assumed in
our general setting.

\paragraph{Hamiltonian dilations and repeated interactions.}
Ding, Li, and Lin~\cite{DingLiLin2024} constructed arbitrarily high-order
approximations to Lindblad evolution using unitary dynamics on an enlarged
Hilbert space, reducing each step to Hamiltonian simulation followed by
tracing out the ancillas.
Pocrnic, Segal, and Wiebe~\cite{PocrnicSegalWiebe2025} approximated
Lindbladian evolution by repeated interactions between the system and
ancillas and analyzed their implementation using Hamiltonian simulation.

\paragraph{Randomized and trajectory-based methods.}
Borras and Marvian~\cite{BorrasMarvian2025} developed a second-order
product-formula algorithm for a class of Lindbladians that simulates the
dissipative part by sampling simple channels, reducing ancilla requirements
and the gate-complexity dependence on the number of jump operators.
Chen et al.~\cite{ChenEtAl2025Randomized} developed a qDRIFT-type
method that randomly samples, at each step, a Lindbladian containing a single
jump operator.
Peng et al.~\cite{PengEtAl2025} introduced a quantum-trajectory-inspired
short-time approximation that reduces the gate-complexity dependence on the
number of jump operators.
David, Sinayskiy, and Petruccione~\cite{DavidEtAl2026} studied randomized
first- and second-order product formulas and a qDRIFT channel implemented by
classical sampling.

\paragraph{Observable estimation.}
Observable-estimation algorithms target expectation values
\(\Tr[O\,e^{t\cL}(\rho)]\) of the state after Lindbladian evolution, for a
specified input state \(\rho\) and observable \(O\).
Kato et al.~\cite{KatoEtAl2026} developed randomized
compilation methods for estimating observables using dissipative processes
with a single jump operator and a small number of ancillas.
Yu et al.~\cite{YuEtAl2025} used linear combinations of superoperators to
compensate Trotter errors and estimate observables with logarithmic dependence
on inverse precision in circuit depth using two ancillas.
Mohammadipour and Li~\cite{MohammadipourLi2025} applied step-size extrapolation
to observable estimates obtained from first-order Lindblad simulation methods,
reducing the maximum circuit depth while retaining an
\(O(1/\varepsilon^2)\) sampling complexity.
These methods obtain such expectation values from repeated circuit executions,
whereas our algorithm approximates the full evolution channel in diamond norm.

\paragraph{Other input models.}
Patel and Wilde~\cite{PatelWilde2023I,PatelWilde2023II} study wave matrix
Lindbladization, in which jump operators are supplied through quantum program
states and the simulation cost is measured by the number of copies of these
states, rather than coherent uses of block-encoding oracles.

\section{Rational CPTP discretization}

\label{sec:rational}

We approximate the evolution over one time step $\delta=t/J$ by a CPTP map,
compose $J$ such steps, and bound the resulting diamond-norm error relative to
$e^{t\cL}$.
Define
\begin{equation}
 K=iH+\frac12B^\dagger B,
 \qquad
 R_\delta=(I+\delta K/2)^{-1}.
 \label{eq:resolvent}
\end{equation}
Define the operators
\begin{equation}
 N_\delta=(I-\delta K/2)R_\delta,
 \qquad
 J_\delta=\sqrt\delta\,BR_\delta.
 \label{eq:rational-blocks}
\end{equation}
Thus \(J_\delta:\cH_{\mathsf S}\to
\cH_{\mathsf E}\otimes\cH_{\mathsf S}\) and
\(\operatorname{ran}J_\delta\subseteq\cH_{\rm L}\otimes\cH_{\mathsf S}\).

The following lemma shows that \(N_\delta\) and \(J_\delta\) define an exact
CPTP step.

\begin{lemma}[Rational CPTP step]
\label{lem:rational-step}
The map
\[
 W_\delta:\cH_{\mathsf S}\longrightarrow
 \cH_{\mathsf E}\otimes\cH_{\mathsf S},
 \qquad
 W_\delta\psi=\ket0_{\mathsf E}\otimes N_\delta\psi+J_\delta\psi
\]
is an isometry.  Consequently
\begin{equation}
 \Phi_\delta(\rho):=\Tr_{\mathsf E}
 (W_\delta\rho W_\delta^\dagger)
 =N_\delta\rho N_\delta^\dagger
 +\Tr_{\mathsf E}(J_\delta\rho J_\delta^\dagger)
 \label{eq:step-channel}
\end{equation}
is CPTP.
\end{lemma}

\begin{proof}
First, \(K+K^\dagger=B^\dagger B\succeq0\) implies, for
\(\varphi\in\cH_{\mathsf S}\),
\[
 \norm\varphi\,\norm{(I+\delta K/2)\varphi}
 \geq\operatorname{Re}\langle\varphi,(I+\delta K/2)\varphi\rangle
 \geq\norm\varphi^2.
\]
For nonzero \(\varphi\), dividing by \(\norm\varphi\) gives
\(\norm{(I+\delta K/2)\varphi}\geq\norm\varphi\).  Hence the smallest
singular value of \(I+\delta K/2\) is at least one, so this operator is
invertible and its inverse satisfies \(\norm{R_\delta}\leq1\).
To verify that \(W_\delta\) is an isometry, compute
\begin{align*}
 N_\delta^\dagger N_\delta+J_\delta^\dagger J_\delta
 &=R_\delta^\dagger
 \left[(I-\delta K^\dagger/2)(I-\delta K/2)
       +\delta B^\dagger B\right]R_\delta\\
 &=R_\delta^\dagger
 (I+\delta K^\dagger/2)(I+\delta K/2)R_\delta=I.
\end{align*}
Since \((\bra0_{\mathsf E}\otimes I_{\mathsf S})J_\delta=0\), the two
terms of \(W_\delta\) have orthogonal ranges.  Hence
\(W_\delta^\dagger W_\delta=N_\delta^\dagger N_\delta+
J_\delta^\dagger J_\delta=I\).
The same orthogonality makes the cross terms vanish under
\(\Tr_{\mathsf E}\), which gives the Kraus representation of
\(\Phi_\delta\) in \eqref{eq:step-channel}.  Tracing the output of an
isometry gives a CPTP map.
\end{proof}

For each step \(j=0,\ldots,J-1\), use a fresh
Kraus-label register \(\mathsf E_j\) to store the Kraus label coherently:
\[
 W_\delta\varphi
 =\ket0_{\mathsf E_j}\otimes N_\delta\varphi
 +\sqrt\delta\sum_{k=1}^m\ket k_{\mathsf E_j}\otimes L_kR_\delta\varphi.
\]
Label \(0\) corresponds to the Kraus operator \(N_\delta\), and
label \(k\geq1\) to \(\sqrt\delta L_kR_\delta\).  For
\(0\leq j\leq J\), set
\[
 \mathcal R_{<j}:=\bigotimes_{r=0}^{j-1}\cH_{\mathsf E_r},
 \qquad \mathcal R_{<0}:=\mathbb C,
 \qquad \cH_J:=\mathcal R_{<J}\otimes\cH_{\mathsf S},
\]
and
\[
 \cH_j:=\mathcal R_{<j}\otimes
 \biggl(\bigotimes_{r=j}^{J-1}\mathbb C\ket0_{\mathsf E_r}\biggr)
 \otimes\cH_{\mathsf S}\subseteq\cH_J.
\]
Thus
\[
 \cH_0\subseteq\cH_1\subseteq\cdots\subseteq\cH_J.
\]
In formulas involving \(\cH_j\), we omit the \(\ket0_{\mathsf E_r}\) factors and write
\(\cH_j\cong\mathcal R_{<j}\otimes\cH_{\mathsf S}\).  Since
\[
 \mathcal R_{<j+1}=\mathcal R_{<j}\otimes\cH_{\mathsf E_j},
\]
the \(j\)-th application of \(W_\delta\) acts as
\[
 I_{\mathcal R_{<j}}\otimes W_\delta:
 \underbrace{\mathcal R_{<j}\otimes\cH_{\mathsf S}}_{\cH_j}
 \longrightarrow
 \underbrace{\mathcal R_{<j}\otimes\cH_{\mathsf E_j}\otimes
 \cH_{\mathsf S}}_{\cH_{j+1}}.
\]
It leaves \(\mathsf E_0,\ldots,\mathsf E_{j-1}\) unchanged and writes
the new label in \(\mathsf E_j\).  For \(\xi\in\mathcal R_{<j}\) and
\(\varphi\in\cH_{\mathsf S}\),
\[
 (I_{\mathcal R_{<j}}\otimes W_\delta)(\xi\otimes\varphi)
 =\xi\otimes\bigl(\ket0_{\mathsf E_j}\otimes N_\delta\varphi
 +J_\delta\varphi\bigr).
\]
Define the product isometries recursively by
\begin{equation}
 W_0:=I_{\cH_{\mathsf S}},\qquad
 W_{j+1}:=(I_{\mathcal R_{<j}}\otimes W_\delta)W_j:
 \cH_{\mathsf S}\longrightarrow\cH_{j+1}.
 \label{eq:product-stinespring}
\end{equation}
We write $\Id_{\mathcal X}$ for the identity channel on a space $\mathcal X$,
omitting the subscript when the space is clear.
We will show by induction that tracing the first \(j\) Kraus-label registers
recovers \(j\) applications of \(\Phi_\delta\):
\begin{equation}
 \Tr_{\mathsf E_0\cdots\mathsf E_{j-1}}
 (W_j\rho W_j^\dagger)=\Phi_\delta^j(\rho).
 \label{eq:product-stinespring-channel}
\end{equation}
For \(j=0\), both sides equal \(\rho\).  Assuming the equality for \(j\),
we substitute the recursive definition of \(W_{j+1}\) in
\eqref{eq:product-stinespring} and trace out \(\mathsf E_j\) using the Kraus
representation \eqref{eq:step-channel}:
\[
\begin{aligned}
 &\Tr_{\mathsf E_0\cdots\mathsf E_j}
 (W_{j+1}\rho W_{j+1}^\dagger)\\
 &=\Tr_{\mathsf E_0\cdots\mathsf E_j}
 \left[(I_{\mathcal R_{<j}}\otimes W_\delta)
 W_j\rho W_j^\dagger
 (I_{\mathcal R_{<j}}\otimes W_\delta^\dagger)\right]\\
 &=\Tr_{\mathsf E_0\cdots\mathsf E_{j-1}}
 \left[(\Id_{\mathcal R_{<j}}\otimes\Phi_\delta)
 (W_j\rho W_j^\dagger)\right]
 =\Phi_\delta\!\left(
 \Tr_{\mathsf E_0\cdots\mathsf E_{j-1}}W_j\rho W_j^\dagger\right)
 =\Phi_\delta^{j+1}(\rho).
\end{aligned}
\]
Here the third line follows by tracing out the fresh register
$\mathsf E_j$, thereby applying $\Phi_\delta$.  The map $\Phi_\delta$ can
then be pulled through the trace over $\mathcal R_{<j}$ because it acts only
on $\mathsf S$.  The final equality uses the induction hypothesis,
completing the induction.  We next bound the difference between
\(\Phi_\delta^J\) and \(e^{t\cL}\).

\begin{proposition}[Discretization error]
\label{prop:discretization}
If \(\alpha\delta\leq1/2\), then
\begin{equation}
 \norm{\Phi_\delta-e^{\delta\cL}}_\diamond
 \leq10(\alpha\delta)^2,
 \qquad
 \norm{\Phi_\delta^J-e^{t\cL}}_\diamond
 \leq\frac{10\tau^2}{J}.
 \label{eq:discretization-error}
\end{equation}
\end{proposition}

\begin{proof}
We first compare $\Phi_\delta$ with its first-order approximation
$\Id+\delta\cL$.  From the definitions of $\cL$ and
$K=iH+\frac12B^\dagger B$, for every $\rho$,
\begin{equation}
 \begin{aligned}
 (\Id+\delta\cL)(\rho)
 &=\rho-i\delta[H,\rho]+\delta\Tr_{\mathsf E}(B\rho B^\dagger)
 -\frac\delta2(B^\dagger B\rho+\rho B^\dagger B)\\
 &=\rho-\delta(K\rho+\rho K^\dagger)
 +\delta\Tr_{\mathsf E}(B\rho B^\dagger).
 \end{aligned}
 \label{eq:first-order-lindblad}
\end{equation}
Comparing the Kraus representation of $\Phi_\delta$ in
\eqref{eq:step-channel} with the first-order formula
\eqref{eq:first-order-lindblad}, we estimate the $N_\delta$ and $J_\delta$
terms separately.

For $N_\delta$, we first use $(I+\delta K/2)R_\delta=I$ to write
\[
 R_\delta-I=R_\delta-(I+\delta K/2)R_\delta
 =-\frac\delta2KR_\delta.
\]
Substituting \(R_\delta-I=-\delta KR_\delta/2\) into the definition of
\(N_\delta\) gives the expansion
\begin{equation}
 \begin{aligned}
 N_\delta
 &=(I-\delta K/2)R_\delta
 =(I+\delta K/2)R_\delta-\delta KR_\delta
 =I-\delta KR_\delta\\
 &=I-\delta K+\delta K(I-R_\delta)
 =I-\delta K+\frac{\delta^2}{2}K^2R_\delta.
 \end{aligned}
 \label{eq:rational-remainders}
\end{equation}
Since $\norm K\leq\alpha$ and $\norm{R_\delta}\leq1$, the formulas above for
$R_\delta-I$ and $N_\delta-(I-\delta K)$ imply
\begin{equation}
 \norm{R_\delta-I}\leq\frac{\alpha\delta}{2},
 \qquad
 \norm{N_\delta-(I-\delta K)}\leq\frac{(\alpha\delta)^2}{2},
 \qquad
 \norm{I-\delta K}\leq1+\alpha\delta.
 \label{eq:rational-remainder-bounds}
\end{equation}
Moreover,
$N_\delta^\dagger N_\delta\preceq
N_\delta^\dagger N_\delta+J_\delta^\dagger J_\delta=I$ by
Lemma~\ref{lem:rational-step}, so $\norm{N_\delta}\leq1$.
To compare $N_\delta(\,\cdot\,)N_\delta^\dagger$ with
$\Id-\delta K(\,\cdot\,)-\delta(\,\cdot\,)K^\dagger$, add and subtract
$(I-\delta K)(\,\cdot\,)(I-\delta K)^\dagger$:
\begin{align}
 &N_\delta(\,\cdot\,)N_\delta^\dagger
 -\Id+\delta K(\,\cdot\,)+\delta(\,\cdot\,)K^\dagger
 \notag\\
 &=N_\delta(\,\cdot\,)N_\delta^\dagger
 -(I-\delta K)(\,\cdot\,)(I-\delta K)^\dagger
 +\delta^2K(\,\cdot\,)K^\dagger \notag\\
 &=\bigl[N_\delta-(I-\delta K)\bigr](\,\cdot\,)N_\delta^\dagger
 +(I-\delta K)(\,\cdot\,)
 \bigl[N_\delta-(I-\delta K)\bigr]^\dagger
 +\delta^2K(\,\cdot\,)K^\dagger.
 \label{eq:no-jump-expansion}
\end{align}
Taking diamond norms in the no-jump expansion
\eqref{eq:no-jump-expansion} and using
\(\norm{X(\,\cdot\,)Y^\dagger}_\diamond=\norm X\norm Y\) and the bounds
\eqref{eq:rational-remainder-bounds}, we obtain
\begin{align}
 &\norm{N_\delta(\,\cdot\,)N_\delta^\dagger
 -\Id+\delta K(\,\cdot\,)+\delta(\,\cdot\,)K^\dagger}_\diamond \notag\\
 &\leq\norm{N_\delta-(I-\delta K)}\norm{N_\delta}
 +\norm{I-\delta K}\norm{N_\delta-(I-\delta K)}
 +\delta^2\norm K^2 \notag\\
 &\leq\frac{(\alpha\delta)^2}{2}
 +\frac{(1+\alpha\delta)(\alpha\delta)^2}{2}
 +(\alpha\delta)^2
 \leq3(\alpha\delta)^2.
 \label{eq:no-jump-error}
\end{align}

For the $J_\delta$ term, we substitute $J_\delta=\sqrt\delta BR_\delta$ and
add and subtract $BR_\delta(\,\cdot\,)B^\dagger$:
\begin{align}
 \Tr_{\mathsf E}\bigl(J_\delta(\,\cdot\,)J_\delta^\dagger\bigr)
 -\delta\Tr_{\mathsf E}\bigl(B(\,\cdot\,)B^\dagger\bigr)
 &=\delta\Tr_{\mathsf E}\Bigl(
 BR_\delta(\,\cdot\,)R_\delta^\dagger B^\dagger
 -B(\,\cdot\,)B^\dagger\Bigr) \notag\\
 &=\delta\Tr_{\mathsf E}\Bigl(
 B(R_\delta-I)(\,\cdot\,)R_\delta^\dagger B^\dagger
 +B(\,\cdot\,)(R_\delta^\dagger-I)B^\dagger\Bigr).
 \label{eq:jump-expansion}
\end{align}
Taking diamond norms in this expansion, using contractivity of the partial
trace and the bounds \eqref{eq:rational-remainder-bounds}, gives
\begin{align}
 \norm{\Tr_{\mathsf E}\bigl(J_\delta(\,\cdot\,)J_\delta^\dagger\bigr)
 -\delta\Tr_{\mathsf E}\bigl(B(\,\cdot\,)B^\dagger\bigr)}_\diamond
 &\leq\delta\norm B^2\bigl(
 \norm{R_\delta-I}\norm{R_\delta}+\norm{R_\delta-I}\bigr) \notag\\
 &=\delta\norm B^2\norm{R_\delta-I}(1+\norm{R_\delta})
 \leq(\alpha\delta)^2.
 \label{eq:jump-error}
\end{align}
Substituting the Kraus representation \eqref{eq:step-channel} and the
first-order formula \eqref{eq:first-order-lindblad} into
\(\Phi_\delta-(\Id+\delta\cL)\) separates the difference into the no-jump
and jump terms estimated above:
\[
 \begin{aligned}
 \Phi_\delta-(\Id+\delta\cL)={}&N_\delta(\,\cdot\,)N_\delta^\dagger
 -\Id+\delta K(\,\cdot\,)+\delta(\,\cdot\,)K^\dagger+\Tr_{\mathsf E}\bigl(J_\delta(\,\cdot\,)J_\delta^\dagger\bigr)
 -\delta\Tr_{\mathsf E}\bigl(B(\,\cdot\,)B^\dagger\bigr).
 \end{aligned}
\]
Combining the no-jump bound \eqref{eq:no-jump-error} and the jump bound
\eqref{eq:jump-error} with the triangle inequality gives
\begin{equation}
 \norm{\Phi_\delta-(\Id+\delta\cL)}_\diamond
 \leq3(\alpha\delta)^2+(\alpha\delta)^2=4(\alpha\delta)^2.
 \label{eq:rational-first-order-error}
\end{equation}

It remains to compare $\Id+\delta\cL$ with $e^{\delta\cL}$.  Applying the
triangle inequality to \eqref{eq:lindbladian}, we have
\begin{align}
 \norm{\cL}_\diamond
 &\leq\norm{-i[H,\,\cdot\,]}_\diamond
 +\norm{\Tr_{\mathsf E}\bigl(B(\,\cdot\,)B^\dagger\bigr)}_\diamond
 +\frac12\norm{B^\dagger B(\,\cdot\,)
 +(\,\cdot\,)B^\dagger B}_\diamond \notag\\
 &\leq2\norm H+\norm B^2
 +\frac12\bigl(\norm{B^\dagger B}+\norm{B^\dagger B}\bigr) \notag\\
 &=2\norm H+2\norm B^2\leq2\alpha.
 \label{eq:lindbladian-diamond-norm}
\end{align}
The generator-norm bound \eqref{eq:lindbladian-diamond-norm} gives
\begin{equation}
 \begin{aligned}
 \norm{e^{\delta\cL}-(\Id+\delta\cL)}_\diamond
 &\leq\sum_{n=2}^\infty\frac{\delta^n\norm{\cL}_\diamond^n}{n!}\leq\frac12e^{\delta\norm{\cL}_\diamond}
 (\delta\norm{\cL}_\diamond)^2
 \leq2e(\alpha\delta)^2.
 \end{aligned}
 \label{eq:exponential-first-order-error}
\end{equation}
Combining the two preceding first-order error bounds,
\eqref{eq:rational-first-order-error} and
\eqref{eq:exponential-first-order-error}, with the triangle inequality gives
\begin{equation}
 \begin{aligned}
 \norm{\Phi_\delta-e^{\delta\cL}}_\diamond
 &\leq\norm{\Phi_\delta-(\Id+\delta\cL)}_\diamond
 +\norm{e^{\delta\cL}-(\Id+\delta\cL)}_\diamond\\
 &\leq(4+2e)(\alpha\delta)^2<10(\alpha\delta)^2.
 \end{aligned}
 \label{eq:one-step-discretization}
\end{equation}

Finally, we telescope the one-step error over \(J\) steps.  Since
\(J\delta=t\),
\begin{equation}
 \Phi_\delta^J-e^{t\cL}
 =\sum_{r=0}^{J-1}\Phi_\delta^{J-1-r}
 (\Phi_\delta-e^{\delta\cL})e^{r\delta\cL}.
 \label{eq:discretization-telescoping}
\end{equation}
Both $\Phi_\delta$ and $e^{s\cL}$ are CPTP and therefore have diamond norm
one.  Taking diamond norms in the telescoping identity
\eqref{eq:discretization-telescoping} and applying the one-step bound
\eqref{eq:one-step-discretization} gives
\begin{align*}
 \norm{\Phi_\delta^J-e^{t\cL}}_\diamond
 &\leq\sum_{r=0}^{J-1}
 \norm{\Phi_\delta}_\diamond^{J-1-r}
 \norm{\Phi_\delta-e^{\delta\cL}}_\diamond
 \norm{e^{r\delta\cL}}_\diamond\\
 &=J\norm{\Phi_\delta-e^{\delta\cL}}_\diamond
 \leq J\,10(\alpha\delta)^2
 =\frac{10(\alpha t)^2}{J}
 =\frac{10\tau^2}{J}.
\end{align*}
\end{proof}

\section{A one-query transducer for a rational step}
\label{sec:local-step}

To implement \(W_\delta\) with one query, we construct a transducer
\(S_\delta\) with public space
\(\cH_{\mathsf E}\otimes\cH_{\mathsf S}\) and private space
\[
 \cH_{\rm priv}=
 (\cH_{\mathsf A_H}\otimes\cH_{\mathsf S})
 \oplus
 (\cH_{\mathsf A_B}\otimes\cH_{\mathsf E}
 \otimes\cH_{\mathsf S})^{\oplus2}.
\]
On the three summands of \(\cH_{\rm priv}\),
\(\Omega=U_H\oplus U_B\oplus U_B^\dagger\) acts as \(U_H\), \(U_B\), and
\(U_B^\dagger\), respectively.  The transducer \(S_\delta\) first applies
\(\Omega\) to the private space and then a unitary \(G\) to the public and
private spaces together.

To write the target transformation in terms of the normalized operators, set
\[
 \bar H=H/\alpha_H,
 \qquad
 \bar B=B/\alpha_B.
\]
For a step of length \(\delta\), set
\begin{equation}
 \kappa=\sqrt{\delta\alpha_H/2},\qquad
 \beta=\sqrt\delta\,\alpha_B/2,\qquad
 \mu=\kappa^2+\beta^2,\qquad
 2\kappa^2+4\beta^2=\alpha\delta=\frac{\tau}{J}.
 \label{eq:local-scalars}
\end{equation}
For \(\psi\in\cH_{\mathsf S}\), let \(y=R_\delta\psi\).  Then
the definitions of \(R_\delta\), \(N_\delta\), and \(J_\delta\) in
\eqref{eq:resolvent} and \eqref{eq:rational-blocks} become
\begin{align}
 \psi
 &=\bigl(I+i\kappa^2\bar H
   +\beta^2\bar B^\dagger\bar B\bigr)y,
 \nonumber\\
 N_\delta\psi
 &=\bigl(I-i\kappa^2\bar H
   -\beta^2\bar B^\dagger\bar B\bigr)y,
 \nonumber\\
 J_\delta\psi&=2\beta\bar B y.
 \label{eq:local-targets}
\end{align}

\Needspace{12\baselineskip}
Using \eqref{eq:local-targets}, the next lemma constructs a one-query
transducer for \(W_\delta\) and bounds the norm of its catalyst.

\begin{lemma}[One-query transducer for a rational step]
\label{lem:local-step}
There exist a unitary \(G\) on
\((\cH_{\mathsf E}\otimes\cH_{\mathsf S})\oplus\cH_{\rm priv}\), depending
only on \(\kappa,\beta\), and a linear map
\(\gamma:\cH_{\mathsf S}\to\cH_{\rm priv}\), such that for every
\(\psi\in\cH_{\mathsf S}\),
\begin{equation}
 G\bigl((\ket0_{\mathsf E}\otimes\psi)
 \oplus\Omega\gamma\psi\bigr)
 =W_\delta\psi\oplus\gamma\psi.
 \label{eq:local-transducer}
\end{equation}
Moreover,
\begin{equation}
 \gamma^\dagger\gamma\preceq\alpha\delta I.
 \label{eq:local-catalyst-bound}
\end{equation}
\end{lemma}

\begin{proof}
We write vectors in
\((\cH_{\mathsf E}\otimes\cH_{\mathsf S})\oplus\cH_{\rm priv}\)
as five-component vectors, whose components belong, in order, to
\(\ket0_{\mathsf E}\otimes\cH_{\mathsf S}\),
\(\cH_{\rm L}\otimes\cH_{\mathsf S}\),
\(\cH_{\mathsf A_H}\otimes\cH_{\mathsf S}\), and two copies of
\(\cH_{\mathsf A_B}\otimes\cH_{\mathsf E}\otimes\cH_{\mathsf S}\).
The first two entries form the public space, and the last three form the
private space.  With \(y=R_\delta\psi\), define
\begin{equation}
 \begin{aligned}
 \gamma\psi
 &=
 \begin{pmatrix}
  \kappa(I+iU_H)(\ket0_{\mathsf A_H}\otimes y)\\
  \beta\bigl[(\ket0_{\mathsf A_B}\ket0_{\mathsf E}\otimes y)
   +U_B^\dagger(\ket0_{\mathsf A_B}\otimes\bar B y)\bigr]\\
  \beta\bigl[I-(\ket0\!\bra0_{\mathsf A_B}\otimes
   P_{\rm L})\bigr]
   U_B(\ket0_{\mathsf A_B}\ket0_{\mathsf E}\otimes y)
 \end{pmatrix}
 \\[0.3em]
 &=
 \underbrace{\begin{pmatrix}
  \kappa(\ket0_{\mathsf A_H}\otimes(I+i\bar H)y)\\
  \beta(\ket0_{\mathsf A_B}\ket0_{\mathsf E}\otimes
   (I+\bar B^\dagger\bar B)y)\\
  0
 \end{pmatrix}}_{\text{\(\ket0\)-components}}
 +
 \underbrace{\begin{pmatrix}
  i\kappa\bigl[I-(\ket0\!\bra0_{\mathsf A_H}\otimes I_{\mathsf S})\bigr]
   U_H(\ket0_{\mathsf A_H}\otimes I_{\mathsf S})y\\
  \beta\bigl[I-(\ket0\!\bra0_{\mathsf A_B}\otimes
   \ket0\!\bra0_{\mathsf E}\otimes I_{\mathsf S})\bigr]
   U_B^\dagger(\ket0_{\mathsf A_B}\otimes\bar B y)\\
  \beta\bigl[I-(\ket0\!\bra0_{\mathsf A_B}\otimes
   P_{\rm L})\bigr]
   U_B(\ket0_{\mathsf A_B}\ket0_{\mathsf E}\otimes I_{\mathsf S})y
 \end{pmatrix}}_{\text{orthogonal components}}.
 \end{aligned}
 \label{eq:local-catalyst}
\end{equation}
We now compute the state after the
oracle query.
Applying \(U_H\), \(U_B\), and \(U_B^\dagger\) to the three entries of
\eqref{eq:local-catalyst} and simplifying with \(U_B^\dagger U_B=I\), the block
encoding of \(\bar H\) in \eqref{eq:H-encoding}, and the projected unitary
encoding of \(\bar B\) in \eqref{eq:B-encoding}, we obtain
\begin{align*}
&\kappa U_H(I+iU_H)(\ket0_{\mathsf A_H}\otimes I_{\mathsf S})y\\
 &=
 \kappa\bigl[U_H(\ket0_{\mathsf A_H}\otimes I_{\mathsf S})
 +i(\ket0_{\mathsf A_H}\otimes I_{\mathsf S})\bigr]y
\\
&=
 \kappa(\ket0_{\mathsf A_H}\otimes(\bar H+iI))y
 +\kappa\bigl[I-(\ket0\!\bra0_{\mathsf A_H}\otimes I_{\mathsf S})\bigr]
 U_H(\ket0_{\mathsf A_H}\otimes I_{\mathsf S})y.
\end{align*}
Similarly,
\begin{align*}
&\beta U_B\bigl[
 (\ket0_{\mathsf A_B}\ket0_{\mathsf E}\otimes I_{\mathsf S})
 +U_B^\dagger(\ket0_{\mathsf A_B}\otimes\bar B)
 \bigr]y
\\
&=
 \beta\bigl[U_B(\ket0_{\mathsf A_B}\ket0_{\mathsf E}\otimes I_{\mathsf S})
 +(\ket0_{\mathsf A_B}\otimes\bar B)\bigr]y
\\
&=
 2\beta(\ket0_{\mathsf A_B}\otimes\bar B)y
 +\beta\bigl[I-(\ket0\!\bra0_{\mathsf A_B}\otimes
 P_{\rm L})\bigr]
 U_B(\ket0_{\mathsf A_B}\ket0_{\mathsf E}\otimes I_{\mathsf S})y.
\end{align*}
For the third entry,
\begin{align*}
&\beta U_B^\dagger
 \bigl[I-(\ket0\!\bra0_{\mathsf A_B}\otimes P_{\rm L})\bigr]
 U_B(\ket0_{\mathsf A_B}\ket0_{\mathsf E}\otimes I_{\mathsf S})y
\\
&=
 \beta\bigl[(\ket0_{\mathsf A_B}\ket0_{\mathsf E}\otimes I_{\mathsf S})
 -U_B^\dagger(\ket0_{\mathsf A_B}\otimes\bar B)\bigr]y
\\
&=
 \beta(\ket0_{\mathsf A_B}\ket0_{\mathsf E}\otimes
 (I-\bar B^\dagger\bar B)y)-
 \beta\bigl[I-(\ket0\!\bra0_{\mathsf A_B}\otimes
 \ket0\!\bra0_{\mathsf E}\otimes I_{\mathsf S})\bigr]
 U_B^\dagger(\ket0_{\mathsf A_B}\otimes\bar B)y.
\end{align*}
Thus the state after applying \(\Omega\) is
\begin{equation}
 \begin{aligned}
 &(\ket0_{\mathsf E}\otimes\psi)\oplus\Omega\gamma\psi\\
 &=
 \underbrace{\begin{pmatrix}
  \ket0_{\mathsf E}\otimes\psi\\
  0\\
  \kappa(\ket0_{\mathsf A_H}\otimes(\bar H+iI)y)\\
  2\beta(\ket0_{\mathsf A_B}\otimes\bar B y)\\
  \beta(\ket0_{\mathsf A_B}\ket0_{\mathsf E}\otimes
   (I-\bar B^\dagger\bar B)y)
 \end{pmatrix}}_{=:v_0}+
 \underbrace{\begin{pmatrix}
  0\\
  0\\
  \kappa\bigl[I-(\ket0\!\bra0_{\mathsf A_H}\otimes I_{\mathsf S})\bigr]
   U_H(\ket0_{\mathsf A_H}\otimes I_{\mathsf S})y\\
  \beta\bigl[I-(\ket0\!\bra0_{\mathsf A_B}\otimes
   P_{\rm L})\bigr]
   U_B(\ket0_{\mathsf A_B}\ket0_{\mathsf E}\otimes I_{\mathsf S})y\\
  -\beta\bigl[I-(\ket0\!\bra0_{\mathsf A_B}\otimes
   \ket0\!\bra0_{\mathsf E}\otimes I_{\mathsf S})\bigr]
   U_B^\dagger(\ket0_{\mathsf A_B}\otimes\bar B y)
 \end{pmatrix}}_{=:v_\perp}.
 \end{aligned}
 \label{eq:post-query-vector}
\end{equation}
We first construct the part of \(G\) that produces the no-jump output and
restores the zero-ancilla components of the catalyst.  We combine the first,
third, and fifth entries of \(v_0\) in the queried state
\eqref{eq:post-query-vector} to produce
\(\ket0_{\mathsf E}\otimes N_\delta\psi\) in the public space and the
zero-ancilla components of the first two private entries of \(\gamma\psi\).
In the input and target triples, the
fixed ancilla factors are \(\ket0_{\mathsf E}\),
\(\ket0_{\mathsf A_H}\), and
\(\ket0_{\mathsf A_B}\ket0_{\mathsf E}\), in that order.  Removing these
factors from both triples identifies them with columns in
\(\cH_{\mathsf S}^{\oplus3}\).  The second component of the input column
comes from the third entry of \(v_0\).  We multiply this component by
\(-i\); after this phase, the input column \(x\) and target column \(x'\) are
\begin{equation}
 \begin{aligned}
 x&=
 \begin{pmatrix}
  \psi\\
  \kappa(I-i\bar H)y\\
  \beta(I-\bar B^\dagger\bar B)y
 \end{pmatrix},
 &
 x'&=
 \begin{pmatrix}
  N_\delta\psi\\
  \kappa(I+i\bar H)y\\
  \beta(I+\bar B^\dagger\bar B)y
 \end{pmatrix}.
 \end{aligned}
 \label{eq:reflection-columns}
\end{equation}
For \(u=(1,\kappa,\beta)^\top\), the formulas for \(\psi\) and
\(N_\delta\psi\) in \eqref{eq:local-targets} give
\begin{equation}
 x+x'=
 \begin{pmatrix}\psi+N_\delta\psi\\2\kappa y\\2\beta y\end{pmatrix}
 =\begin{pmatrix}2y\\2\kappa y\\2\beta y\end{pmatrix}
 =2u\otimes y.
 \label{eq:reflection-sum}
\end{equation}
Substituting the formula for \(\psi\) from \eqref{eq:local-targets} gives
\begin{equation}
 \begin{aligned}
 (u^\dagger\otimes I_{\mathsf S})x
 &=\psi+\kappa^2(I-i\bar H)y
 +\beta^2(I-\bar B^\dagger\bar B)y\\
 &=(1+\kappa^2+\beta^2)y=(u^\dagger u)y.
 \end{aligned}
 \label{eq:reflection-conditions}
\end{equation}
Let \(P_u=(uu^\dagger/u^\dagger u)\otimes I_{\mathsf S}\).  Equations
\eqref{eq:reflection-sum} and \eqref{eq:reflection-conditions} imply
\begin{equation}
 P_u x=u\otimes y=\frac{x+x'}2,
 \qquad
 x'=2P_u x-x.
 \label{eq:reflection-action}
\end{equation}
Thus \(2P_u-I\) maps \(x\) to \(x'\).  Composing \(2P_u-I\) with
\(\operatorname{diag}(1,-i,1)\) gives the coefficient matrix \(M\):
\begin{equation}
 \begin{aligned}
 M
 &=\left[
  \frac{2}{1+\mu}
  \begin{pmatrix}1\\\kappa\\\beta\end{pmatrix}
  \begin{pmatrix}1&\kappa&\beta\end{pmatrix}-I_3
  \right]
  \begin{pmatrix}1&0&0\\0&-i&0\\0&0&1\end{pmatrix}
 \\
 &=\frac1{1+\mu}
  \begin{pmatrix}
   1-\mu&-2i\kappa&2\beta\\
   2\kappa&i(1-\kappa^2+\beta^2)&2\kappa\beta\\
   2\beta&-2i\kappa\beta&-1-\kappa^2+\beta^2
  \end{pmatrix}.
 \end{aligned}
 \label{eq:reflection-matrix}
\end{equation}
Since \(uu^\dagger/(u^\dagger u)\) is an orthogonal projection, both factors
in the definition of \(M\) in \eqref{eq:reflection-matrix} are unitary, and hence
\(M^\dagger M=I_3\).
Applying \(M\) to the three system vectors obtained from the first, third,
and fifth entries of \(v_0\), we obtain
\begin{equation}
 \begin{aligned}
 (M\otimes I_{\mathsf S})
 \begin{pmatrix}
  \psi\\
  \kappa(\bar H+iI)y\\
  \beta(I-\bar B^\dagger \bar B)y
 \end{pmatrix}
 &=(2P_u-I)
 \begin{pmatrix}
  \psi\\
  -i\kappa(\bar H+iI)y\\
  \beta(I-\bar B^\dagger\bar B)y
 \end{pmatrix}\\
 &=(2P_u-I)x=x'.
 \end{aligned}
 \label{eq:reflection-matrix-action}
\end{equation}

We next produce the jump output.  By Lemma~\ref{lem:rational-step}, the
second public component of \(W_\delta\psi\) is \(J_\delta\psi\).  Projecting
the fourth entry of the queried state \eqref{eq:post-query-vector} onto
\(\ket0_{\mathsf A_B}\otimes\cH_{\rm L}\otimes\cH_{\mathsf S}\) gives
\begin{equation}
 \begin{aligned}
 &(\bra0_{\mathsf A_B}\otimes P_{\rm L})
 \left\{2\beta(\ket0_{\mathsf A_B}\otimes\bar B y)
 +\beta[I-(\proj0_{\mathsf A_B}\otimes P_{\rm L})]
 U_B(\ket0_{\mathsf A_B}\ket0_{\mathsf E}\otimes I_{\mathsf S})y\right\}\\
 &=2\beta\bar B y=J_\delta\psi.
 \end{aligned}
 \label{eq:local-jump-output}
\end{equation}

We now define \(G\) on all five summands and verify its unitarity.
Let \(s\in\cH_{\mathsf S}\),
\(\ell\in\cH_{\rm L}\otimes\cH_{\mathsf S}\),
\(h\in\cH_{\mathsf A_H}\otimes\cH_{\mathsf S}\), and
\(z,b\in\cH_{\mathsf A_B}\otimes\cH_{\mathsf E}\otimes\cH_{\mathsf S}\).
The matrix \(M\) acts on \(s\), the \(\ket0_{\mathsf A_H}\)-component of
\(h\), and the
\(\ket0_{\mathsf A_B}\ket0_{\mathsf E}\)-component of \(b\).  Write the
corresponding system vectors and orthogonal remainders as
\[
 \begin{aligned}
 h_0&=(\bra0_{\mathsf A_H}\otimes I_{\mathsf S})h,
 &h^\perp&=h-(\ket0_{\mathsf A_H}\otimes I_{\mathsf S})h_0,\\
 b_0&=(\bra0_{\mathsf A_B}\bra0_{\mathsf E}\otimes I_{\mathsf S})b,
 &b^\perp&=b-(\ket0_{\mathsf A_B}\ket0_{\mathsf E}\otimes I_{\mathsf S})b_0,
 \end{aligned}
\]
and let
\[
 \begin{pmatrix}s'\\h_0'\\b_0'\end{pmatrix}
 =(M\otimes I_{\mathsf S})
 \begin{pmatrix}s\\h_0\\b_0\end{pmatrix}.
\]
Define \(G\) on the five summands by
\begin{equation}
 G\begin{pmatrix}\ket0_{\mathsf E}\otimes s\\\ell\\h\\z\\b\end{pmatrix}
 =\begin{pmatrix}
  \ket0_{\mathsf E}\otimes s'\\
  (\bra0_{\mathsf A_B}\otimes P_{\rm L})z\\
  (\ket0_{\mathsf A_H}\otimes I_{\mathsf S})h_0'+ih^\perp\\
  (\ket0_{\mathsf A_B}\ket0_{\mathsf E}\otimes I_{\mathsf S})b_0'-b^\perp\\
  \ket0_{\mathsf A_B}\otimes\ell
  +[I-(\ket0\!\bra0_{\mathsf A_B}\otimes P_{\rm L})]z
 \end{pmatrix}.
 \label{eq:G-full-action}
\end{equation}
To prove that \(G\) is unitary, we decompose its domain and codomain into
mutually orthogonal subspaces and verify that each restriction in
\eqref{eq:G-full-action} is a unitary map between the corresponding
subspaces.  The definitions of \(h^\perp\) and
\(b^\perp\) give
\[
 (\bra0_{\mathsf A_H}\otimes I_{\mathsf S})h^\perp=0,
 \qquad
 (\bra0_{\mathsf A_B}\bra0_{\mathsf E}\otimes I_{\mathsf S})b^\perp=0.
\]
Thus \((\ket0_{\mathsf A_H}\otimes I_{\mathsf S})h_0\) is orthogonal to
\(h^\perp\), and
\((\ket0_{\mathsf A_B}\ket0_{\mathsf E}\otimes I_{\mathsf S})b_0\) is
orthogonal to \(b^\perp\).  Equation~\eqref{eq:G-full-action} maps
\((s,h_0,b_0)\) by \(M\otimes I_{\mathsf S}\), maps \(h^\perp\) to the
corresponding orthogonal subspace of the third output summand with phase
\(i\), and maps \(b^\perp\) to the corresponding orthogonal subspace of the
fourth output summand with phase \(-1\).  These restrictions are unitary.

It remains to check the action of \(G\) on \(\ell\) and \(z\).  The
decomposition
\[
 z=\ket0_{\mathsf A_B}\otimes
 (\bra0_{\mathsf A_B}\otimes P_{\rm L})z
 +[I-(\proj0_{\mathsf A_B}\otimes P_{\rm L})]z
\]
is orthogonal because \(\proj0_{\mathsf A_B}\otimes P_{\rm L}\) and its
complement are orthogonal projections.  The second and fifth entries in
\eqref{eq:G-full-action} send the
\(\ket0_{\mathsf A_B}\otimes\cH_{\rm L}\otimes\cH_{\mathsf S}\) subspace of
the fourth input summand to the second output summand, the second input
summand to the corresponding subspace of the fifth output summand, and the
remaining subspace of the fourth input summand to its orthogonal complement
in the fifth output summand.  These three restrictions are isometries with
mutually orthogonal ranges spanning the second and fifth output summands.
Hence this part of \(G\) is unitary.

Together with the unitary actions on \((s,h_0,b_0)\), \(h^\perp\), and
\(b^\perp\) described above, this proves that \(G\) is unitary.

Figure~\ref{fig:local-catalyst-action} summarizes the query and the action of
\(G\).

\begin{figure}[H]
 \centering
 \begin{tikzpicture}[
  x=1cm,y=1cm,font=\footnotesize,
  space/.style={draw=black!35,fill=white,rounded corners=2pt,line width=0.6pt},
  amplitude/.style={inner sep=0pt,minimum width=4.2cm,align=center},
  wire/.style={-{Latex[length=1.6mm]},line width=0.8pt,rounded corners=2pt,
   preaction={draw=white,line width=2.6pt}},
  mix wire/.style={wire,draw=blue!60!black},
  jump wire/.style={wire,draw=teal!65!black},
  other wire/.style={wire,draw=black!65}
 ]
 \node[font=\small] at (2.25,1.4)
  {$\bigl(\ket0_{\mathsf E}\otimes\psi\bigr)\oplus\gamma\psi$};
 \node[font=\small] at (7.55,1.4)
  {$\bigl(\ket0_{\mathsf E}\otimes\psi\bigr)\oplus\Omega\gamma\psi$};
 \node[font=\small] at (14.05,1.4)
 {$W_\delta\psi\oplus\gamma\psi$};
 \draw[-{Latex[length=1.6mm]},line width=.65pt]
  (4.62,1.4)--(5.18,1.4) node[midway,above=2pt] {$I\oplus\Omega$};
 \draw[-{Latex[length=1.6mm]},line width=.65pt]
  (9.95,1.4)--(11.65,1.4) node[midway,above=2pt] {$G$};
 % Space labels on the left apply to all three aligned columns.
 \foreach \x in {0,5.3,11.8} {
  \draw[space] (\x,0.8) rectangle (\x+4.5,-1.1);
  \draw[draw=black!25,fill=black!2,rounded corners=3pt]
   (\x,-1.55) rectangle (\x+4.5,-8.95);
  \foreach \top/\bottom in {-2.3/-4.15,-4.6/-6.45,-6.9/-8.75} {
   \draw[space] (\x+0.08,\top) rectangle (\x+4.42,\bottom);
  }
  \foreach \y in {-0.2,-3.35,-5.65,-7.95} {
   \node[text=black!50] at (\x+2.25,\y) {$+$};
  }
 }
 \node[fill=white,inner xsep=5pt,inner ysep=0pt] at (2.25,0.8)
  {$\cH_{\mathsf E}\otimes\cH_{\mathsf S}$};
 \node[fill=white,inner xsep=5pt,inner ysep=0pt] at (2.25,-2.3)
  {$\cH_{\mathsf A_H}\otimes\cH_{\mathsf S}$};
 \foreach \y in {-4.6,-6.9} {
  \node[fill=white,inner xsep=5pt,inner ysep=0pt] at (2.25,\y)
   {$\cH_{\mathsf A_B}\otimes\cH_{\mathsf E}\otimes\cH_{\mathsf S}$};
 }
 \node at (2.25,-1.85) {$\gamma\psi\in\cH_{\rm priv}$};
 \node at (7.55,-1.85) {$\Omega\gamma\psi$};
 \node at (14.05,-1.85) {$\gamma\psi$};
 % The initial and final catalyst columns are identical term by term.
 \foreach \x in {2.25,14.05} {
  \node[amplitude,text=blue!60!black] at (\x,-2.95)
   {$\kappa\ket0_{\mathsf A_H}\otimes(I+i\bar H)y$};
  \node[amplitude] at (\x,-3.75) {$ih^\perp$};
  \node[amplitude,text=blue!60!black] at (\x,-5.25)
   {$\beta\ket0_{\mathsf A_B}\ket0_{\mathsf E}\otimes
    (I+\bar B^\dagger\bar B)y$};
  \node[amplitude] at (\x,-6.05) {$-b^\perp$};
  \node[amplitude,text=black!45] at (\x,-7.55) {$0$};
  \node[amplitude] at (\x,-8.35)
   {$z-\ket0_{\mathsf A_B}\otimes J_\delta\psi$};
 }
 \node[amplitude,text=blue!60!black] at (2.25,0.2)
  {$\ket0_{\mathsf E}\otimes\psi$};
 \foreach \x in {2.25,7.55} {
  \node[amplitude,text=black!45] at (\x,-0.6) {$0$};
 }
 % Query arrows act on whole boxes, not on their individual terms.
 \foreach \y/\op in {-0.2/I,-3.35/U_H,-5.65/U_B,-7.95/U_B^\dagger} {
  \draw[other wire] (4.5,\y)--(5.3,\y)
   node[midway,above=2pt,inner sep=0pt] {$\op$};
 }
 \node[amplitude,text=blue!60!black] (s) at (7.55,0.2)
  {$\ket0_{\mathsf E}\otimes\psi$};
 \node[amplitude,text=blue!60!black] (sp) at (14.05,0.2)
  {$\ket0_{\mathsf E}\otimes N_\delta\psi$};
 \node[amplitude,text=teal!65!black] (jp) at (14.05,-0.6)
  {$J_\delta\psi$};
 \node[amplitude,text=blue!60!black] (h) at (7.55,-2.95)
  {$\kappa\ket0_{\mathsf A_H}\otimes(\bar H+iI)y$};
 \node[amplitude] (ho) at (7.55,-3.75) {$h^\perp$};
 \node[amplitude,text=teal!65!black] (zj) at (7.55,-5.25)
  {$2\beta\ket0_{\mathsf A_B}\otimes\bar B y$};
 \node[amplitude] (zo) at (7.55,-6.05)
  {$z-\ket0_{\mathsf A_B}\otimes J_\delta\psi$};
 \node[amplitude,text=blue!60!black] (b) at (7.55,-7.55)
  {$\beta\ket0_{\mathsf A_B}\ket0_{\mathsf E}\otimes
   (I-\bar B^\dagger\bar B)y$};
 \node[amplitude] (bo) at (7.55,-8.35) {$b^\perp$};
 % One M box mixes three inputs and delivers three outputs.
 \draw[mix wire] (s.east)--(10.5,0.2);
 \draw[mix wire] (h.east)--(10.15,-2.95)--(10.5,-1.3);
 \draw[mix wire] (b.east)--(10,-7.55)--(10,-3.35)--(10.5,-2.8);
 \draw[mix wire] (11.1,0.2)--(sp.west);
 \draw[mix wire] (11.1,-1.3)--(11.4,-2.95)--(11.95,-2.95);
 \draw[mix wire] (11.1,-2.8)--(11.55,-3.2)--(11.55,-5.25)--(11.95,-5.25);
 \draw[draw=blue!60!black,fill=blue!7,rounded corners=3pt,line width=0.9pt]
  (10.5,0.5) rectangle (11.1,-3.1);
 \node[text=blue!60!black,font=\small] at (10.8,-1.3) {$M$};
 % White under-strokes distinguish crossings from connections.
 \draw[jump wire] (zj.east)--(10.3,-5.25)--(10.3,-4.3)
  --(11.3,-4.3)--(11.3,-0.6)--(jp.west);
 \draw[other wire] (ho.east)--(11.95,-3.75)
  node[pos=0.5,above=2pt,fill=white,inner sep=1pt] {$i$};
 \draw[other wire] (zo.east)--(10.3,-6.05)--(11.3,-8.35)--(11.95,-8.35);
 \draw[other wire] (bo.east)--(10.3,-8.35)--(11.3,-6.05)--(11.95,-6.05)
  node[pos=0.85,above=2pt,fill=white,inner sep=1pt] {$-1$};
 \end{tikzpicture}
 \caption{One query followed by \(G\) implements \(W_\delta\) and restores
 \(\gamma\psi\).  Here \(h,z,b\) are the three components of the private
 input to \(G\).
 Blue connections show \(M\otimes I_{\mathsf S}\) acting as in
 \eqref{eq:reflection-matrix-action}; the other connections implement the
 projections, swaps, and phases in the explicit action of \(G\) in
 \eqref{eq:G-full-action}.}
 \label{fig:local-catalyst-action}
\end{figure}
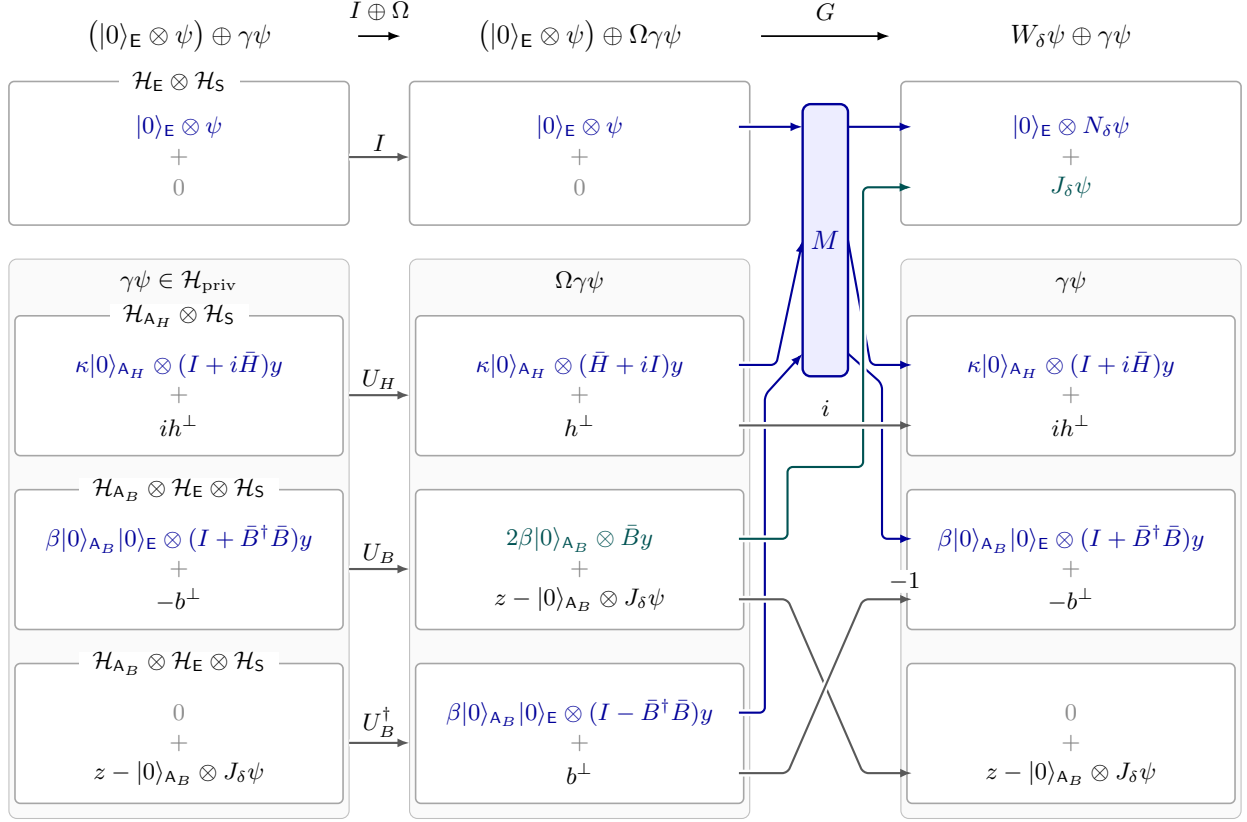

To verify the transducer identity \eqref{eq:local-transducer}, apply the
explicit action of \(G\) in \eqref{eq:G-full-action} to the queried state
\eqref{eq:post-query-vector}, taking
\begin{equation}
\begin{aligned}
 s&=\psi,\qquad \ell=0,\\
 h&=\kappa(\ket0_{\mathsf A_H}\otimes(\bar H+iI)y)
   +\kappa\bigl[I-(\ket0\!\bra0_{\mathsf A_H}\otimes I_{\mathsf S})\bigr]
    U_H(\ket0_{\mathsf A_H}\otimes I_{\mathsf S})y,\\
 z&=2\beta(\ket0_{\mathsf A_B}\otimes\bar B y)
   +\beta\bigl[I-(\ket0\!\bra0_{\mathsf A_B}\otimes
      P_{\rm L})\bigr]
      U_B(\ket0_{\mathsf A_B}\ket0_{\mathsf E}\otimes I_{\mathsf S})y,\\
 b&=\beta(\ket0_{\mathsf A_B}\ket0_{\mathsf E}\otimes
      (I-\bar B^\dagger\bar B)y)
   -\beta\bigl[I-(\ket0\!\bra0_{\mathsf A_B}\otimes
      \ket0\!\bra0_{\mathsf E}\otimes I_{\mathsf S})\bigr]
      U_B^\dagger(\ket0_{\mathsf A_B}\otimes\bar B y).
\end{aligned}
\label{eq:post-query-entries}
\end{equation}
Applying \(M\) as in \eqref{eq:reflection-matrix-action} gives
\[
 s'=N_\delta\psi,\qquad
 h_0'=\kappa(I+i\bar H)y,\qquad
 b_0'=\beta(I+\bar B^\dagger\bar B)y.
\]
The first identity, together with the relation
\((\bra0_{\mathsf A_B}\otimes P_{\rm L})z=J_\delta\psi\) from
\eqref{eq:local-jump-output} shows that the public output produced by \(G\) is
\[
 \ket0_{\mathsf E}\otimes s'
 +(\bra0_{\mathsf A_B}\otimes P_{\rm L})z
 =\ket0_{\mathsf E}\otimes N_\delta\psi+J_\delta\psi
 =W_\delta\psi.
\]
For the remaining terms in the private output, substituting the values of
\(h,z,b\) from \eqref{eq:post-query-entries} gives
\begin{align*}
 ih^\perp
 &=i\kappa[I-(\proj0_{\mathsf A_H}\otimes I_{\mathsf S})]
 U_H(\ket0_{\mathsf A_H}\otimes I_{\mathsf S})y,\\
 -b^\perp
 &=\beta[I-(\proj0_{\mathsf A_B}\otimes
 \proj0_{\mathsf E}\otimes I_{\mathsf S})]
 U_B^\dagger(\ket0_{\mathsf A_B}\otimes\bar B y),\\
 [I-(\proj0_{\mathsf A_B}\otimes P_{\rm L})]z
 &=\beta[I-(\proj0_{\mathsf A_B}\otimes P_{\rm L})]
 U_B(\ket0_{\mathsf A_B}\ket0_{\mathsf E}\otimes I_{\mathsf S})y.
\end{align*}
Combining these three identities with the values of \(h_0'\) and \(b_0'\)
above, the three private entries produced by \(G\) are
\[
 \begin{pmatrix}
  (\ket0_{\mathsf A_H}\otimes I_{\mathsf S})h_0'+ih^\perp\\
  (\ket0_{\mathsf A_B}\ket0_{\mathsf E}\otimes I_{\mathsf S})b_0'-b^\perp\\
  [I-(\proj0_{\mathsf A_B}\otimes P_{\rm L})]z
 \end{pmatrix}
 =\gamma\psi
\]
by the definition of \(\gamma\) in \eqref{eq:local-catalyst}.  Hence
\begin{equation}
 G\bigl[(\ket0_{\mathsf E}\otimes\psi)\oplus\Omega\gamma\psi\bigr]
 =W_\delta\psi\oplus\gamma\psi.
 \label{eq:G-on-post-query}
\end{equation}

It remains to prove the catalyst bound \eqref{eq:local-catalyst-bound}.
Using the first formula for \(\gamma\psi\) in \eqref{eq:local-catalyst}, we
compute the squared norms of its three direct-sum entries.  Since
\(U_H^\dagger=U_H\) and \(U_H^2=I\),
\((I-iU_H)(I+iU_H)=2I\).  Hence the first entry satisfies
\[
 \left\|\kappa(I+iU_H)
 (\ket0_{\mathsf A_H}\otimes I_{\mathsf S})y\right\|^2
 =2\kappa^2\norm{y}^2.
\]

\Needspace{7\baselineskip}
For the second entry, unitarity of \(U_B\) and the projected unitary
encoding of \(\bar B\) in \eqref{eq:B-encoding} give
\begin{align*}
 \left\langle
 \ket0_{\mathsf A_B}\ket0_{\mathsf E}\otimes y,\,
 U_B^\dagger(\ket0_{\mathsf A_B}\otimes\bar B y)\right\rangle
 &=\left\langle
 U_B(\ket0_{\mathsf A_B}\ket0_{\mathsf E}\otimes y),\,
 \ket0_{\mathsf A_B}\otimes\bar B y\right\rangle\\
 &=\left\langle
 (\bra0_{\mathsf A_B}\otimes P_{\rm L})
 U_B(\ket0_{\mathsf A_B}\ket0_{\mathsf E}\otimes y),\,
 \bar B y\right\rangle
 =\norm{\bar B y}^2.
\end{align*}
Together with \(\norm{U_B^\dagger
(\ket0_{\mathsf A_B}\otimes\bar B y)}=\norm{\bar B y}\), we obtain
\begin{align*}
\left\|\beta\left[(\ket0_{\mathsf A_B}\ket0_{\mathsf E}\otimes I_{\mathsf S})
 +U_B^\dagger(\ket0_{\mathsf A_B}\otimes\bar B)\right]y
 \right\|^2&=\beta^2\bigl(\norm{y}^2+\norm{\bar B y}^2
 +2\operatorname{Re}\langle\bar B y,\bar B y\rangle\bigr)\\
&=\beta^2\bigl(\norm{y}^2+3\norm{\bar B y}^2\bigr).
\end{align*}
For the third entry, the same projected unitary encoding gives
\[
 (\proj0_{\mathsf A_B}\otimes P_{\rm L})
 U_B(\ket0_{\mathsf A_B}\ket0_{\mathsf E}\otimes y)
 =\ket0_{\mathsf A_B}\otimes\bar B y.
\]
The two projections \(\proj0_{\mathsf A_B}\otimes P_{\rm L}\) and
\(I-\proj0_{\mathsf A_B}\otimes P_{\rm L}\) have orthogonal ranges.
Together with unitarity of \(U_B\), this gives
\begin{align*}
 &\left\|\beta\bigl[I-(\proj0_{\mathsf A_B}\otimes P_{\rm L})\bigr]
 U_B(\ket0_{\mathsf A_B}\ket0_{\mathsf E}\otimes y)\right\|^2\\
 &=\beta^2\left(
 \left\|U_B(\ket0_{\mathsf A_B}\ket0_{\mathsf E}\otimes y)\right\|^2
 -\left\|(\proj0_{\mathsf A_B}\otimes P_{\rm L})
 U_B(\ket0_{\mathsf A_B}\ket0_{\mathsf E}\otimes y)\right\|^2\right)\\
 &=\beta^2\bigl(\norm{y}^2-\norm{\bar B y}^2\bigr).
\end{align*}
The projected unitary encoding of \(\bar B\) in \eqref{eq:B-encoding} and
unitarity of \(U_B\) imply
\[
 \norm{\bar B}
 =\norm{(\bra0_{\mathsf A_B}\otimes P_{\rm L})U_B
 (\ket0_{\mathsf A_B}\ket0_{\mathsf E}\otimes I_{\mathsf S})}\leq1.
\]
Thus \(\norm{\bar B y}\leq\norm{y}\).
Lemma~\ref{lem:rational-step} gives \(\norm{R_\delta}\leq1\), so
\(y=R_\delta\psi\) satisfies \(\norm{y}\leq\norm{\psi}\).
Because the three entries of \(\gamma\psi\) belong to orthogonal direct-sum
summands, their squared norms add.  Combining the three calculations above
and using \(2\kappa^2+4\beta^2=\alpha\delta\) from
\eqref{eq:local-scalars}, we obtain
\begin{equation}
 \begin{aligned}
 \norm{\gamma\psi}^2
 &=2\kappa^2\norm{y}^2
 +2\beta^2\bigl(\norm{y}^2+\norm{\bar B y}^2\bigr)\\
 &\leq(2\kappa^2+4\beta^2)\norm{y}^2
 \leq\alpha\delta\norm{\psi}^2.
 \end{aligned}
\end{equation}
Since this holds for every \(\psi\), it proves the catalyst bound
\eqref{eq:local-catalyst-bound}.
\end{proof}

\section{A one-query transducer for the full product}
\label{sec:full-product}

We now compose \(J\) copies of the local transducer \(S_\delta\) into a one-query
transducer \(S\) implementing \(W_J\), the Stinespring isometry for
\(\Phi_\delta^J\).  We use sequential composition of
transducers~\cite[Prop.~9.9]{BelovsJefferyYolcu2024}, as in
\cite[Sec.~4.2]{ChenEtAl2026}, placing the catalyst for each step in a
separate direct-sum component of the private space.

Recall from Section~\ref{sec:rational} that
\[
 \begin{aligned}
 \mathcal R_{<j}&=\bigotimes_{r=0}^{j-1}\cH_{\mathsf E_r},
 &\cH_j&=\mathcal R_{<j}\otimes
 \biggl(\bigotimes_{r=j}^{J-1}\mathbb C\ket0_{\mathsf E_r}\biggr)
 \otimes\cH_{\mathsf S},\\
 W_0&=I_{\cH_{\mathsf S}},
 &W_{j+1}&=(I_{\mathcal R_{<j}}\otimes W_\delta)W_j:
 \cH_{\mathsf S}\longrightarrow\cH_{j+1}.
 \end{aligned}
\]
Let \(\cH_{{\rm priv},j}\) be a copy of \(\cH_{\rm priv}\) for step \(j\), and set
\[
 \cK_j:=\mathcal R_{<j}\otimes\cH_{{\rm priv},j},
 \qquad \cK:=\bigoplus_{j=0}^{J-1}\cK_j.
\]
Thus \(\cH_J\) is the public space and \(\cK\) is the private space of the
global transducer.  The unitary constructed below acts on
\begin{equation}
 \cH_J\oplus\cK.
 \label{eq:global-space}
\end{equation}
For an input \(\psi\), set \(\psi_j:=W_j\psi\in\cH_j\); in particular,
\(\psi=\psi_0\in\cH_0\subseteq\cH_J\).  After omitting the fixed factors
\(\bigotimes_{r=j}^{J-1}\ket0_{\mathsf E_r}\), we regard \(\psi_j\) as a
vector in \(\mathcal R_{<j}\otimes\cH_{\mathsf S}\) and define
\begin{equation}
 x_j:=(I_{\mathcal R_{<j}}\otimes\gamma)\psi_j\in\cK_j.
 \label{eq:step-state-catalyst}
\end{equation}
Since \(\cH_j\subseteq\cH_{j+1}\), the pair
\(\psi_j\oplus x_j\) belongs to \(\cH_{j+1}\oplus\cK_j\).  Omitting the
fixed one-dimensional factors
\(\bigotimes_{r=j+1}^{J-1}\mathbb C\ket0_{\mathsf E_r}\) identifies this
space as
\begin{equation}
\begin{aligned}
 \cH_{j+1}\oplus\cK_j
 &=\left[\mathcal R_{<j+1}\otimes
 \biggl(\bigotimes_{r=j+1}^{J-1}\mathbb C\ket0_{\mathsf E_r}\biggr)
 \otimes\cH_{\mathsf S}\right]
 \oplus\left[\mathcal R_{<j}\otimes\cH_{{\rm priv},j}\right]\\
 &\cong(\mathcal R_{<j+1}\otimes\cH_{\mathsf S})
 \oplus(\mathcal R_{<j}\otimes\cH_{{\rm priv},j})\\
 &=\mathcal R_{<j}\otimes\left[
 (\cH_{\mathsf E_j}\otimes\cH_{\mathsf S})
 \oplus\cH_{{\rm priv},j}\right].
\end{aligned}
\label{eq:step-active-space}
\end{equation}
Under this identification, define
\[
 \left.G_j\right|_{\cH_{j+1}\oplus\cK_j}
 :=I_{\mathcal R_{<j}}\otimes G,
 \qquad
 \left.G_j\right|_{(\cH_{j+1}\oplus\cK_j)^\perp}:=I.
\]
Figure~\ref{fig:gj-register-spaces} displays this decomposition and the
action of \(G_j\) on \(\cH_{j+1}\oplus\cK_j\).
To distinguish the private copy of \(\mathsf E\) from the public Kraus-label
register \(\mathsf E_j\), the figure denotes the former by \(\mathsf F\).

% Inserted in Sec. 4 immediately after the displayed definition of G_j.
% The manuscript's original Figure 1 is unchanged.
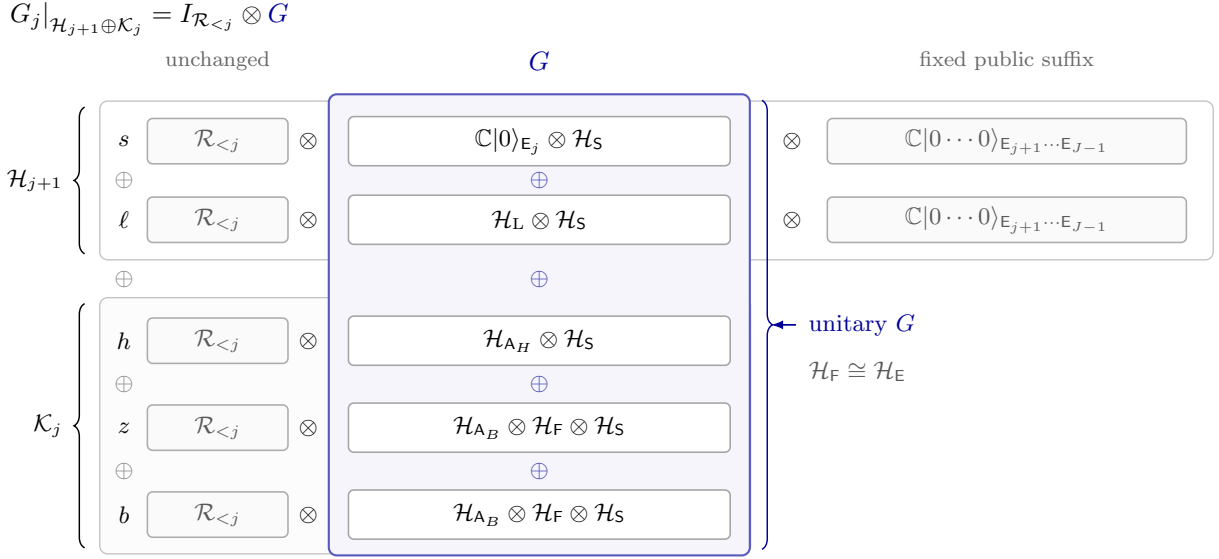
\begin{figure}[tbp]
\centering
\begingroup
\colorlet{GJlocal}{blue!60!black}
\begin{tikzpicture}[
x=1cm,y=1cm,font=\footnotesize,
every node/.style={inner sep=1pt},
reg/.style={draw=black!32,fill=black!2,rounded corners=2pt,
line width=.6pt,minimum height=.60cm,inner xsep=4pt},
local/.style={draw=black!35,fill=white,rounded corners=2pt,
line width=.6pt,minimum height=.65cm,minimum width=5.05cm},
brace/.style={decorate,decoration={brace,amplitude=4pt},line width=.55pt}
]

\node[anchor=west,font=\small] at (.05,1.22)
{$\left.G_j\right|_{\cH_{j+1}\oplus\cK_j}
=I_{\mathcal R_{<j}}\otimes{\color{GJlocal}G}$};

% Public and private groupings.  The local five-summand column spans both.
\draw[draw=black!24,rounded corners=3pt,line width=.5pt]
(1.32,.13) rectangle (16.05,-1.97);
\draw[draw=black!24,fill=black!1,rounded corners=3pt,line width=.5pt]
(1.32,-2.47) rectangle (9.93,-5.86);
\filldraw[fill=GJlocal!4,draw=GJlocal!65,rounded corners=3pt,line width=.8pt]
(4.35,.22) rectangle (9.92,-5.88);

\node[text=black!60,font=\scriptsize] at (2.88,.66) {unchanged};
\node[text=GJlocal,font=\small] at (7.135,.66) {$G$};
\node[text=black!60,font=\scriptsize] at (13.32,.66) {fixed public suffix};

\draw[brace] (1.12,-1.89)--(1.12,.05)
node[midway,left=7pt,align=center] {$\cH_{j+1}$};
\draw[brace] (1.12,-5.78)--(1.12,-2.55)
node[midway,left=7pt] {$\cK_j$};

% Five mutually orthogonal summands, in the order used in Figure 1.
\foreach \y/\name in {-.40/s,-1.44/\ell,-3.04/h,-4.19/z,-5.34/b} {
\node at (1.64,\y) {$\name$};
\node[reg,minimum width=1.86cm,text=black!65] at (2.88,\y)
{$\mathcal R_{<j}$};
\node at (4.08,\y) {$\otimes$};
}
\foreach \y in {-.92,-2.22,-3.615,-4.765} {
\node[text=black!50] at (1.64,\y) {$\oplus$};
\node[text=GJlocal!75] at (7.135,\y) {$\oplus$};
}
\node[local] at (7.135,-.40)
{$\mathbb C\ket0_{\mathsf E_j}\otimes\cH_{\mathsf S}$};
\node[local] at (7.135,-1.44)
{$\cH_{\rm L}\otimes\cH_{\mathsf S}$};
\node[local] at (7.135,-3.04)
{$\cH_{\mathsf A_H}\otimes\cH_{\mathsf S}$};
\foreach \y in {-4.19,-5.34} {
\node[local] at (7.135,\y)
{$\cH_{\mathsf A_B}\otimes\cH_{\mathsf F}\otimes\cH_{\mathsf S}$};
}
\foreach \y in {-.40,-1.44} {
\node at (10.48,\y) {$\otimes$};
\node[reg,minimum width=4.75cm,text=black!65] at (13.32,\y)
{$\mathbb C\ket{0\cdots0}_{\mathsf E_{j+1}\cdots\mathsf E_{J-1}}$};
}

% One brace selects the whole local direct sum, not a separate G per row.
\draw[brace,draw=GJlocal] (10.08,.14)--(10.08,-5.80);
\node[anchor=west,align=left,text=GJlocal] at (10.66,-2.83)
{unitary $G$};
\draw[-{Latex[length=1.6mm]},draw=GJlocal,line width=.65pt]
(10.53,-2.83)--(10.20,-2.83);
\node[anchor=west,text=black!65] at (10.66,-3.45)
{$\cH_{\mathsf F}\cong\cH_{\mathsf E}$};
\end{tikzpicture}
\endgroup

\caption{The five summands of $\cH_{j+1}\oplus\cK_j$.  The highlighted
column is the space on which the unitary $G$ constructed in
Lemma~\ref{lem:local-step} acts; its action is shown in
Figure~\ref{fig:local-catalyst-action}.  The unitary $G_j$ applies $G$ to this
column and leaves $\mathcal R_{<j}$ unchanged.  The first row is $\cH_j$.
Fixed public suffixes are displayed after regrouping tensor factors;
$\mathsf F$ denotes the private copy of $\mathsf E$.}
\label{fig:gj-register-spaces}
\end{figure}

Since \(G\) is unitary, so is \(G_j\).  Tensoring the local transducer
identity \eqref{eq:local-transducer} with \(I_{\mathcal R_{<j}}\) gives
\begin{equation}
 \begin{aligned}
 G_j\bigl(\psi_j\oplus
 (I_{\mathcal R_{<j}}\otimes\Omega)x_j\bigr)
 &=(I_{\mathcal R_{<j}}\otimes W_\delta)\psi_j\oplus x_j\\
 &=\psi_{j+1}\oplus x_j.
 \end{aligned}
 \label{eq:local-step-with-labels}
\end{equation}

Define the global catalyst map by
\begin{equation}
 \Gamma:\cH_{\mathsf S}\longrightarrow\cK,
 \qquad
 \Gamma\psi:=\bigoplus_{j=0}^{J-1}x_j.
 \label{eq:global-catalyst}
\end{equation}

\Needspace{12\baselineskip}
The step identities \eqref{eq:local-step-with-labels} compose to give a
one-query transducer for \(W_J\), while the local catalyst bounds sum to the
corresponding global bound.

\begin{proposition}[One-query transducer for the full product]
\label{prop:full-product}
There is a unitary \(S\) using one controlled application of
\(\Omega\) such that
\begin{equation}
 S\bigl(\psi\oplus\Gamma\psi\bigr)
 =W_J\psi\oplus\Gamma\psi,
 \label{eq:global-transducer}
\end{equation}
Moreover,
\begin{equation}
 \Gamma^\dagger\Gamma\preceq\tau I.
 \label{eq:global-catalyst-bound}
\end{equation}
\end{proposition}

\begin{proof}
Define
\begin{equation}
 Q:=I_{\cH_J}\oplus
 \bigoplus_{j=0}^{J-1}(I_{\mathcal R_{<j}}\otimes\Omega).
 \label{eq:global-query}
\end{equation}
On each \(\cK_j\), the factor
\(I_{\mathcal R_{<j}}\otimes\Omega\) applies \(\Omega\) to
\(\cH_{{\rm priv},j}\) and leaves the Kraus-label registers
\(\mathsf E_0,\ldots,\mathsf E_{j-1}\) unchanged.  Thus \(Q\) is one
controlled application of \(\Omega\) to the private space and counts as
one query under the convention in Section~\ref{sec:model}.  Set
\begin{equation}
 S:=(G_{J-1}\cdots G_0)Q.
 \label{eq:global-transducer-factorization}
\end{equation}
Since \(Q\) and each \(G_j\) are unitary, \(S\) is unitary.
Applying \(Q\) to the public input and the catalyst gives
\[
 Q(\psi\oplus\Gamma\psi)
 =\psi\oplus\bigoplus_{j=0}^{J-1}
 (I_{\mathcal R_{<j}}\otimes\Omega)x_j.
\]
We prove by induction on \(\ell\) that, for \(0\leq\ell\leq J\),
\begin{equation}
\begin{aligned}
 (G_{\ell-1}\cdots G_0)
 Q(\psi\oplus\Gamma\psi)
 &=\psi_\ell\oplus\left[
 \bigoplus_{j<\ell}x_j
 \oplus\bigoplus_{j\geq\ell}
 (I_{\mathcal R_{<j}}\otimes\Omega)x_j\right].
\end{aligned}
 \label{eq:global-induction-state}
\end{equation}
At \(\ell=0\), the empty product is the identity and \(\psi_0=\psi\), so the
induction claim follows from the definition of \(Q\).  Assume it holds for
\(\ell<J\).  Applying the local update
\eqref{eq:local-step-with-labels} at \(j=\ell\), while \(G_\ell\) acts as the
identity on \(\cK_j\) for \(j\ne\ell\), gives
\[
\begin{aligned}
&G_\ell\left\{\psi_\ell\oplus\left[
 \bigoplus_{j<\ell}x_j\oplus
 (I_{\mathcal R_{<\ell}}\otimes\Omega)x_\ell\oplus
 \bigoplus_{j>\ell}(I_{\mathcal R_{<j}}\otimes\Omega)x_j\right]\right\}\\
&=\psi_{\ell+1}\oplus\left[
 \bigoplus_{j<\ell}x_j\oplus x_\ell\oplus
 \bigoplus_{j>\ell}(I_{\mathcal R_{<j}}\otimes\Omega)x_j\right]\\
&=\psi_{\ell+1}\oplus\left[
 \bigoplus_{j<\ell+1}x_j\oplus
 \bigoplus_{j\geq\ell+1}(I_{\mathcal R_{<j}}\otimes\Omega)x_j\right].
\end{aligned}
\]
This proves the induction step.
At \(\ell=J\), the right-hand side is
\(\psi_J\oplus\Gamma\psi\).  Since \(\psi_J=W_J\psi\), this proves
the global transducer identity \eqref{eq:global-transducer}.

For every \(j\), the local catalyst bound
\eqref{eq:local-catalyst-bound} and the fact that \(W_j\) is an isometry give
\[
 \norm{x_j}^2
 =\langle\psi_j,
 (I_{\mathcal R_{<j}}\otimes\gamma^\dagger\gamma)\psi_j\rangle
 \leq\alpha\delta\norm{\psi_j}^2
 =\alpha\delta\norm\psi^2.
\]
Since \(\Gamma\psi=\bigoplus_{j=0}^{J-1}x_j\), the components \(x_j\) are
orthogonal, and hence
\[
 \norm{\Gamma\psi}^2
 =\sum_{j=0}^{J-1}\norm{x_j}^2
 \leq J\alpha\delta\norm\psi^2
 =\tau\norm\psi^2.
\]
Since this holds for every \(\psi\), it proves the global catalyst bound
\eqref{eq:global-catalyst-bound}.
\end{proof}

\section{Reuse without preparing the catalyst}
\label{sec:reuse}

The global transducer identity \eqref{eq:global-transducer} assumes that the
input-dependent catalyst \(\Gamma\psi\) is supplied, so it does not implement \(W_J\) alone.  We remove this requirement using the
reuse circuit of
Belovs, Jeffery, and Yolcu~\cite[Theorem~3.2]{BelovsJefferyYolcu2024},
following the analysis in~\cite[Sec.~5]{ChenEtAl2026}.  For an integer
\(N\geq1\), the
circuit distributes the amplitude of the public input uniformly over \(N\)
mutually orthogonal copies of \(\cH_J\), applies \(S\) successively to them
while passing the private output of each call to the next, and finally
recombines the public outputs.  We initialize the shared private input to the
zero vector in \(\cK\) and call \(N\) the reuse length.  We first describe
the reuse unitary \(U_N\) and show that it is a projected unitary encoding of
\(P_N\).  We then derive an expression for the approximation error
\(W_J-P_N\).

For a fixed reuse length \(N\), introduce a register \(\mathsf M\) with basis
states \(\ket\ell_{\mathsf M}\), \(0\leq\ell<N\).  The reuse circuit acts on
\[
 (\cH_{\mathsf M}\otimes\cH_J)\oplus\cK,
\]
where \(\mathsf M\) labels \(N\) copies of the public space and the single
private space \(\cK\) is shared by all \(N\) calls.  We use the following
input and output isometries for the projected unitary encoding:
\[
 \begin{aligned}
 \jmath_{\rm in}\psi
 &=\bigl(\ket0_{\mathsf M}\ket0_{\mathsf E_0}\cdots
 \ket0_{\mathsf E_{J-1}}\otimes\psi\bigr)\oplus0,
 \qquad \psi\in\cH_{\mathsf S},\\
 \jmath_{\rm out}\varphi
 &=(\ket0_{\mathsf M}\otimes\varphi)\oplus0,
 \qquad \varphi\in\cH_J.
 \end{aligned}
\]
Fix \(\psi\in\cH_{\mathsf S}\) and set
\(\psi_0:=\ket0_{\mathsf E_0}\cdots
\ket0_{\mathsf E_{J-1}}\otimes\psi\in\cH_0\).  Thus the circuit input is
\(\jmath_{\rm in}\psi=(\ket0_{\mathsf M}\otimes\psi_0)\oplus0\).
To prepare the public input uniformly over the \(N\) copies, choose a
unitary \(F_N\) satisfying
\[
 F_N\ket0_{\mathsf M}
 =\frac1{\sqrt N}\sum_{\ell=0}^{N-1}\ket\ell_{\mathsf M}.
\]
Applying \(F_N\) to \(\mathsf M\) in \(\jmath_{\rm in}\psi\) prepares
\[
 \left(\frac1{\sqrt N}\sum_{\ell=0}^{N-1}
 \ket\ell_{\mathsf M}\otimes\psi_0\right)\oplus z_0,
 \qquad z_0=0.
\]
At the \(\ell\)-th call, \(S\) acts on the public component labeled \(\ell\)
together with the shared private space \(\cK\).  Denote the public output in
the copy labeled \(\ell\) by \(y_\ell\) and the updated private component by
\(z_{\ell+1}\).  The complete state then changes as follows:
\[
 \begin{aligned}
 &\left(\sum_{r<\ell}\ket r_{\mathsf M}\otimes y_r
 +\frac1{\sqrt N}\sum_{r=\ell}^{N-1}
 \ket r_{\mathsf M}\otimes\psi_0\right)\oplus z_\ell\longmapsto
 \left(\sum_{r\leq\ell}\ket r_{\mathsf M}\otimes y_r
 +\frac1{\sqrt N}\sum_{r=\ell+1}^{N-1}
 \ket r_{\mathsf M}\otimes\psi_0\right)\oplus z_{\ell+1}.
 \end{aligned}
\]
The input to the two summands on which \(S\) acts is
\(N^{-1/2}\psi_0\oplus z_\ell\).  Hence the two updated components satisfy
\begin{equation}
 y_\ell\oplus z_{\ell+1}
 =S\bigl(N^{-1/2}\psi_0\oplus z_\ell\bigr),
 \qquad 0\leq\ell<N.
 \label{eq:reuse-recurrence}
\end{equation}
After the \(N\) calls, applying \(F_N^\dagger\) to \(\mathsf M\) maps
\[
 \left(\sum_{\ell=0}^{N-1}\ket\ell_{\mathsf M}\otimes y_\ell\right)
 \oplus z_N
 \longmapsto
 \left((F_N^\dagger\otimes I_{\cH_J})
 \sum_{\ell=0}^{N-1}\ket\ell_{\mathsf M}\otimes y_\ell\right)
 \oplus z_N.
\]
Let \(U_N\) denote the unitary implemented on this space by the reuse circuit,
including \(F_N\) and \(F_N^\dagger\).  Its projected block with respect to
\(\jmath_{\rm in}\) and \(\jmath_{\rm out}\) acts as
\begin{equation}
 \begin{aligned}
 \jmath_{\rm out}^\dagger U_N\jmath_{\rm in}\psi
 &=\left(\bra0_{\mathsf M}F_N^\dagger\otimes I_{\cH_J}\right)
 \sum_{\ell=0}^{N-1}\ket\ell_{\mathsf M}\otimes y_\ell\\
 &=\frac1{\sqrt N}\sum_{\ell=0}^{N-1}y_\ell
 =:P_N\psi.
 \end{aligned}
 \label{eq:reuse-block-encoding}
\end{equation}
Thus \((U_N,\jmath_{\rm in},\jmath_{\rm out})\) is a projected unitary
encoding of \(P_N:\cH_{\mathsf S}\to\cH_J\).  The circuit uses \(N\) calls to
\(S\).

We next derive an exact expression for the error of \(P_N\) as an
approximation to \(W_J\).  To this end, write \(S\) in blocks with respect to
\(\cH_J\oplus\cK\):
\begin{equation}
 S=
 \begin{pmatrix}
  S_{00}&S_{01}\\
  S_{10}&S_{11}
 \end{pmatrix},
 \label{eq:transducer-blocks}
\end{equation}
where
\[
 S_{00}:\cH_J\to\cH_J,\qquad S_{01}:\cK\to\cH_J,\qquad
 S_{10}:\cH_J\to\cK,\qquad S_{11}:\cK\to\cK.
\]
Comparing the public and private components in the global transducer identity
\eqref{eq:global-transducer} gives
\begin{equation}
 W_J\psi=S_{00}\psi_0+S_{01}\Gamma\psi,
 \qquad
 S_{10}\psi_0=(I-S_{11})\Gamma\psi.
 \label{eq:block-transducer-identity}
\end{equation}
Define the averaging polynomial
\[
 g_N(\zeta):=\frac1N\sum_{k=0}^{N-1}\zeta^k.
\]

The approximation error of the reuse map has the following exact form.

\begin{lemma}[Approximation error of \(P_N\)]
\label{lem:reuse}
For every \(N\geq1\),
\begin{equation}
 W_J-P_N=S_{01}g_N(S_{11})\Gamma.
 \label{eq:reuse-error}
\end{equation}
\end{lemma}

\begin{proof}
Following the reuse analysis in
\cite[Sec.~5, Lemma~4]{ChenEtAl2026}, we first solve the recurrence for the
private state \(z_\ell\) and then substitute the result into the expression
for \(P_N\) in \eqref{eq:reuse-block-encoding}.
Applying the block decomposition \eqref{eq:transducer-blocks} to the
\(\ell\)-th call in \eqref{eq:reuse-recurrence} gives
\[
 y_\ell=N^{-1/2}S_{00}\psi_0+S_{01}z_\ell,
 \qquad
 z_{\ell+1}=N^{-1/2}S_{10}\psi_0+S_{11}z_\ell.
\]
Iterating the recurrence for \(z_\ell\) from \(z_0=0\) and using
\(S_{10}\psi_0=(I-S_{11})\Gamma\psi\) from
\eqref{eq:block-transducer-identity} gives
\[
 \begin{aligned}
 z_\ell
 &=\frac1{\sqrt N}\sum_{r=0}^{\ell-1}S_{11}^rS_{10}\psi_0\\
 &=\frac1{\sqrt N}\sum_{r=0}^{\ell-1}
 S_{11}^r(I-S_{11})\Gamma\psi\\
 &=\frac1{\sqrt N}(I-S_{11}^\ell)\Gamma\psi.
 \end{aligned}
\]
We now substitute this expression for \(z_\ell\) into
\eqref{eq:reuse-block-encoding}.  Using
\(W_J\psi=S_{00}\psi_0+S_{01}\Gamma\psi\) from
\eqref{eq:block-transducer-identity}, we obtain
\begin{align*}
 P_N\psi
 &=\frac1{\sqrt N}\sum_{\ell=0}^{N-1}
   \left(N^{-1/2}S_{00}\psi_0+S_{01}z_\ell\right)\\
 &=S_{00}\psi_0+\frac1N\sum_{\ell=0}^{N-1}
 S_{01}(I-S_{11}^\ell)\Gamma\psi\\
 &=S_{00}\psi_0+S_{01}\Gamma\psi
   -S_{01}\left(\frac1N\sum_{\ell=0}^{N-1}S_{11}^\ell\right)
   \Gamma\psi\\
 &=W_J\psi-S_{01}g_N(S_{11})\Gamma\psi.
\end{align*}
Since this equality holds for every \(\psi\),
\(P_N=W_J-S_{01}g_N(S_{11})\Gamma\), which is equivalent to
\eqref{eq:reuse-error}.
\end{proof}

\section{Combining reuse lengths}
\label{sec:lcu}

For every \(N\), the reuse construction in Section~\ref{sec:reuse} gives a
projected unitary encoding of \(P_N\) using \(N\) calls to \(S\).  The
error identity \eqref{eq:reuse-error} expresses the approximation error
\(W_J-P_N\) through \(g_N(S_{11})\), the average of
\(I,S_{11},\ldots,S_{11}^{N-1}\).  The norm of this average need not decrease
rapidly with \(N\).  However, taking a linear combination
\(\sum_N\lambda_NP_N\) with \(\sum_N\lambda_N=1\) replaces
\(g_N(S_{11})\) in the error formula by
\(\sum_N\lambda_Ng_N(S_{11})\), allowing cancellation among different reuse
lengths.  We first construct a projected unitary encoding of
\(\sum_N\lambda_NP_N\) for arbitrary coefficients using an LCU circuit and
then choose the coefficients so that \(\sum_N\lambda_Ng_N(S_{11})\) has
rapidly decaying norm.

\subsection{A general LCU construction}
By \eqref{eq:reuse-block-encoding}, each reuse unitary \(U_N\) is a projected
unitary encoding of \(P_N\).  Coherently selecting these unitaries
gives the following construction.

\begin{lemma}[LCU of reuse maps]
\label{lem:lcu-reuse-maps}
Fix \(R\geq1\) and coefficients \(\lambda_1,\ldots,\lambda_R\).  Set
\[
 \Lambda:=\sum_{N=1}^R|\lambda_N|>0,
 \qquad
 \widetilde W(\lambda):=\sum_{N=1}^R\lambda_NP_N.
\]
There is a normalization-\(\Lambda\) projected unitary encoding
\((\mathcal V,\iota_{\rm in},\iota_{\rm out})\) of
\(\widetilde W(\lambda)\), where \(\mathcal V\) acts on a Hilbert space
\(\cH_{\rm work}\),
\(\iota_{\rm in}:\cH_{\mathsf S}\to\cH_{\rm work}\), and
\(\iota_{\rm out}:\cH_J\to\cH_{\rm work}\).  Thus
\[
 \iota_{\rm out}^\dagger\mathcal V\iota_{\rm in}
 =\frac{\widetilde W(\lambda)}{\Lambda}.
\]
The unitary \(\mathcal V\) uses at most \(R\) calls to \(S\).
\end{lemma}

\begin{proof}
We apply the LCU construction for projected unitary encodings
\cite[Lemma~52 of the full version]{GilyenEtAl2019} to the reuse unitaries
\(U_1,\ldots,U_R\).  Use a \(\lceil\log_2 R\rceil\)-qubit reuse register
\(\mathsf M\), large enough for every reuse length \(N\leq R\).  A
\(\lceil\log_2(R+1)\rceil\)-qubit coefficient
register \(\mathsf N\) stores labels \(0,1,\ldots,R\), with
\(N\geq1\) selecting reuse length \(N\).
The LCU circuit acts on
\[
 \cH_{\rm work}:=\cH_{\mathsf N}\otimes
 \bigl[(\cH_{\mathsf M}\otimes\cH_J)\oplus\cK\bigr].
\]
Adjoin \(\ket0_{\mathsf N}\) to the input and output isometries in
\eqref{eq:reuse-block-encoding}, and denote the resulting isometries by
\(\iota_{\rm in}:\cH_{\mathsf S}\to\cH_{\rm work}\) and
\(\iota_{\rm out}:\cH_J\to\cH_{\rm work}\):
\begin{equation}
 \begin{aligned}
 \iota_{\rm in}\psi
 &=\ket0_{\mathsf N}\otimes\jmath_{\rm in}\psi
 =\ket0_{\mathsf N}\otimes
 \bigl[(\ket0_{\mathsf M}\ket0_{\mathsf E_0}\cdots
 \ket0_{\mathsf E_{J-1}}\otimes\psi)\oplus0\bigr],\\
 \iota_{\rm out}\varphi
 &=\ket0_{\mathsf N}\otimes\jmath_{\rm out}\varphi
 =\ket0_{\mathsf N}\otimes
 \bigl[(\ket0_{\mathsf M}\otimes\varphi)\oplus0\bigr].
 \end{aligned}
 \label{eq:lcu-embeddings}
\end{equation}
Choose a unitary \(F_\lambda\) on \(\mathsf N\) satisfying
\begin{equation}
 F_\lambda\ket0_{\mathsf N}
 =\sum_{N=1}^R\sqrt{\frac{|\lambda_N|}{\Lambda}}\ket N_{\mathsf N}.
 \label{eq:lcu-preparation}
\end{equation}
For \(\chi\in(\cH_{\mathsf M}\otimes\cH_J)\oplus\cK\), define SELECT by
\begin{equation}
 \operatorname{SELECT}(\ket N_{\mathsf N}\otimes\chi)
 =\operatorname{sgn}(\lambda_N)\ket N_{\mathsf N}\otimes U_N\chi.
 \label{eq:lcu-select}
\end{equation}
Preparing the coefficients, applying SELECT, and undoing the preparation
defines the LCU unitary
\begin{equation}
 \mathcal V:=(F_\lambda^\dagger\otimes I)
 \operatorname{SELECT}(F_\lambda\otimes I).
 \label{eq:lcu-unitary}
\end{equation}
For \(\psi\in\cH_{\mathsf S}\), expand \(\mathcal V\) using
\eqref{eq:lcu-embeddings}--\eqref{eq:lcu-unitary} and substitute
\(\jmath_{\rm out}^\dagger U_N\jmath_{\rm in}=P_N\) from
\eqref{eq:reuse-block-encoding}:
\[
 \begin{aligned}
 \iota_{\rm out}^\dagger\mathcal V\iota_{\rm in}\psi
 &=\bigl(\bra0_{\mathsf N}F_\lambda^\dagger\otimes
   \jmath_{\rm out}^\dagger\bigr)
   \sum_{\lambda_N\ne0}\sqrt{\frac{|\lambda_N|}{\Lambda}}\,
   \operatorname{sgn}(\lambda_N)\ket N_{\mathsf N}\otimes
   U_N\jmath_{\rm in}\psi\\
 &=\sum_{\lambda_N\ne0}\frac{|\lambda_N|}{\Lambda}
   \operatorname{sgn}(\lambda_N)
   \jmath_{\rm out}^\dagger U_N\jmath_{\rm in}\psi\\
 &=\sum_{\lambda_N\ne0}\frac{|\lambda_N|}{\Lambda}
   \operatorname{sgn}(\lambda_N)P_N\psi\\
 &=\frac1\Lambda\sum_{N=1}^R\lambda_NP_N\psi
 =\frac{\widetilde W(\lambda)}\Lambda\psi.
 \end{aligned}
\]
Since \(\psi\) is arbitrary,
\(\iota_{\rm out}^\dagger\mathcal V\iota_{\rm in}
=\widetilde W(\lambda)/\Lambda\), which proves the projected-unitary-encoding
identity.  The coefficient \(\ell_1\)-norm \(\Lambda\) is the LCU
normalization.

To implement SELECT coherently, condition the reuse circuit on the
coefficient label \(N\) stored in \(\mathsf N\).  First apply
\((F_N\otimes I_{\cH_J})\oplus I_{\cK}\).  For
\(\ell=0,\ldots,R-1\), apply \(S\) to
\((\ket\ell_{\mathsf M}\otimes\cH_J)\oplus\cK\), as in
\eqref{eq:reuse-recurrence}, conditioned on \(\ell<N\).  Finally apply
\((F_N^\dagger\otimes I_{\cH_J})\oplus I_{\cK}\) and multiply by
\(\operatorname{sgn}(\lambda_N)\).  This implements
\eqref{eq:lcu-select} with \(R\) controlled calls to \(S\), so
\(\mathcal V\) uses at most \(R\) calls to \(S\).  The gate complexity of
the oracle-independent operations is analyzed in
Section~\ref{sec:gate-complexity}.
\end{proof}

\subsection{Choosing the coefficients}
We now choose a linear combination that approximates \(W_J\) while keeping
the LCU normalization \(\Lambda\) bounded.  For coefficients satisfying
\(\sum_N\lambda_N=1\), summing the error identity
\eqref{eq:reuse-error} gives
\[
 W_J-\sum_N\lambda_NP_N
 =S_{01}\left[\sum_N\lambda_Ng_N(S_{11})\right]\Gamma.
\]
Since \(S\) is unitary, \(\|S_{01}\|\leq1\), and the global catalyst bound
\eqref{eq:global-catalyst-bound} gives \(\|\Gamma\|\leq\sqrt\tau\).
Therefore
\begin{equation}
 \left\|W_J-\sum_N\lambda_NP_N\right\|
 \leq\sqrt\tau
 \left\|\sum_N\lambda_Ng_N(S_{11})\right\|.
 \label{eq:weighted-reuse-error}
\end{equation}
Thus we seek coefficients for which
\(\sum_N\lambda_Ng_N(S_{11})\) has small norm and
\(\Lambda=\sum_N|\lambda_N|\) remains bounded.  The first condition controls the
approximation error in \eqref{eq:weighted-reuse-error}; the second keeps the
normalization in Lemma~\ref{lem:lcu-reuse-maps}, and hence the amplification
overhead in Section~\ref{sec:oaa}, constant.

For an integer \(q\geq1\), define
\begin{equation}
 p(\zeta):=\zeta^2(1+\zeta^2),\qquad
 \mathcal F_q(\zeta):=\left(\frac{p(\zeta)}2\right)^{2q}.
 \label{eq:Fq}
\end{equation}
Appendix~\ref{sec:private-block-proof} proves the following bound.

\begin{proposition}[Polynomial bound]
\label{prop:private-polynomial}
There is a universal constant \(C_0\) such that, for every positive
integer \(q\) satisfying \(\alpha\delta\leq1/2\), \(J\geq8q\), and
\(q\geq C_0\tau\),
\begin{equation}
 \left\|\mathcal F_q(S_{11})\right\|
 \leq\left(\frac{C_0\tau}{q}\right)^q.
 \label{eq:factorial-bound}
\end{equation}
\end{proposition}

Expressing \(\mathcal F_q\) directly as a linear combination of the
polynomials \(g_N\) yields a coefficient \(\ell_1\)-norm that grows with
\(q\), and hence a growing LCU normalization.  Following the smoothing
technique of \cite[Sec.~7.1]{ChenEtAl2026}, we instead define
\begin{equation}
 \mathcal Q_q(\zeta)=
 \mathcal F_q(\zeta)g_{12q}(\zeta)
 =\sum_k e_k\zeta^k.
 \label{eq:Qq}
\end{equation}
Since \(S_{11}\) is a contraction,
\begin{equation}
 \norm{\mathcal Q_q(S_{11})}
 \leq\norm{\mathcal F_q(S_{11})}
       \norm{g_{12q}(S_{11})}
  \leq\norm{\mathcal F_q(S_{11})}.
 \label{eq:smoothing-bound}
\end{equation}
Thus the additional factor \(g_{12q}\) does not increase the operator norm
that controls the approximation error in \eqref{eq:weighted-reuse-error}.
At the same time, multiplying by \(g_{12q}\) replaces the monomial
coefficients of \(\mathcal F_q\) by their averages.  This smoothing
makes the coefficient \(\ell_1\)-norm constant after the change of basis to
the polynomials \(g_N\), as shown below.  Set
\(e_k=0\) for \(k<0\) or \(k>\deg\mathcal Q_q\), and define
\begin{equation}
 \lambda_N=N(e_{N-1}-e_N).
 \label{eq:lambda}
\end{equation}

\begin{lemma}[Coefficient identity]
\label{lem:reuse-combination}
The polynomial \(\mathcal Q_q\) has the representation
\begin{equation}
 \mathcal Q_q(\zeta)=
 \sum_{N=1}^{20q}\lambda_Ng_N(\zeta),
 \qquad
 \sum_N\lambda_N=1,
 \qquad
 \sum_N|\lambda_N|=2.
 \label{eq:reuse-combination}
\end{equation}
\end{lemma}

\begin{proof}
Since \(\deg\mathcal Q_q\leq8q+(12q-1)=20q-1\),
\(e_N=0\) for \(N\geq20q\).  By the definition in \eqref{eq:lambda},
\(\lambda_N=0\) for \(N>20q\).
We first verify that the coefficients defined in \eqref{eq:lambda} give the
claimed expansion of \(\mathcal Q_q\) in the polynomials \(g_N\).  Using the
definition of \(g_N\) and then telescoping the inner sum, we obtain
\begin{align*}
 \sum_{N=1}^{20q}\lambda_Ng_N(\zeta)
 &=\sum_{N=1}^{20q}(e_{N-1}-e_N)
   \sum_{k=0}^{N-1}\zeta^k\\
 &=\sum_{k=0}^{20q-1}\zeta^k
   \sum_{N=k+1}^{20q}(e_{N-1}-e_N)\\
 &=\sum_{k=0}^{20q-1}(e_k-e_{20q})\zeta^k
 =\sum_{k=0}^{20q-1}e_k\zeta^k
 =\mathcal Q_q(\zeta).
\end{align*}
Since \(g_N(1)=1\) and \(\mathcal F_q(1)=1\), evaluating this expansion at
\(\zeta=1\) gives
\(\sum_N\lambda_N=\mathcal Q_q(1)
=\mathcal F_q(1)g_{12q}(1)=1\).

It remains to compute the coefficient \(\ell_1\)-norm.  Expand
\[
 \mathcal F_q(\zeta)=
 \zeta^{4q}\left(\frac{1+\zeta^2}{2}\right)^{2q}
 =\sum_k b_k\zeta^k.
\]
Its coefficients are
\[
 b_{4q+2r}=2^{-2q}\binom{2q}{r}\quad(0\leq r\leq2q),
 \qquad b_k=0\quad\text{otherwise}.
\]
The coefficients satisfy \(\sum_kb_k=1\).  Using
\(r\binom{2q}{r}=2q\binom{2q-1}{r-1}\), their first moment is
\begin{align*}
 \sum_k k b_k
 &=2^{-2q}\sum_{r=0}^{2q}(4q+2r)\binom{2q}{r}\\
 &=4q+2^{1-2q}\sum_{r=1}^{2q}
   2q\binom{2q-1}{r-1}\\
 &=4q+2^{1-2q}\bigl(2q\,2^{2q-1}\bigr)=6q.
\end{align*}
Since \(g_{12q}(\zeta)=(12q)^{-1}\sum_{t=0}^{12q-1}\zeta^t\), the monomial
coefficients of \(\mathcal Q_q=\mathcal F_qg_{12q}\) are
\[
 e_N=\frac1{12q}\sum_{t=0}^{12q-1}b_{N-t}.
\]
Their consecutive differences satisfy
\begin{align*}
 e_{N-1}-e_N
 &=\frac1{12q}\left(
   \sum_{t=0}^{12q-1}b_{N-1-t}
   -\sum_{t=0}^{12q-1}b_{N-t}\right)\\
 &=\frac1{12q}\left(
   \sum_{t=1}^{12q}b_{N-t}
   -\sum_{t=0}^{12q-1}b_{N-t}\right)\\
 &=\frac{b_{N-12q}-b_N}{12q}.
\end{align*}
The sequences \(b_N\) and \(b_{N-12q}\) are nonnegative and have disjoint
supports \([4q,8q]\) and \([16q,20q]\), respectively.  Therefore
\(\lvert b_{N-12q}-b_N\rvert=b_{N-12q}+b_N\).  Substituting
\eqref{eq:lambda} and then reindexing the term involving \(b_{N-12q}\)
gives
\begin{align*}
 \sum_N|\lambda_N|
 &=\frac1{12q}\sum_N
   N\bigl(b_{N-12q}+b_N\bigr)\\
 &=\frac1{12q}\left[
   \sum_k(k+12q)b_k+\sum_k k b_k\right]\\
 &=\sum_kb_k+\frac{2}{12q}\sum_k k b_k
 =1+\frac{2(6q)}{12q}=2.
\end{align*}
\end{proof}

We now combine the LCU construction, the coefficient identity, and the
polynomial bound to obtain a normalization-\(2\) projected unitary encoding
of an approximation to \(W_J\).

\begin{proposition}[LCU approximation of \(W_J\)]
\label{prop:lcu-approximation}
For the transducer \(S\) constructed in
Proposition~\ref{prop:full-product}, let \(P_N\) be the reuse maps from
Section~\ref{sec:reuse}.  Let \(q\geq1\) satisfy
\(\alpha\delta\leq1/2\), \(J\geq8q\), and \(q\geq C_0\tau\).  Define
\(\lambda_N\) by \eqref{eq:Qq} and \eqref{eq:lambda}, and set
\(\widetilde W_J:=\sum_{N=1}^{20q}\lambda_NP_N\).  Then there is a
normalization-\(2\) projected unitary encoding
\((\mathcal V,\iota_{\rm in},\iota_{\rm out})\) of \(\widetilde W_J\).
The unitary \(\mathcal V\) uses at most \(20q\) calls to \(S\), and
\begin{equation}
 \norm{W_J-\widetilde W_J}
 \leq\sqrt\tau\left(\frac{C_0\tau}{q}\right)^q.
 \label{eq:lcu-approximation-error}
\end{equation}
\end{proposition}

\begin{proof}
Lemma~\ref{lem:reuse-combination} gives
\(\lambda_N=0\) for \(N>20q\) and
\(\Lambda=\sum_N|\lambda_N|=2\).  Applying
Lemma~\ref{lem:lcu-reuse-maps} with \(R=20q\) proves the claimed projected
unitary encoding and query count.

Lemma~\ref{lem:reuse-combination} also gives
\(\sum_N\lambda_N=1\).  Evaluating the polynomial
identity in \eqref{eq:reuse-combination} at \(S_{11}\) gives
\[
 \sum_N\lambda_Ng_N(S_{11})
 =\mathcal Q_q(S_{11})
 =\mathcal F_q(S_{11})g_{12q}(S_{11}).
\]
Combining this identity with the approximation bound
\eqref{eq:weighted-reuse-error} and the smoothing bound
\eqref{eq:smoothing-bound}, and then applying
Proposition~\ref{prop:private-polynomial}, we obtain
\begin{align*}
 \norm{W_J-\widetilde W_J}
 &\leq\sqrt\tau\,\norm{\mathcal Q_q(S_{11})}\leq\sqrt\tau\,\norm{\mathcal F_q(S_{11})}\leq\sqrt\tau\left(\frac{C_0\tau}{q}\right)^q.
\end{align*}
This proves \eqref{eq:lcu-approximation-error}.
\end{proof}

\section{The final algorithm}
\label{sec:oaa}

\subsection{Oblivious amplitude amplification}

We now apply oblivious amplitude amplification (OAA) for isometries to the
normalization-\(2\) projected unitary encoding from
Proposition~\ref{prop:lcu-approximation}, following
Cleve and Wang
\cite[Sec.~4.4 of the full version]{CleveWang2017}.
The OAA circuit uses reflections about the ranges of the input and
output isometries in \eqref{eq:lcu-embeddings}.  The corresponding orthogonal
projectors are
\[
 \Pi_{\rm in}:=\iota_{\rm in}\iota_{\rm in}^\dagger,
 \qquad
 \Pi_{\rm out}:=\iota_{\rm out}\iota_{\rm out}^\dagger.
\]

\begin{lemma}[Oblivious amplitude amplification for isometries]
\label{lem:oaa}
Define
\begin{equation}
 \mathcal U_{\rm amp}=
 -\mathcal V(2\Pi_{\rm in}-I)\mathcal V^\dagger
 (2\Pi_{\rm out}-I)\mathcal V.
 \label{eq:oaa}
\end{equation}
With respect to \(\iota_{\rm in}\) and \(\iota_{\rm out}\), the projected
block of \(\mathcal U_{\rm amp}\) is
\begin{equation}
 \iota_{\rm out}^\dagger\mathcal U_{\rm amp}\iota_{\rm in}
 =\frac32\widetilde W_J
 -\frac12\widetilde W_J\widetilde W_J^\dagger\widetilde W_J.
 \label{eq:oaa-good}
\end{equation}
If \(\norm{\widetilde W_J-W_J}\leq\eta\leq1/8\), then
\begin{equation}
 \norm{\mathcal U_{\rm amp}\iota_{\rm in}
 -\iota_{\rm out}W_J}\leq5\eta.
 \label{eq:oaa-isometry}
\end{equation}
\end{lemma}

\begin{proof}
We first compute the projected block, then bound its error and the component
orthogonal to \(\operatorname{ran}\Pi_{\rm out}\).  Set
\(X=\iota_{\rm out}^\dagger\mathcal V\iota_{\rm in}
=\widetilde W_J/2\).  The definitions of \(\Pi_{\rm in}\) and
\(\Pi_{\rm out}\) imply
\[
 \iota_{\rm out}^\dagger\mathcal V(2\Pi_{\rm in}-I)
 =2X\iota_{\rm in}^\dagger-\iota_{\rm out}^\dagger\mathcal V,
 \qquad
 \iota_{\rm in}^\dagger\mathcal V^\dagger
 (2\Pi_{\rm out}-I)\mathcal V\iota_{\rm in}
 =2X^\dagger X-I.
\]
Using these identities, the unitarity of \(\mathcal V\), and
\(\iota_{\rm out}^\dagger(2\Pi_{\rm out}-I)
=\iota_{\rm out}^\dagger\), we expand \eqref{eq:oaa} as
\begin{align*}
 \iota_{\rm out}^\dagger\mathcal U_{\rm amp}\iota_{\rm in}
 &=-\bigl(2X\iota_{\rm in}^\dagger
   -\iota_{\rm out}^\dagger\mathcal V\bigr)
   \mathcal V^\dagger(2\Pi_{\rm out}-I)
   \mathcal V\iota_{\rm in}\\
 &=X-2X\,\iota_{\rm in}^\dagger\mathcal V^\dagger
   (2\Pi_{\rm out}-I)\mathcal V\iota_{\rm in}\\
 &=X-2X(2X^\dagger X-I)=3X-4XX^\dagger X,
\end{align*}
which proves \eqref{eq:oaa-good}.
To compare this projected block with \(W_J\), take the polar decomposition
\(\widetilde W_J=U_{\rm p}\Sigma\).  Lemma~\ref{lem:rational-step} and the
product construction \eqref{eq:product-stinespring} show that \(W_J\) is an
isometry.  Hence, for every unit vector \(\psi\),
\[
 1-\eta\leq\norm{W_J\psi}-\norm{(\widetilde W_J-W_J)\psi}
 \leq\norm{\widetilde W_J\psi}
 \leq\norm{W_J\psi}+\norm{(\widetilde W_J-W_J)\psi}
 \leq1+\eta.
\]
Since \(\Sigma=(\widetilde W_J^\dagger\widetilde W_J)^{1/2}\), these bounds
imply \(\operatorname{spec}(\Sigma)\subset[1-\eta,1+\eta]\).
In particular, \(\Sigma\) is invertible.  Thus the polar factor is
\(U_{\rm p}=\widetilde W_J\Sigma^{-1}\), and
\[
 U_{\rm p}^\dagger U_{\rm p}
 =\Sigma^{-1}\widetilde W_J^\dagger\widetilde W_J\Sigma^{-1}=I.
\]
Hence \(U_{\rm p}\) is an isometry.  The spectral bound on \(\Sigma\) gives
\(\norm{I-\Sigma}\leq\eta\), and therefore
\[
 \norm{U_{\rm p}-\widetilde W_J}
 =\norm{U_{\rm p}(I-\Sigma)}
 =\norm{I-\Sigma}\leq\eta,
\]
where the second equality uses \(U_{\rm p}^\dagger U_{\rm p}=I\).
Consequently,
\begin{equation}
 \norm{U_{\rm p}-W_J}
 \leq\norm{U_{\rm p}-\widetilde W_J}
    +\norm{\widetilde W_J-W_J}
 \leq2\eta.
 \label{eq:oaa-polar-error}
\end{equation}
Define \(f(\sigma):=\frac32\sigma-\frac12\sigma^3\).  The formula for the
projected block in \eqref{eq:oaa-good} and
\(\widetilde W_J\widetilde W_J^\dagger\widetilde W_J
=U_{\rm p}\Sigma^3\) give
\(\iota_{\rm out}^\dagger\mathcal U_{\rm amp}\iota_{\rm in}
=U_{\rm p}f(\Sigma)\).
For \(|x|\leq\eta\leq1/8\),
\[
 1-f(1+x)=\frac{x^2(3+x)}2\in[0,2\eta^2].
\]
Applying this inequality to each eigenvalue of \(\Sigma\) gives
\[
 0\preceq I-f(\Sigma)\preceq2\eta^2 I.
\]
In particular, \(0\preceq f(\Sigma)\preceq I\) and
\(\norm{I-f(\Sigma)}\leq2\eta^2\).  Combining this bound with
\eqref{eq:oaa-polar-error} gives
\begin{equation}
 \norm{U_{\rm p}f(\Sigma)-W_J}
 \leq\norm{U_{\rm p}-W_J}+\norm{f(\Sigma)-I}
 \leq2\eta+2\eta^2.
 \label{eq:oaa-output-error}
\end{equation}

Since \(\mathcal U_{\rm amp}\) is unitary and \(\iota_{\rm in}\) is an
isometry, \(\mathcal U_{\rm amp}\iota_{\rm in}\) is an isometry.  Its
projection onto \(\operatorname{ran}\Pi_{\rm out}\) is
\[
 \Pi_{\rm out}\mathcal U_{\rm amp}\iota_{\rm in}
 =\iota_{\rm out}
  (\iota_{\rm out}^\dagger\mathcal U_{\rm amp}\iota_{\rm in})
 =\iota_{\rm out}U_{\rm p}f(\Sigma).
\]
Using \(\iota_{\rm out}^\dagger\iota_{\rm out}=I\) and
\(U_{\rm p}^\dagger U_{\rm p}=I\), we obtain
\[
 \begin{aligned}
 &\bigl[(I-\Pi_{\rm out})\mathcal U_{\rm amp}\iota_{\rm in}\bigr]^\dagger
 \bigl[(I-\Pi_{\rm out})\mathcal U_{\rm amp}\iota_{\rm in}\bigr]\\
 &=I-\bigl[\iota_{\rm out}U_{\rm p}f(\Sigma)\bigr]^\dagger
       \bigl[\iota_{\rm out}U_{\rm p}f(\Sigma)\bigr]
 =I-f(\Sigma)^2=(I-f(\Sigma))(I+f(\Sigma)).
 \end{aligned}
\]
Taking operator norms and using \(\norm{I+f(\Sigma)}\leq2\), we obtain
\begin{equation}
 \norm{(I-\Pi_{\rm out})\mathcal U_{\rm amp}\iota_{\rm in}}
 =\sqrt{\norm{(I-f(\Sigma))(I+f(\Sigma))}}
 \leq\sqrt{2\norm{I-f(\Sigma)}}\leq2\eta.
 \label{eq:oaa-leakage}
\end{equation}
Since \(\iota_{\rm out}W_J\) lies in
\(\operatorname{ran}\Pi_{\rm out}\), splitting the error into its
\(\Pi_{\rm out}\) and \(I-\Pi_{\rm out}\) parts gives
\[
 \begin{aligned}
 &\Pi_{\rm out}
  (\mathcal U_{\rm amp}\iota_{\rm in}-\iota_{\rm out}W_J)
 =\iota_{\rm out}(U_{\rm p}f(\Sigma)-W_J),\\
 &(I-\Pi_{\rm out})
  (\mathcal U_{\rm amp}\iota_{\rm in}-\iota_{\rm out}W_J)
 =(I-\Pi_{\rm out})\mathcal U_{\rm amp}\iota_{\rm in}.
 \end{aligned}
\]
Therefore, the triangle inequality together with
\eqref{eq:oaa-output-error} and \eqref{eq:oaa-leakage} gives
\[
 \norm{\mathcal U_{\rm amp}\iota_{\rm in}-\iota_{\rm out}W_J}
 \leq\norm{U_{\rm p}f(\Sigma)-W_J}
    +\norm{(I-\Pi_{\rm out})\mathcal U_{\rm amp}\iota_{\rm in}}
 \leq(2\eta+2\eta^2)+2\eta
 \leq5\eta,
\]
which is \eqref{eq:oaa-isometry}.
\end{proof}

We next convert the isometry error \eqref{eq:oaa-isometry} into a
diamond-norm error.  Let \(\Tr_{\rm aux}\) denote the partial trace over all
output registers except \(\mathsf S\).  The OAA circuit implements the system
channel
\[
 \widetilde{\cE}(\rho):=\Tr_{\rm aux}\!\left[
 \mathcal U_{\rm amp}\iota_{\rm in}\rho\,
 \iota_{\rm in}^\dagger\mathcal U_{\rm amp}^\dagger\right].
\]
To relate the two errors, let \(V\) and \(W\) be isometries with the same
input and output spaces.  For any reference system \(\mathsf R\) and density
operator \(\rho\) on the input together with \(\mathsf R\),
\[
 \begin{aligned}
 &(V\otimes I_{\mathsf R})\rho(V^\dagger\otimes I_{\mathsf R})
 -(W\otimes I_{\mathsf R})\rho(W^\dagger\otimes I_{\mathsf R})\\
 &=[(V-W)\otimes I_{\mathsf R}]\rho(V^\dagger\otimes I_{\mathsf R})
 +(W\otimes I_{\mathsf R})\rho[(V-W)^\dagger\otimes I_{\mathsf R}].
 \end{aligned}
\]
Since \(\norm V=\norm W=\norm{\rho}_1=1\), the triangle inequality and
\(\norm{AXB}_1\leq\norm A\norm X_1\norm B\) imply
\[
 \begin{aligned}
 &\norm{(V\otimes I_{\mathsf R})\rho(V^\dagger\otimes I_{\mathsf R})
 -(W\otimes I_{\mathsf R})\rho(W^\dagger\otimes I_{\mathsf R})}_1\leq2\norm{V-W}.
 \end{aligned}
\]
For \(V=\mathcal U_{\rm amp}\iota_{\rm in}\) and
\(W=\iota_{\rm out}W_J\), tracing out the auxiliary registers gives
\(\widetilde{\cE}\) and \(\Phi_\delta^J\), respectively.  Partial trace does
not increase the trace norm, so taking the supremum over \(\mathsf R\) and
\(\rho\) gives
\begin{equation}
 \norm{\widetilde{\cE}-\Phi_\delta^J}_\diamond
 \leq2\norm{\mathcal U_{\rm amp}\iota_{\rm in}-\iota_{\rm out}W_J}
 \leq10\eta,
 \label{eq:oaa-channel-error}
\end{equation}
where the last step is \eqref{eq:oaa-isometry}.

We finally count oracle queries.  A call to \(\mathcal V\) uses at most
\(20q\) calls to \(S\), and a call to \(\mathcal V^\dagger\) uses at most
\(20q\) calls to \(S^\dagger\).  The OAA circuit \eqref{eq:oaa} contains two
calls to \(\mathcal V\), one call to \(\mathcal V^\dagger\), and two
oracle-independent reflections.  Since each call to \(S\) or \(S^\dagger\)
uses one oracle query, the OAA circuit uses at most \(60q\) oracle queries.

\subsection{Parameter choice}
\label{sec:parameters}

We choose \(q\) and \(J\) so that
\(\norm{\widetilde{\cE}-\Phi_\delta^J}_\diamond\leq\varepsilon/2\) and
\(\norm{\Phi_\delta^J-e^{t\cL}}_\diamond\leq\varepsilon/2\), respectively.
Set
\[
 \eta_q:=\sqrt\tau\left(\frac{C_0\tau}{q}\right)^q.
\]
If \(\alpha\delta\leq1/2\), \(J\geq8q\), \(q\geq C_0\tau\), and
\(\eta_q\leq1/8\), applying the OAA circuit \eqref{eq:oaa} to the encoding
from Proposition~\ref{prop:lcu-approximation} uses at most \(60q\) queries.
The channel-error bound \eqref{eq:oaa-channel-error} gives
\(\norm{\widetilde{\cE}-\Phi_\delta^J}_\diamond\leq10\eta_q\).
The next lemma chooses \(q\) so that
\(\eta_q\leq\varepsilon/20\) and \(q\) has the scaling stated in
Theorem~\ref{thm:main}.
\begin{lemma}[Choosing the polynomial degree]
\label{lem:factorial-inversion}
Let \(0<\varepsilon\leq1/2\), \(\tau>0\), and
\(L_\varepsilon=\log(20/\varepsilon)\).  There is a universal constant
\(c>0\) such that
\begin{equation}
 q=\left\lceil
 c\left(
 \tau+\frac{L_\varepsilon}{\log(e+L_\varepsilon/\tau)}
 \right)\right\rceil
 \label{eq:q-choice}
\end{equation}
satisfies \(q\geq C_0\tau\) and \(\eta_q\leq\varepsilon/20\).
If $2\tau>\varepsilon$, it also satisfies
\begin{equation}
 q=O\!\left(
 \tau+\frac{\log(1/\varepsilon)}
 {\log\!\bigl(e+\log(1/\varepsilon)/\tau\bigr)}
 \right).
 \label{eq:q-no-additive-constant}
\end{equation}
\end{lemma}

\begin{proof}
Set \(\alpha T=\tau\) in the parameter estimate of
\cite[proof of Theorem~9]{ChenEtAl2026}.  Repeating its argument with
\(C_0\) in place of \(12e\) and target error \(\varepsilon/20\) gives
\(q\geq C_0\tau\) and \(\eta_q\leq\varepsilon/20\) for a sufficiently large
universal \(c\).  If \(2\tau>\varepsilon\), the quantity inside the ceiling
in \eqref{eq:q-choice} is bounded below by a positive universal constant.
The ceiling therefore changes \(q\) by at most a constant factor, and
\(L_\varepsilon=\Theta(\log(1/\varepsilon))\) gives
\eqref{eq:q-no-additive-constant}.
\end{proof}

We now combine the amplification error \eqref{eq:oaa-channel-error} with
the discretization error \eqref{eq:discretization-error} and count the
resulting oracle queries.
\begin{proof}[Proof of the error and query bounds in
Theorem~\ref{thm:main}]
For \(2\tau\leq\varepsilon\), the identity channel already achieves the
required accuracy.  Indeed,
\[
 e^{t\cL}-\Id=\int_0^t\cL e^{s\cL}\dd{s}.
\]
Since \(e^{s\cL}\) is CPTP and has diamond norm one,
\[
 \norm{e^{t\cL}-\Id}_\diamond
 \leq\int_0^t\norm{\cL}_\diamond\dd{s}
 \leq2\alpha t=2\tau\leq\varepsilon,
\]
where we used the generator-norm bound
\eqref{eq:lindbladian-diamond-norm}.
For \(2\tau>\varepsilon\), choose \(q\) by
Lemma~\ref{lem:factorial-inversion} and choose an integer \(J\) satisfying
\begin{equation}
 J\geq
 \max\left\{2\tau,\,8q,\,20\tau^2/\varepsilon\right\}.
 \label{eq:J-choice}
\end{equation}
These choices ensure \(\alpha\delta=\tau/J\leq1/2\), \(J\geq8q\),
\(q\geq C_0\tau\), and \(\eta_q\leq\varepsilon/20\leq1/8\).
Proposition~\ref{prop:lcu-approximation}, the channel-error bound
\eqref{eq:oaa-channel-error}, and the discretization bound
\eqref{eq:discretization-error} therefore give
\begin{equation}
 \norm{\widetilde{\cE}-e^{t\cL}}_\diamond
 \leq10\eta_q+\frac{10\tau^2}{J}
 \leq\frac\varepsilon2+\frac\varepsilon2=\varepsilon.
 \label{eq:autonomous-error}
\end{equation}
The circuit uses at most \(60q\) queries.  Increasing \(J\) affects the
gate count and ancilla space, but not this query count, so
\eqref{eq:q-no-additive-constant} proves the query bound in Theorem~\ref{thm:main}.
\end{proof}

\section{Time-dependent Lindbladians}
\label{sec:time-dependent}

We extend the construction to \(H(s)\) and \(B(s)\) varying over
\(s\in[0,T]\), where \(T\geq0\) is the total simulation time.  The generator
at time \(s\) is
\begin{equation}
 \cL_s(\rho)=-i[H(s),\rho]
 +\Tr_{\mathsf E}\!\bigl(B(s)\rho B(s)^\dagger\bigr)
 -\frac12\{B(s)^\dagger B(s),\rho\},
 \qquad 0\leq s\leq T,
 \label{eq:td-generator}
\end{equation}
where the nonzero-label subspace \(\cH_{\rm L}\) is fixed.  We aim to
approximate the CPTP propagator
\[
 \cE(T,0):=\mathcal T\exp\!\left(\int_0^T\cL_s\dd{s}\right)
\]
to diamond-norm error \(\varepsilon\).

We assume the uniform bounds and Lipschitz condition
\begin{equation}
 \norm{H(s)}\leq\alpha_H,\qquad
 \norm{B(s)}\leq\alpha_B,\qquad
 \norm{\cL_s-\cL_r}_\diamond\leq L_{\cL}|s-r|.
 \label{eq:td-promises}
\end{equation}
Set \(\alpha=\alpha_H+\alpha_B^2\) and \(\tau=\alpha T\), as in the
time-independent setting.
For a grid \(t_j=jT/J\), we assume coherent time-indexed oracle access as
in~\cite{ChenEtAl2026}:
\begin{equation}
 U_H^{(J)}=
 \sum_{j=0}^{J-1}\proj{j}_{\mathsf T}\otimes U_{H,j},
 \qquad
 U_B^{(J)}=
 \sum_{j=0}^{J-1}\proj{j}_{\mathsf T}\otimes U_{B,j}.
 \label{eq:lind-t-oracles}
\end{equation}
For every \(j\), \(U_{H,j}\) is a Hermitian
\(\alpha_H\)-block encoding of \(H(t_j)\), and \(U_{B,j}\) is a projected
unitary encoding of \(B(t_j)/\alpha_B\).  These encodings use the same
ancilla registers \(\mathsf A_H,\mathsf A_B\) and the same Kraus-label
register \(\mathsf E\) for every \(j\), and satisfy
\begin{equation}
 (\bra0_{\mathsf A_H}\otimes I_{\mathsf S})U_{H,j}
 (\ket0_{\mathsf A_H}\otimes I_{\mathsf S})=H(t_j)/\alpha_H,
 \qquad
 (\bra0_{\mathsf A_B}\otimes P_{\rm L})U_{B,j}
 (\ket0_{\mathsf A_B}\ket0_{\mathsf E}\otimes I_{\mathsf S})
 =B(t_j)/\alpha_B.
 \label{eq:td-encodings}
\end{equation}
The coherent oracle family in \eqref{eq:lind-t-oracles} is available for
the discretization size \(J\) selected by the algorithm.
On the interval \([t_j,t_{j+1}]\), the algorithm uses the rational step of
Section~\ref{sec:rational} with \(H=H(t_j)\) and \(B=B(t_j)\).

\Needspace{20\baselineskip}
Applying the transducer, reuse, LCU, and amplification constructions to these
time-dependent rational steps gives the same asymptotic query bound as in the
time-independent setting.

\begin{theorem}[Time-dependent extension]
\label{thm:time-dependent}
Under \eqref{eq:td-promises} and the coherent time-indexed oracle model
\eqref{eq:lind-t-oracles}--\eqref{eq:td-encodings}, for every
\(0<\varepsilon\leq1/2\),
\(\cE(T,0)\) can be approximated within diamond-norm error
\(\varepsilon\) using
\[
 O\!\left(
 \tau+
 \frac{\log(1/\varepsilon)}
 {\log\!\bigl(e+\log(1/\varepsilon)/\tau\bigr)}
 \right)
\]
queries.
\end{theorem}

\begin{proof}
For \(2\tau\leq\varepsilon\), the identity channel already achieves the
required accuracy.  Indeed, Duhamel's formula gives
\[
 \cE(T,0)-\Id=\int_0^T\cE(T,s)\cL_s\dd{s}.
\]
The propagator \(\cE(T,s)\) is CPTP, and the generator-norm bound
\eqref{eq:lindbladian-diamond-norm} applies pointwise to \(\cL_s\).  Hence
\begin{equation}
 \norm{\cE(T,0)-\Id}_\diamond
 \leq\int_0^T\norm{\cL_s}_\diamond\dd{s}
 \leq2\tau\leq\varepsilon.
 \label{eq:td-identity-error}
\end{equation}
For \(2\tau>\varepsilon\), choose \(q\) by
Lemma~\ref{lem:factorial-inversion} and take \(J\geq\max\{2\tau,8q\}\);
we will increase \(J\) to control the discretization error.
Let \(\delta=T/J\) and \(t_j=j\delta\).  At each \(t_j\), use the
rational channel and Stinespring isometry of Section~\ref{sec:rational}
with \(H=H(t_j)\) and \(B=B(t_j)\), denoting them by \(\Phi_j\) and \(V_j\):
\[
 V_j:\cH_{\mathsf S}\longrightarrow
 \cH_{\mathsf E_j}\otimes\cH_{\mathsf S},\qquad
 \Phi_j(\rho)=\Tr_{\mathsf E_j}(V_j\rho V_j^\dagger).
\]
First, Proposition~\ref{prop:discretization} applies for every \(j\), so
\[
 \norm{\Phi_j-e^{\delta\cL_{t_j}}}_\diamond
 \leq10(\alpha\delta)^2.
\]
To compare \(\Phi_j\) with the actual evolution on this interval, we also
bound the error from replacing \(\cL_s\) by \(\cL_{t_j}\).
Duhamel's formula gives
\[
 \cE(t_{j+1},t_j)-e^{\delta\cL_{t_j}}
 =\int_{t_j}^{t_{j+1}}
 \cE(t_{j+1},s)(\cL_s-\cL_{t_j})
 e^{(s-t_j)\cL_{t_j}}\dd{s}.
\]
The factors \(\cE(t_{j+1},s)\) and \(e^{(s-t_j)\cL_{t_j}}\)
are CPTP and have diamond norm one, so
\[
 \norm{\cE(t_{j+1},t_j)-e^{\delta\cL_{t_j}}}_\diamond
 \leq\int_{t_j}^{t_{j+1}}L_{\cL}(s-t_j)\dd{s}
 =\frac12L_{\cL}\delta^2.
\]
Consequently,
\[
 \norm{\Phi_j-\cE(t_{j+1},t_j)}_\diamond
 \leq10(\alpha\delta)^2+\frac12L_{\cL}\delta^2.
\]
We telescope the difference between
the two \(J\)-step channels:
\[
 \begin{aligned}
 \Phi_{J-1}\circ\cdots\circ\Phi_0-\cE(T,0)=\sum_{j=0}^{J-1}
 \Phi_{J-1}\circ\cdots\circ\Phi_{j+1}
 \circ\bigl(\Phi_j-\cE(t_{j+1},t_j)\bigr)\circ\cE(t_j,0).
 \end{aligned}
\]
For \(j=J-1\), the composition to the left of
\(\Phi_j-\cE(t_{j+1},t_j)\) is the identity; also,
\(\cE(t_0,0)=\Id\).
The compositions \(\Phi_{J-1}\circ\cdots\circ\Phi_{j+1}\) and
\(\cE(t_j,0)\) are CPTP\@.  Taking diamond norms and summing gives
\begin{equation}
 \begin{aligned}
 \norm{\Phi_{J-1}\circ\cdots\circ\Phi_0
 -\cE(T,0)}_\diamond
 \leq J\left(10\alpha^2\delta^2+\frac12L_{\cL}\delta^2\right)=\frac{10\tau^2+L_{\cL}T^2/2}{J}.
 \end{aligned}
 \label{eq:td-product-error}
\end{equation}
We next implement \(\Phi_{J-1}\circ\cdots\circ\Phi_0\).
Set \(\Omega_j=U_{H,j}\oplus U_{B,j}\oplus U_{B,j}^\dagger\).
The step size, normalization factors, ancilla registers, and projector
\(P_{\rm L}\) are the same for every \(j\).  Hence
Lemma~\ref{lem:local-step} applies
with the same unitary \(G\) and a possibly different catalyst map
\(\gamma^{(j)}\), satisfying
\[
 G\bigl[(\ket0_{\mathsf E_j}\otimes\varphi)
   \oplus\Omega_j\gamma^{(j)}\varphi\bigr]
 =V_j\varphi\oplus\gamma^{(j)}\varphi,
 \qquad
 (\gamma^{(j)})^\dagger\gamma^{(j)}\preceq\alpha\delta I.
\]
For \(\psi\in\cH_{\mathsf S}\), use the spaces and the convention for
omitting fixed zero factors from Section~\ref{sec:full-product}, and set
\[
 \psi_0=\psi,\qquad
 \psi_{j+1}=(I_{\mathcal R_{<j}}\otimes V_j)\psi_j,\qquad
 x_j=(I_{\mathcal R_{<j}}\otimes\gamma^{(j)})\psi_j.
\]
Define \(W_J\psi=\psi_J\), \(\Gamma\psi=\bigoplus_jx_j\), and
\begin{equation}
 Q=I_{\cH_J}\oplus\bigoplus_j
 (I_{\mathcal R_{<j}}\otimes\Omega_j),\qquad
 S=(G_{J-1}\cdots G_0)Q.
 \label{eq:td-transducer-factorization}
\end{equation}
Let \(S_{11}:\cK\to\cK\) denote the private-to-private block of \(S\).
Repeating the global induction argument leading to
\eqref{eq:global-induction-state} gives
\[
 S(\psi\oplus\Gamma\psi)=W_J\psi\oplus\Gamma\psi.
\]
Here \(Q\) uses one time-indexed query.  Moreover,
\[
 \norm{\Gamma\psi}^2=\sum_j\norm{x_j}^2
 \leq J\alpha\delta\norm\psi^2=\tau\norm\psi^2.
\]
The same induction used to establish the Stinespring representation
\eqref{eq:product-stinespring-channel} shows that \(W_J\) induces
\(\Phi_{J-1}\circ\cdots\circ\Phi_0\).

The time-dependent polynomial bound for the private block in
Proposition~\ref{prop:td-private-polynomial} gives
\[
 \norm{\mathcal F_q(S_{11})}
 \leq\left(\frac{C_0\tau}{q}\right)^q
\]
whenever \(\alpha\delta\leq1/2\), \(J\geq8q\), and \(q\geq C_0\tau\).
Apply the reuse construction of Section~\ref{sec:reuse} to this transducer,
and define
\[
 \widetilde W_J:=\sum_{N=1}^{20q}\lambda_NP_N
\]
using the coefficients in \eqref{eq:lambda}.
Lemmas~\ref{lem:lcu-reuse-maps} and~\ref{lem:reuse-combination} give a
normalization-\(2\) projected unitary encoding
\((\mathcal V,\iota_{\rm in},\iota_{\rm out})\) of \(\widetilde W_J\), where
\(\mathcal V\) uses at most \(20q\) calls to \(S\).  The approximation bound
\eqref{eq:weighted-reuse-error}, the coefficient identity
\eqref{eq:reuse-combination}, and the smoothing bound
\eqref{eq:smoothing-bound} give
\[
 \begin{aligned}
 \norm{W_J-\widetilde W_J}
 &\leq\sqrt\tau\,\norm{\mathcal Q_q(S_{11})}\leq\sqrt\tau\,\norm{\mathcal F_q(S_{11})}
 \leq\eta_q.
 \end{aligned}
\]
For \(\eta_q\leq1/8\), applying Lemma~\ref{lem:oaa} to \(\widetilde W_J\)
gives a circuit using at most \(60q\) queries, with channel error
\begin{equation}
 \norm{\widetilde{\cE}
 -\Phi_{J-1}\circ\cdots\circ\Phi_0}_\diamond
 \leq2(5\eta_q)=10\eta_q.
 \label{eq:uniform-product-error}
\end{equation}
Combining \eqref{eq:uniform-product-error} with
\eqref{eq:td-product-error} gives
\begin{equation}
 \norm{\widetilde{\cE}-\cE(T,0)}_\diamond
 \leq
 10\eta_q
 +\frac{10\tau^2+L_{\cL}T^2/2}{J}.
 \label{eq:td-ledger}
\end{equation}
Finally, take
\[
 J\geq\max\left\{
 2\tau,\,8q,\,
 \frac{20\tau^2+L_{\cL}T^2}{\varepsilon}
 \right\}.
\]
The two terms in \eqref{eq:td-ledger} are then each at most \(\varepsilon/2\).
The query bound follows from \(60q\) and \eqref{eq:q-no-additive-constant}.
\end{proof}

The choice of \(J\) depends on \(L_{\cL}\), whereas the query count
\(60q\) does not.  Section~\ref{sec:gate-complexity} bounds the resulting
gate cost. The matching lower bound in \eqref{eq:intro-ham-bound}
already holds for the time-independent Hamiltonian subclass
\(B=0\)~\cite{LowChuang2019,GilyenEtAl2019}.

\begin{remark}[Lipschitz bounds from \(H\) and \(B\)]
Suppose that the Hamiltonian \(H\) and the stacked jump operator \(B\) are
Lipschitz continuous with constants \(L_H\) and \(L_B\), respectively; that
is, for all \(r,s\in[0,T]\),
\[
 \norm{H(s)-H(r)}\leq L_H|s-r|,
 \qquad
 \norm{B(s)-B(r)}\leq L_B|s-r|.
\]
Then \eqref{eq:td-promises} holds with
\(L_{\cL}=2L_H+4\alpha_BL_B\).  To verify this, write
\(B_s=B(s)\), \(B_r=B(r)\), and estimate
\begin{align*}
 \norm{[H(s)-H(r),\,\cdot\,]}_\diamond
 &\leq2L_H|s-r|,\\
 \norm{B_s(\,\cdot\,)B_s^\dagger
       -B_r(\,\cdot\,)B_r^\dagger}_\diamond
 &\leq(\norm{B_s}+\norm{B_r})\norm{B_s-B_r}
 \leq2\alpha_BL_B|s-r|,\\
 \norm{B_s^\dagger B_s-B_r^\dagger B_r}
 &\leq(\norm{B_s}+\norm{B_r})\norm{B_s-B_r}
 \leq2\alpha_BL_B|s-r|.
\end{align*}
The partial trace has diamond norm one, and
\[
 \left\|\frac12\{B_s^\dagger B_s-B_r^\dagger B_r,\,\cdot\,\}
 \right\|_\diamond\leq\norm{B_s^\dagger B_s-B_r^\dagger B_r}.
\]
Adding the three contributions gives
\(\norm{\cL_s-\cL_r}_\diamond\leq(2L_H+4\alpha_BL_B)|s-r|\).
\end{remark}

\section{Gate complexity}
\label{sec:gate-complexity}

A direct implementation of \(S=(G_{J-1}\cdots G_0)Q\) uses \(J\) local
updates and \(J\) registers \(\mathsf E_0,\ldots,\mathsf E_{J-1}\)
to encode the Kraus labels.
We reduce the gate and qubit counts by factoring the update product into
rotations and exchanges, using a compressed Kraus-label representation that
stores only the positions and values of nonzero labels, and
implementing the rotations with the factorization in
\cite[Lemma~10]{ChenEtAl2026Gate}.

Recall that \(\tau=\alpha T\), that \(L_{\cL}\) is the Lipschitz constant in
\eqref{eq:td-promises}, and that \(a\) is the oracle ancilla size from
Section~\ref{sec:model}.
The following theorem gives the resulting gate complexity.  The
time-independent bound is obtained by setting \(T=t\) and \(L_{\cL}=0\).

\begin{theorem}[Gate complexity]
\label{thm:gate-complexity}
Under the assumptions of Theorem~\ref{thm:time-dependent}, let
\(0<\varepsilon\leq1/2\) and set
\[
 \ell_\varepsilon=\log\!\frac{e(1+\tau+L_{\cL}T^2)}{\varepsilon}.
\]
The channel \(\cE(T,0)\) can be approximated to diamond-norm error \(\varepsilon\)
with the query complexity of Theorem~\ref{thm:time-dependent} and
\begin{equation}
 O\!\left[
 \left(\tau+
 \frac{\log(1/\varepsilon)}
 {\log\!\bigl(e+\log(1/\varepsilon)/\tau\bigr)}\right)
 \bigl(a+\ell_\varepsilon[\log(m+1)+\ell_\varepsilon]\bigr)
 \right]
\label{eq:gate-main-bound}
\end{equation}
one- and two-qubit gates, excluding oracle calls.
\end{theorem}

For one \(J\)-step time segment, the construction below makes \(O(q)\) calls
to the transducer, and each call updates \(O(q)\) stored Kraus labels.  This
produces a \(q^2\) term in the gate count.  In the final parameter choice, we
divide the evolution into time segments so that the degree \(q\) on each
segment is \(O(\ell_\varepsilon)\), while the total query count remains within
the bound of Theorem~\ref{thm:time-dependent}.

\subsection{Factoring the update product}
\label{sec:gate-factorization}

We first construct a circuit for the transducer \(S\) built from
\(J\) evolution steps of length \(\delta\).  Fix an integer \(q\geq1\),
take \(J\geq8q\) to be a power of two, and assume
\(0<\alpha\delta\leq1/2\).  Section~\ref{sec:gate-parameters} will
choose \(q,J,\delta\) and divide \([0,T]\) into time segments on which
this circuit is used.

The coefficients \(\kappa,\beta,\mu\) in \eqref{eq:local-scalars} and
\(d=(1-\mu)/(1+\mu)\) depend only on
\(\delta,\alpha_H,\alpha_B\), so the unitary \(G\) is the same
for all \(j\).  Time dependence enters only through the oracles
\(U_{H,j},U_{B,j}\) in \(Q\).

To separate changes to the Kraus labels from rotations, set
\begin{equation}
 G(0):=G\big|_{\kappa=\beta=0},\qquad
 R:=GG(0)^\dagger,\qquad G=RG(0).
 \label{eq:gate-local-factorization}
\end{equation}
Setting \(\kappa=\beta=0\) in \eqref{eq:reflection-matrix} makes
\(M=\operatorname{diag}(1,i,-1)\).  Set
\[
 \begin{aligned}
 P_\perp
 &=I-\proj0_{\mathsf A_B}\otimes P_{\rm L}\\
 &=\bigl[(I_{\mathsf A_B}-\proj0_{\mathsf A_B})\otimes I_{\mathsf E}
   +\proj0_{\mathsf A_B}\otimes\proj0_{\mathsf E}\bigr]\otimes I_{\mathsf S}.
 \end{aligned}
\]
Thus \(P_\perp\) projects onto the subspace defined by
\(\mathsf A_B\ne0\) or \(\mathsf E=0\).  Substituting
\(M=\operatorname{diag}(1,i,-1)\) into the explicit action of \(G\) in
\eqref{eq:G-full-action}, we obtain
\begin{equation}
 G(0)\begin{pmatrix}\ket0_{\mathsf E}\otimes s\\\ell\\h\\z\\b\end{pmatrix}
 =\begin{pmatrix}
 \ket0_{\mathsf E}\otimes s\\
 (\bra0_{\mathsf A_B}\otimes P_{\rm L})z\\
 ih\\-b\\
 \ket0_{\mathsf A_B}\otimes\ell+P_\perp z
 \end{pmatrix}.
 \label{eq:gate-zero-update}
\end{equation}
Here \(s\in\cH_{\mathsf S}\), \(\ell\in\cH_{\rm L}\otimes\cH_{\mathsf S}\),
and \((h,z,b)^\top\in\cH_{\rm priv}\).

We factor \(G(0)\) into a unitary acting only on \(\cH_{\rm priv}\) and an
exchange with \(\cH_{\rm L}\otimes\cH_{\mathsf S}\).  Define
\begin{equation}
 D\begin{pmatrix}h\\z\\b\end{pmatrix}
 =\begin{pmatrix}ih\\-b\\z\end{pmatrix}
 \quad\text{on }\cH_{\rm priv},
 \label{eq:gate-private-update}
\end{equation}
and, on \((\cH_{\mathsf E}\otimes\cH_{\mathsf S})\oplus\cH_{\rm priv}\), define
\begin{equation}
 E\begin{pmatrix}\ket0_{\mathsf E}\otimes s\\\ell\\h\\z\\b\end{pmatrix}
 =\begin{pmatrix}
 \ket0_{\mathsf E}\otimes s\\
 (\bra0_{\mathsf A_B}\otimes P_{\rm L})b\\
 h\\z\\
 \ket0_{\mathsf A_B}\otimes\ell+P_\perp b
 \end{pmatrix}.
 \label{eq:gate-exchange}
\end{equation}
The operator \(E\) exchanges \(\ell\) with the
\(\ket0_{\mathsf A_B}\otimes\cH_{\rm L}\otimes\cH_{\mathsf S}\)
part of \(b\) and fixes all other components, so
\(E^\dagger=E=E^{-1}\).
Substituting \((ih,-b,z)\) into \eqref{eq:gate-exchange} reproduces
\eqref{eq:gate-zero-update}:
\[
 G(0)=E\bigl(I_{\cH_{\mathsf E}\otimes\cH_{\mathsf S}}\oplus D\bigr).
\]

To determine the subspaces on which \(R=GG(0)^\dagger\) acts nontrivially,
first invert
\eqref{eq:gate-zero-update}:
\[
 G(0)^\dagger
 \begin{pmatrix}\ket0_{\mathsf E}\otimes s\\\ell\\h\\z\\b\end{pmatrix}
 =
 \begin{pmatrix}
  \ket0_{\mathsf E}\otimes s\\
  (\bra0_{\mathsf A_B}\otimes P_{\rm L})b\\
  -ih\\
  \ket0_{\mathsf A_B}\otimes\ell+P_\perp b\\
  -z
 \end{pmatrix}.
\]
Write \(h_0=(\bra0_{\mathsf A_H}\otimes I_{\mathsf S})h\) and
\(z_0=(\bra0_{\mathsf A_B}\bra0_{\mathsf E}\otimes I_{\mathsf S})z\).
Applying \(G\) to the vector on the right-hand side, the matrix \(M\) acts on
\(s,-ih_0,-z_0\) by \eqref{eq:G-full-action}.  Set
\[
 \begin{pmatrix}s'\\h_0'\\z_0'\end{pmatrix}
 =\bigl[M\operatorname{diag}(1,-i,-1)\otimes I_{\mathsf S}\bigr]
 \begin{pmatrix}s\\h_0\\z_0\end{pmatrix}.
\]
Evaluating \(R=GG(0)^\dagger\) using the explicit action of \(G\) in
\eqref{eq:G-full-action} gives
\begin{equation}
 R\begin{pmatrix}\ket0_{\mathsf E}\otimes s\\\ell\\h\\z\\b\end{pmatrix}
 =
 \begin{pmatrix}
  \ket0_{\mathsf E}\otimes s'\\\ell\\
  h+\ket0_{\mathsf A_H}\otimes(h_0'-h_0)\\
  z+\ket0_{\mathsf A_B}\ket0_{\mathsf E}\otimes(z_0'-z_0)\\b
 \end{pmatrix}.
 \label{eq:gate-rotation-action}
\end{equation}
On \((s,h_0,z_0)\), the unitary \(R\) acts by
\[
 M\operatorname{diag}(1,-i,-1)
 =\frac1{1+\mu}
 \begin{pmatrix}
  1-\mu&-2\kappa&-2\beta\\
  2\kappa&1-\kappa^2+\beta^2&-2\kappa\beta\\
  2\beta&-2\kappa\beta&1+\kappa^2-\beta^2
 \end{pmatrix}.
\]

To reduce the action of \(R\) to a two-dimensional rotation, choose
\begin{equation}
 V=\frac1{\sqrt\mu}\begin{pmatrix}\kappa&-\beta\\\beta&\kappa\end{pmatrix},
 \qquad
 \operatorname{Rot}(x)=
 \begin{pmatrix}x&-\sqrt{1-x^2}\\\sqrt{1-x^2}&x\end{pmatrix}
 \quad(0\leq x\leq1).
 \label{eq:gate-basis-change}
\end{equation}
Since \(V^\dagger(\kappa,\beta)^\top=(\sqrt\mu,0)^\top\),
\(d=(1-\mu)/(1+\mu)\), and
\(\sqrt{1-d^2}=2\sqrt\mu/(1+\mu)\), conjugating by \(1\oplus V\) gives
\begin{equation}
 (1\oplus V^\dagger)M\operatorname{diag}(1,-i,-1)(1\oplus V)
 =\operatorname{Rot}(d)\oplus1.
 \label{eq:gate-rotation-matrix}
\end{equation}

We next reorder the factors in \(G_{J-1}\cdots G_0\).
As in Section~\ref{sec:full-product}, extend \(G(0)\) and \(E\) to
\(\cH_{j+1}\oplus\cK_j\) by tensoring with
\(I_{\mathcal R_{<j}}\), and by the identity on its orthogonal
complement; denote these extensions by \(G_j(0)\) and \(E_j\).
Let \(D_j\) act as \(I_{\mathcal R_{<j}}\otimes D\) on \(\cK_j\)
and as the identity elsewhere, and set \(R_j=G_jG_j(0)^\dagger\).
The local factorization then reads
\[
 G_j=R_jE_jD_j.
\]
To move \(E_j\) past \(R_0,\ldots,R_{j-1}\), first identify the
subspaces on which \(E_j\) and \(R_j\) can act nontrivially.
Recall that
\begin{equation}
 \cH_{j+1}\cap\cH_j^\perp
 =\mathcal R_{<j}\otimes\cH_{\rm L}\otimes
 \biggl(\bigotimes_{r=j+1}^{J-1}\mathbb C\ket0_{\mathsf E_r}\biggr)
 \otimes\cH_{\mathsf S}.
 \label{eq:gate-last-label-space}
\end{equation}
Equations \eqref{eq:gate-exchange} and \eqref{eq:gate-rotation-action}
give
\begin{equation}
 \begin{aligned}
 \operatorname{ran}(E_j-I)
 &\subseteq(\cH_{j+1}\cap\cH_j^\perp)\oplus\cK_j,\\
 \operatorname{ran}(R_j-I)&\subseteq\cH_j\oplus\cK_j,
 \qquad \operatorname{ran}(D_j-I)\subseteq\cK_j.
 \end{aligned}
 \label{eq:gate-supports}
\end{equation}
For \(k<j\), \(\cH_k\subseteq\cH_j\) and \(\cK_k\perp\cK_j\), so
\[
 \bigl[(\cH_{j+1}\cap\cH_j^\perp)\oplus\cK_j\bigr]
 \perp(\cH_k\oplus\cK_k),\qquad
 (E_j-I)(R_k-I)=(R_k-I)(E_j-I)=0.
\]
Thus \(E_jR_k=R_kE_j\) for \(k<j\).  Similarly,
\(\cK_j\perp(\cH_{k+1}\oplus\cK_k)\) for \(j\ne k\) gives
\(D_jE_k=E_kD_j\) and \(D_jR_k=R_kD_j\).
Move all \(D_j\) to the right, then move each \(E_j\) past
\(R_0,\ldots,R_{j-1}\):
\begin{equation}
 Q':=(D_{J-1}\cdots D_0)Q,\qquad
 S=(R_{J-1}\cdots R_0)(E_{J-1}\cdots E_0)Q'.
 \label{eq:gate-global-factorization}
\end{equation}
The factor \(Q'\) is one query followed by \(D\) on each private
summand.  It fixes \(\cH_J\), and for \(\xi\in\mathcal R_{<j}\),
\begin{equation}
 Q'\left[\xi\otimes\begin{pmatrix}h\\z\\b\end{pmatrix}\right]
 =\xi\otimes\begin{pmatrix}iU_{H,j}h\\-U_{B,j}^\dagger b\\U_{B,j}z\end{pmatrix}.
 \label{eq:gate-query-action}
\end{equation}
Here \(U_{H,j}=U_H,U_{B,j}=U_B\) in the time-independent case.

It remains to implement the exchanges and rotations in
\eqref{eq:gate-global-factorization} without applying all \(J\) factors
individually.
The spaces \((\cH_{j+1}\cap\cH_j^\perp)\oplus\cK_j\) are mutually
orthogonal, and each \(E_j\) fixes their orthogonal complement.
Consequently
\begin{equation}
 E_{J-1}\cdots E_0
 =I_{\cH_0}\oplus\bigoplus_{j=0}^{J-1}
 \left.E_j\right|_{(\cH_{j+1}\cap\cH_j^\perp)\oplus\cK_j}.
 \label{eq:gate-zero-product}
\end{equation}
The next subsection implements this direct sum using the positions and values
of the nonzero Kraus labels.

\subsection{Compressing and updating the Kraus labels}
\label{sec:gate-labels}

We represent only the subspace reached by the simulation circuit, storing the
positions and values of its nonzero Kraus labels.
For a string \(\boldsymbol k=(k_0,\ldots,k_{J-1})\in\{0,\ldots,m\}^J\), set
\[
 w(\boldsymbol k):=\#\{0\leq r<J:k_r\ne0\}.
\]
Write \(w=w(\boldsymbol k)\) when the string is fixed, and let
\(j_1<\cdots<j_w\) be its nonzero positions, with \(j_0=-1\).

We first bound \(w\) on states reached by the circuit.  The query operation
\eqref{eq:gate-query-action} leaves the earlier labels
\(k_0,\ldots,k_{j-1}\) unchanged.  On \(\cH_j\oplus\cK_j\), the operator
\(R_j=I_{\mathcal R_{<j}}\otimes R\) also leaves these labels unchanged.
Moreover, \eqref{eq:gate-rotation-action} shows that \(R\) only mixes
\(s,h_0,z_0\), all with label zero at position \(j\); the later labels are
zero throughout \(\cH_j\oplus\cK_j\).  Thus \(Q'\), every \(R_j\), and
their adjoints preserve \(w\).  The exchange \(E_j\) in
\eqref{eq:gate-exchange} adds or removes the label at position \(j\), so a
call to \(S\) or \(S^\dagger\) changes \(w\) by at most one.

The state preparations, SELECT controls, and reflections in the OAA circuit
do not change the Kraus labels.  Since that circuit makes at most \(60q\)
calls to \(S\) or \(S^\dagger\), a state that starts with \(w=0\) satisfies
\begin{equation}
 w(\boldsymbol k)\leq v\leq60q
 \label{eq:gate-label-bound}
\end{equation}
after \(v\) calls.  Thus every input to a call has \(w\leq60q\).  Because the
exchange within that call may add one label, it suffices to encode strings
with \(w\leq60q+1\).

For \(w\leq60q+1\), define
\begin{equation}
 \ket{\operatorname{enc}(\boldsymbol k)}
 :=\ket{j_w,k_{j_w}}\cdots\ket{j_1,k_{j_1}}
       \otimes\ket{0,0}^{\otimes(60q+1-w)}.
 \label{eq:gate-label-encoding}
\end{equation}
For \(w=0\), the string of nonzero position--label pairs is empty.
Write \((\mathsf J_r,\mathsf L_r)\) for the \(r\)-th register pair,
where \(\mathsf J_r\) has \(\log_2J\) qubits and \(\mathsf L_r\)
has \(\lceil\log_2(m+1)\rceil\) qubits.
In the compressed representation \eqref{eq:gate-label-encoding},
\(\mathsf L_1=0\) if and only if \(w=0\).  When \(w>0\),
\(\mathsf J_1=j_w\) and \(\mathsf L_1=k_{j_w}>0\).
The \(60q+1\) position--label register pairs use
\begin{equation}
 (60q+1)\left[\log_2J+\left\lceil\log_2(m+1)\right\rceil\right]
 =O\!\left(q[\log J+\log(m+1)]\right)
 \label{eq:gate-label-space}
\end{equation}
qubits.

Following the flag-encoding convention of
\cite[Convention~1]{ChenEtAl2026}, we encode this subspace of
\(\cH_J\oplus\cK\) using registers \(\mathsf P,\mathsf T\) for the
direct-sum indices.
The values \(\mathsf P=0\) and
\(\mathsf P=1,\mathsf T=j\) identify the public summand \(\cH_J\) and the
private summand \(\cK_j\), respectively.  Within \(\cH_{\rm priv}\), a
two-qubit register \(\mathsf C\) uses the values \(00,01,10\) to identify
\(h,z,b\), respectively.
Following the notation of Section~\ref{sec:full-product} and
Figure~\ref{fig:gj-register-spaces}, write \(\mathsf F\) for the private copy
of \(\mathsf E\) in \(\cH_{\rm priv}\).  Thus
\(z,b\in\cH_{\mathsf A_B}\otimes\cH_{\mathsf F}\otimes\cH_{\mathsf S}\)
occupy \(\mathsf A_B\mathsf F\mathsf S\).
A public basis vector in this subspace is encoded as
\begin{equation}
 \left(\bigotimes_{r=0}^{J-1}\ket{k_r}_{\mathsf E_r}\right)\otimes\psi
 \ \longleftrightarrow\
 \ket{\operatorname{enc}(\boldsymbol k)}
 \ket0_{\mathsf P}\ket0_{\mathsf T}\ket{00}_{\mathsf C}
 \ket0_{\mathsf A_H\mathsf A_B\mathsf F}\otimes\psi.
 \label{eq:gate-public-encoding}
\end{equation}
For a vector in \(\cK_j\), extend its earlier Kraus labels to a length-\(J\)
string by setting \(k_r=0\) for \(r\geq j\).  Then
\begin{equation}
 \begin{aligned}
 &\left(\bigotimes_{r=0}^{j-1}\ket{k_r}_{\mathsf E_r}\right)
       \otimes\begin{pmatrix}h\\z\\b\end{pmatrix}\longleftrightarrow
 \ket{\operatorname{enc}(\boldsymbol k)}\ket1_{\mathsf P}\ket j_{\mathsf T}
 \otimes\begin{multlined}[t]
 \Bigl[\ket{00}_{\mathsf C}\otimes
 (I_{\mathsf A_H}\otimes\ket0_{\mathsf A_B\mathsf F}\otimes I_{\mathsf S})h\\
 +\ket{01}_{\mathsf C}\ket0_{\mathsf A_H}\otimes z
 +\ket{10}_{\mathsf C}\ket0_{\mathsf A_H}\otimes b\Bigr].
 \end{multlined}
 \end{aligned}
 \label{eq:gate-private-encoding}
\end{equation}
The register order after \(\ket{\operatorname{enc}(\boldsymbol k)}\) is
\(\mathsf P\mathsf T\mathsf C\mathsf A_H\mathsf A_B\mathsf F\mathsf S\).
In addition to the position--label registers in
\eqref{eq:gate-label-encoding}, the register states in
\eqref{eq:gate-public-encoding} and
\eqref{eq:gate-private-encoding} satisfy
\begin{equation}
 \begin{array}{c|l}
 (\mathsf P,\mathsf C)&\text{conditions}\\ \hline
 (0,00)&\mathsf T=\mathsf A_H=\mathsf A_B=\mathsf F=0\\
 (1,00)&j_w<\mathsf T<J,\quad\mathsf A_B=\mathsf F=0\\
 (1,01)\text{ or }(1,10)&j_w<\mathsf T<J,\quad
   \mathsf A_H=0,\quad0\leq\mathsf F\leq m.
 \end{array}
 \label{eq:gate-encoding-conditions}
\end{equation}
We call a register representation of the form
\eqref{eq:gate-public-encoding} or~\eqref{eq:gate-private-encoding} a valid
encoding if its position--label registers have the form
\eqref{eq:gate-label-encoding} and its remaining registers satisfy
\eqref{eq:gate-encoding-conditions}.
In these registers, the private-space unitary \(D\) in
\eqref{eq:gate-private-update} acts as
\[
 \ket{00}_{\mathsf C}\mapsto i\ket{00}_{\mathsf C},\qquad
 \ket{01}_{\mathsf C}\mapsto\ket{10}_{\mathsf C},\qquad
 \ket{10}_{\mathsf C}\mapsto-\ket{01}_{\mathsf C}
 \quad(\mathsf P=1).
\]
It uses \(O(1)\) gates, so \(Q'\) still uses one oracle query.

\paragraph{Implementing \(E_{J-1}\cdots E_0\).}
We first determine which summand in \eqref{eq:gate-zero-product} contains a
given valid encoding.  A public encoding with \(w>0\) lies in
\(\cH_{j_w+1}\cap\cH_{j_w}^\perp\) by
\eqref{eq:gate-last-label-space}, so only \(E_{j_w}\) can act nontrivially;
when \(w=0\), the state lies in \(\cH_0\) and every \(E_j\) acts trivially.
A private encoding with \(\mathsf T=j\) lies in \(\cK_j\), so only \(E_j\)
can act nontrivially.  Moreover,
\eqref{eq:gate-encoding-conditions} ensures that its stored Kraus-label
positions satisfy \(j_w<j\).

Fix \(0\leq j<J\) and a valid encoding in the \(b\)-component of \(\cK_j\).
Suppose that its Kraus-label string \(\boldsymbol k\) has \(w\leq60q\)
nonzero positions \(j_1<\cdots<j_w<j\) and that
\(\mathsf F=k\in\{1,\ldots,m\}\).  By \eqref{eq:gate-exchange}, \(E_j\)
exchanges this state with the public state obtained by adding \(k_j=k\).
The private state has \(w\) nonzero Kraus labels, whereas the public state
has \(w+1\) and largest nonzero position \(j\).
With \(\mathsf A_H=\mathsf A_B=0\) and the common system state left implicit,
the exchange is
\begin{equation}
 \begin{aligned}
 &\ket{j,k}\ket{j_w,k_{j_w}}\cdots\ket{j_1,k_{j_1}}
 \ket{0,0}^{\otimes(60q-w)}
 \ket{0,0,00,0}_{\mathsf P\mathsf T\mathsf C\mathsf F}\\
 &\qquad\longleftrightarrow
 \ket{j_w,k_{j_w}}\cdots\ket{j_1,k_{j_1}}
 \ket{0,0}^{\otimes(60q+1-w)}
 \ket{1,j,10,k}_{\mathsf P\mathsf T\mathsf C\mathsf F}.
 \end{aligned}
 \label{eq:gate-label-exchange}
\end{equation}

Writing \((x_r,y_r)\) for the contents of
\((\mathsf J_r,\mathsf L_r)\), let \(U\) perform the following cyclic shift,
conditioned on \((\mathsf P,\mathsf C)=(1,10)\):
\begin{equation}
 \left(\bigotimes_{r=1}^{60q+1}\ket{x_r,y_r}\right)
       \ket{t,f}_{\mathsf T\mathsf F}
 \longmapsto
 \ket{t,f}\left(\bigotimes_{r=1}^{60q}\ket{x_r,y_r}\right)
       \ket{x_{60q+1},y_{60q+1}}_{\mathsf T\mathsf F},
 \label{eq:gate-label-shift}
\end{equation}

Let \(X\) exchange
\begin{equation}
 \ket{0,00}_{\mathsf P\mathsf C}
 \longleftrightarrow\ket{1,10}_{\mathsf P\mathsf C}
 \quad\text{when}\quad
 \mathsf A_H=\mathsf A_B=\mathsf T=\mathsf F=0,
 \quad 1\leq\mathsf L_1\leq m,
 \label{eq:gate-type-exchange}
\end{equation}
and fix every other basis state.  Undo the shift with \(U^\dagger\) after
applying \(X\).  Since \(X^\dagger=X\),
\begin{equation}
 (U^\dagger XU)^\dagger=U^\dagger XU.
 \label{eq:gate-exchange-circuit}
\end{equation}
\begin{algorithm}[!htb]
\caption{Implementing \(E_{J-1}\cdots E_0\)}
\label{alg:gate-zero-product}
\algrenewcommand\algorithmicrequire{\textbf{Input:}}
\small
\begin{algorithmic}[1]
 \Require A valid encoding
 \eqref{eq:gate-public-encoding}--\eqref{eq:gate-private-encoding}
 with \(w\leq60q\).
 \For{\(r=1,\ldots,60q+1\)}
  \State Swap \((\mathsf J_r,\mathsf L_r)\) with
    \((\mathsf T,\mathsf F)\), controlled on
    \((\mathsf P,\mathsf C)=(1,10)\).
 \EndFor
 \State Apply \(X\) from \eqref{eq:gate-type-exchange}.
 \For{\(r=60q+1,\ldots,1\)}
  \State Swap \((\mathsf J_r,\mathsf L_r)\) with
    \((\mathsf T,\mathsf F)\), controlled on
    \((\mathsf P,\mathsf C)=(1,10)\).
 \EndFor
\end{algorithmic}
\end{algorithm}

We verify that the algorithm implements the exchange
\eqref{eq:gate-label-exchange}.  The direct-sum form
\eqref{eq:gate-zero-product} and \eqref{eq:gate-exchange-circuit} show that
both the target operation and the circuit are self-adjoint.  It therefore
suffices to check the \(\cK_j\to\cH_J\) direction.  Applying \(U\), \(X\), and
\(U^\dagger\) to
the private encoding in the second line of
\eqref{eq:gate-label-exchange} gives
\[
 \begin{aligned}
 &U^\dagger XU\bigl[
 \ket{j_w,k_{j_w}}\cdots\ket{j_1,k_{j_1}}
 \ket{0,0}^{\otimes(60q+1-w)}
 \ket{1,j,10,k}_{\mathsf P\mathsf T\mathsf C\mathsf F}\bigr]\\
 &=U^\dagger X\bigl[\ket{j,k}\ket{j_w,k_{j_w}}\cdots\ket{j_1,k_{j_1}}
 \ket{0,0}^{\otimes(60q-w)}
 \ket{1,0,10,0}_{\mathsf P\mathsf T\mathsf C\mathsf F}\bigr]\\
 &=U^\dagger\bigl[\ket{j,k}\ket{j_w,k_{j_w}}\cdots\ket{j_1,k_{j_1}}
 \ket{0,0}^{\otimes(60q-w)}
 \ket{0,0,00,0}_{\mathsf P\mathsf T\mathsf C\mathsf F}\bigr]\\
 &=\ket{j,k}\ket{j_w,k_{j_w}}\cdots\ket{j_1,k_{j_1}}
 \ket{0,0}^{\otimes(60q-w)}
 \ket{0,0,00,0}_{\mathsf P\mathsf T\mathsf C\mathsf F}.
 \end{aligned}
\]
The first two equalities use \eqref{eq:gate-label-shift} and
\eqref{eq:gate-type-exchange}, respectively.  In the last equality,
\(U^\dagger\) does not act because \(\mathsf P=0\).

For every other valid encoding, one of the following conditions makes \(X\)
act as the identity after \(U\):
\[
 \begin{array}{c|c}
 \text{input condition}&\text{condition after }U\text{ that makes }X=I\\ \hline
 \mathsf P=0,\ w=0&\mathsf L_1=0\\
 \mathsf P=1,\ \mathsf C=00\text{ or }01&\mathsf P=1,\ \mathsf C=00\text{ or }01\\
 \mathsf P=1,\ \mathsf C=10,\ \mathsf A_B\ne0&\mathsf A_B\ne0\\
 \mathsf P=1,\ \mathsf C=10,\ \mathsf F=0&\mathsf L_1=0
 \end{array}
\]
Thus \(U^\dagger XU=U^\dagger U=I\) on these inputs, matching
\(E_{J-1}\cdots E_0\) by \eqref{eq:gate-exchange}.

Algorithm~\ref{alg:gate-zero-product} uses \(2(60q+1)\) swaps of
\((\mathsf J_r,\mathsf L_r)\) with \((\mathsf T,\mathsf F)\), each costing
\(O(\log J+\log(m+1))\) one- and two-qubit gates when controlled by
\(\mathsf P\) and \(\mathsf C\).
Throughout this section, we implement zero tests with multi-controlled NOT
gates \cite[Corollary~7.4]{BarencoEtAl1995} and additions and comparisons
with ripple-carry circuits \cite[Secs.~4.1--4.3]{CuccaroEtAl2004}.
The comparisons and zero tests for \(X\) cost
\(O(a+\log J+\log(m+1))\) gates.
Adding the swaps and the single application of \(X\), we obtain
\begin{equation}
 O\!\left(q[\log J+\log(m+1)]+a\right)
 \label{eq:gate-label-cost}
\end{equation}
gates.

\subsection{Implementing the rotations}
\label{sec:gate-rotations}

It remains to implement \(R_{J-1}\cdots R_0\) in
\eqref{eq:gate-global-factorization}.  Since each \(R_j\) leaves the
Kraus-label string unchanged, fix \(\boldsymbol k\) and determine which
factors \(R_j\) can act nontrivially on states with this Kraus-label string.
By
\eqref{eq:gate-supports}, \(R_j\) can act nontrivially only on
\(\cH_j\oplus\cK_j\).  A public state with Kraus-label string \(\boldsymbol k\)
belongs to \(\cH_j\) only when \(j>j_w\), because all labels at positions
\(r\geq j\) must be zero.  A valid encoding of a state in \(\cK_j\)
also has \(j>j_w\) by \eqref{eq:gate-encoding-conditions}.  Hence every state
\(v\) with a valid encoding and Kraus-label string \(\boldsymbol k\) satisfies
\begin{equation}
 R_{J-1}\cdots R_0v=R_{J-1}\cdots R_{j_w+1}v.
 \label{eq:gate-active-rotations}
\end{equation}

For \(j>j_w\), we now express the action of \(R_j\) on \(s,h_0,z_0\) using
the encodings \eqref{eq:gate-public-encoding} and
\eqref{eq:gate-private-encoding}.  Since \(R_j\) leaves
\(\ket{\operatorname{enc}(\boldsymbol k)}\) unchanged, we omit the
position--label registers below.  The action on \((s,h_0,z_0)^\top\) is
\([M\operatorname{diag}(1,-i,-1)]\otimes I_{\mathsf S}\), so we also omit
the system register.  In \eqref{eq:gate-public-encoding}--
\eqref{eq:gate-private-encoding}, \(s\) is the public component, while
\(h_0\) and \(z_0\) are the \(h\)- and \(z\)-components of \(\cK_j\).
The zero registers displayed in these encodings, together with the
definitions of \(h_0\) and \(z_0\), imply
\(\mathsf A_H=\mathsf A_B=\mathsf F=0\) for all three components.  Omitting
these fixed zero registers, their remaining register states are
\begin{equation}
 \begin{gathered}
  \underbrace{\ket{0,0,00}_{\mathsf P\mathsf T\mathsf C}}_{s},
  \qquad
  \underbrace{\ket{1,j,00}_{\mathsf P\mathsf T\mathsf C}}_{h_0},
  \qquad
  \underbrace{\ket{1,j,01}_{\mathsf P\mathsf T\mathsf C}}_{z_0},\\[-2pt]
  \text{with the common }\ket{\operatorname{enc}(\boldsymbol k)},\ \mathsf S,
  \text{ and the fixed zero registers }\mathsf A_H,\mathsf A_B,\mathsf F
  \text{ omitted}.
 \end{gathered}
 \label{eq:gate-rotation-registers}
\end{equation}
We implement the basis change \(V\) from \eqref{eq:gate-basis-change} on
the subspace spanned by
\(\ket{00}_{\mathsf C},\ket{01}_{\mathsf C}\), controlled on
\(\mathsf P=1\) and \(\mathsf A_H=\mathsf A_B=\mathsf F=0\).
By \eqref{eq:gate-rotation-matrix},
\(V^\dagger R_jV=\operatorname{Rot}(d)\oplus1\) on \((s,h_0,z_0)\), so
\(z_0\) is fixed, while the action on \((s,h_0)\) is
\begin{equation}
 \begin{aligned}
 (V^\dagger R_jV)\ket0_{\mathsf P}\ket0_{\mathsf T}
 &=d\ket0_{\mathsf P}\ket0_{\mathsf T}
       +\sqrt{1-d^2}\ket1_{\mathsf P}\ket j_{\mathsf T},\\
 (V^\dagger R_jV)\ket1_{\mathsf P}\ket j_{\mathsf T}
 &=-\sqrt{1-d^2}\ket0_{\mathsf P}\ket0_{\mathsf T}
       +d\ket1_{\mathsf P}\ket j_{\mathsf T}.
 \end{aligned}
 \label{eq:gate-qubit-rotation}
\end{equation}

\paragraph{Partitioning the time-index range.}
We implement the rotations \eqref{eq:gate-qubit-rotation} for
\(j_w<j<J\) using the factorization in
\cite[Lemma~10]{ChenEtAl2026Gate} on intervals of power-of-two length.
Since \(J-j_w-1\) need not be a power of two, we partition
\([j_w+1,J)\) into intervals of length \(2^\nu\), with each left
endpoint a multiple of \(2^\nu\).  Define
\begin{equation}
 u_\nu=2^\nu\left\lceil\frac{j_w+1}{2^\nu}\right\rceil
 \quad(0\leq\nu\leq\log_2J),
 \qquad u_{\log_2J+1}=J.
 \label{eq:gate-interval-start}
\end{equation}
For \(\nu<\log_2J\), \(u_\nu/2^\nu\) is an integer, so
\[
 \begin{aligned}
 u_{\nu+1}=2^{\nu+1}\left\lceil\frac{u_\nu}{2^{\nu+1}}\right\rceil,\quad
 u_{\nu+1}-u_\nu
 =2^\nu\left(2\left\lceil\frac{u_\nu/2^\nu}{2}\right\rceil
 -\frac{u_\nu}{2^\nu}\right)\in\{0,2^\nu\}.
 \end{aligned}
\]
The conclusion \(u_{\nu+1}-u_\nu\in\{0,2^\nu\}\) also holds for
\(\nu=\log_2J\), since
\(u_\nu\in\{0,J\}\) and \(u_{\nu+1}=J\).
Thus each nonempty interval has length \(2^\nu\), with
\(2^\nu\mid u_\nu\).  Since \(u_0=j_w+1\),
\begin{equation}
 [j_w+1,J)=
 \bigsqcup_{\substack{0\leq\nu\leq\log_2J\\u_{\nu+1}>u_\nu}}
 [u_\nu,u_{\nu+1}).
 \label{eq:gate-interval-partition}
\end{equation}
The intervals occur in increasing index order as \(\nu\) increases;
Figure~\ref{fig:gate-intervals} illustrates the partition.
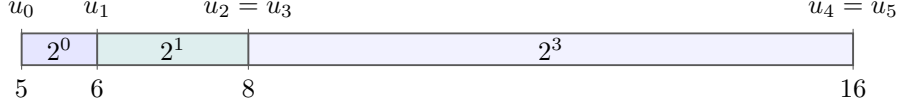
\begin{figure}[!htb]
 \centering
 \begin{tikzpicture}[x=1cm,y=1cm,font=\small]
  \fill[blue!10] (0,0) rectangle (1,0.42);
  \fill[teal!13] (1,0) rectangle (3,0.42);
  \fill[blue!5] (3,0) rectangle (11,0.42);
  \draw[black!65,line width=0.7pt] (0,0) rectangle (11,0.42);
  \foreach \x in {1,3} {
   \draw[black!65,line width=0.7pt] (\x,0)--(\x,0.42);
  }
  \node at (0.5,0.21) {$2^0$};
  \node at (2,0.21) {$2^1$};
  \node at (7,0.21) {$2^3$};
  \foreach \x/\n/\u in {0/5/{u_0},1/6/{u_1},3/8/{u_2=u_3},11/16/{u_4=u_5}} {
   \draw[black!65] (\x,-0.06)--(\x,0.48);
   \node[above] at (\x,0.5) {$\u$};
   \node[below] at (\x,-0.06) {$\n$};
  }
 \end{tikzpicture}
 \caption{Partition of $[5,16)$, with interval lengths shown inside
 the bars.  Repeated endpoints give empty intervals.}
 \label{fig:gate-intervals}
\end{figure}

\paragraph{Implementing one interval.}
Fix a nonempty interval \([u,u+L)\) in
\eqref{eq:gate-interval-partition}, so \(L=2^\nu\) and \(L\mid u\).
For the base interval \([0,L)\), \cite[Lemma~10]{ChenEtAl2026Gate}, with
\(c\) replaced by \(d\), gives
\[
 V^\dagger(R_{L-1}\cdots R_0)V
 =W_L^{\rm out}\operatorname{Rot}(d^L)_{\mathsf P}
              (W_L^{\rm in})^\dagger.
\]
Here \(W_L^{\rm in}\) and \(W_L^{\rm out}\) are \(\nu\)-qubit unitaries
acting on the lowest \(\nu\) bits of \(\mathsf T\).  For \(L=1\), they act
on no qubits and are the identity.
The unitaries \(W_L^{\rm in/out}\) are built from the single-qubit gates
\[
 M_{\rm in}(x)=\frac1{\sqrt{1+x^2}}
 \begin{pmatrix}x&1\\1&-x\end{pmatrix},\qquad
 M_{\rm out}(x)=\frac1{\sqrt{1+x^2}}
 \begin{pmatrix}1&x\\x&-1\end{pmatrix}.
\]
For \(\star\in\{\mathrm{in},\mathrm{out}\}\), the circuit
\(W_L^\star\) applies \(M_\star(d^{2^\ell})\) to \(\mathsf T_\ell\), from
\(\ell=\nu-1\) down to \(0\), conditioned on
\(\mathsf T_0=\cdots=\mathsf T_{\ell-1}=0\), with no control when
\(\ell=0\).

To apply this circuit on \([u,u+L)\), translate the private time labels
from \(u+r\) to \(r\), for \(0\leq r<L\).  Since \(L\mid u\), the lowest
\(\nu\) bits of \(u\) are zero.  Since \(0\leq r<2^\nu\), adding \(r\)
changes only these bits and produces no carry.  Therefore
\[
 u+r=u\mathbin{\mathrm{xor}}r,\qquad
 (u+r)\mathbin{\mathrm{xor}}u=r.
\]
Define the controlled XOR
\begin{equation}
 X_u\ket p_{\mathsf P}\ket j_{\mathsf T}
 =\ket p_{\mathsf P}\ket{j\mathbin{\mathrm{xor}}(pu)}_{\mathsf T}
 \quad(p\in\{0,1\}).
 \label{eq:gate-time-xor}
\end{equation}
The operator \(X_u\) fixes the public state
\(\ket0_{\mathsf P}\ket0_{\mathsf T}\) and maps
\(\ket1_{\mathsf P}\ket{u+r}_{\mathsf T}\) to
\(\ket1_{\mathsf P}\ket r_{\mathsf T}\).  The controlled XOR is
self-inverse, so \(X_u^\dagger=X_u\).  Conjugating the factorization for
\([0,L)\) by \(X_u\) therefore gives
\begin{equation}
 V^\dagger(R_{u+L-1}\cdots R_u)V
 =X_u W_L^{\rm out}\operatorname{Rot}(d^L)_{\mathsf P}
              (W_L^{\rm in})^\dagger X_u.
 \label{eq:gate-translated-interval}
\end{equation}
Index the bits of \(\mathsf T\) starting with the least significant bit
\(\mathsf T_0\).  The unitaries \(W_L^{\rm in},W_L^{\rm out}\) and
\(\operatorname{Rot}(d^L)_{\mathsf P}\) in
\eqref{eq:gate-translated-interval} are applied with the following controls:
\begin{equation}
 \begin{array}{c|l}
 W_L^{\rm in},W_L^{\rm out}
 &\mathsf P=1,\ \mathsf C=00,\
  \mathsf A_H=\mathsf A_B=\mathsf F=0,\
  \mathsf T_\nu=\cdots=\mathsf T_{\log_2J-1}=0,\\
 \operatorname{Rot}(d^L)_{\mathsf P}
 &\mathsf T=0,\ \mathsf C=00,\
  \mathsf A_H=\mathsf A_B=\mathsf F=0.
 \end{array}
 \label{eq:gate-interval-controls}
\end{equation}
The condition
\(\mathsf T_\nu=\cdots=\mathsf T_{\log_2J-1}=0\) in the first row of
\eqref{eq:gate-interval-controls} means that the value stored in \(\mathsf T\)
is less than \(2^\nu=L\), and hence lies in \([0,L)\).
Together, the gates between the two applications of \(X_u\) implement
\eqref{eq:gate-qubit-rotation} for \(u\leq j<u+L\) and act as the identity
on every other valid encoding.  In the latter case, the complete circuit
reduces to \(X_u^2=I\).

\paragraph{Combining the intervals.}
For fixed \(\boldsymbol k\), \eqref{eq:gate-active-rotations} shows that only
the factors \(R_j\) with \(j_w<j<J\) need to be implemented.  Let
\(\nu_1<\cdots<\nu_s\) be the values of \(\nu\) for which
\([u_\nu,u_{\nu+1})\) in \eqref{eq:gate-interval-partition} is nonempty.  On
states with label string \(\boldsymbol k\),
\[
 \begin{aligned}
 R_{J-1}\cdots R_{j_w+1}
 &=\bigl(R_{u_{\nu_s+1}-1}\cdots R_{u_{\nu_s}}\bigr)\cdots
   \bigl(R_{u_{\nu_1+1}-1}\cdots R_{u_{\nu_1}}\bigr)\\
 &=V\bigl[V^\dagger(R_{u_{\nu_s+1}-1}\cdots R_{u_{\nu_s}})V\bigr]\cdots
   \bigl[V^\dagger(R_{u_{\nu_1+1}-1}\cdots R_{u_{\nu_1}})V\bigr]V^\dagger.
 \end{aligned}
\]
The second equality inserts \(VV^\dagger=I\) between consecutive factors
on the first line.  Each bracket on the second line is implemented by
\eqref{eq:gate-translated-interval}.

To implement this \(\boldsymbol k\)-dependent factorization coherently, we
compute its interval endpoints from
\(\ket{\operatorname{enc}(\boldsymbol k)}\) in
\eqref{eq:gate-label-encoding}.  Its first position--label pair stores the
largest nonzero position, so
\[
 j_w+1=\begin{cases}0,&\mathsf L_1=0,\\
                     \mathsf J_1+1,&\mathsf L_1\ne0.
       \end{cases}
\]
From \(j_w+1\), reversibly compute
\[
 u_\nu=2^\nu\left\lceil\frac{j_w+1}{2^\nu}\right\rceil,
 \quad 0\leq\nu\leq\log_2J,
 \qquad u_{\log_2J+1}=J,
\]
as in \eqref{eq:gate-interval-start}.  None of
\(X_u,W_L^{\rm in},W_L^{\rm out}\), and
\(\operatorname{Rot}(d^L)_{\mathsf P}\) modifies the position--label pairs, so
the endpoints can be uncomputed afterward.

Applying \(V^\dagger\), followed by the right-hand side of
\eqref{eq:gate-translated-interval} with the controls in
\eqref{eq:gate-interval-controls} for each nonempty interval, and finally
\(V\), therefore implements \(R_{J-1}\cdots R_0\) on every valid encoding.

\subsection{The complete circuit and its gate count}
\label{sec:gate-parameters}

The preceding subsections provide circuits for the exchange and rotation
products in \eqref{eq:gate-global-factorization}.
Algorithm~\ref{alg:gate-call} assembles them with \(Q'\) in the order prescribed
by that factorization.  We denote the resulting circuit by \(\widehat S\).
After verifying that it
implements \(S\), we use it in the reuse unitaries \(U_N\), the LCU unitary
\(\mathcal V\), and the OAA circuit \(\mathcal U_{\rm amp}\).

\begin{algorithm}[!htb]
\caption{One call to \(S\) using the compressed Kraus-label representation}
\label{alg:gate-call}
\algrenewcommand\algorithmicrequire{\textbf{Input:}}
\small
\begin{algorithmic}[1]
 \Require A valid encoding
 \eqref{eq:gate-public-encoding}--\eqref{eq:gate-private-encoding},
 with \(w\leq60q\).
 \State Apply the query \(Q\).\label{line:gate-call-query}
 \State Conditioned on \(\mathsf P=1\), apply phase \(i\) for
        \(\mathsf C=00\), phase \(-1\) for \(\mathsf C=10\), then
        swap the values \(\mathsf C=01\) and \(\mathsf C=10\).
        \Comment{Together with the preceding query, this implements \(Q'\).}
        \label{line:gate-call-D}
 \State Apply Algorithm~\ref{alg:gate-zero-product}.
        \label{line:gate-call-exchange}
 \State From the updated first pair, set \(j_w+1=0\) if
        \(\mathsf L_1=0\), and set \(j_w+1=\mathsf J_1+1\) otherwise.
        \label{line:gate-call-left-endpoint}
 \State Compute the predicate \(\mathsf A_H=\mathsf A_B=\mathsf F=0\) into a
        work qubit, and apply \(V^\dagger\) conditioned on this qubit and
        \(\mathsf P=1\).
        \label{line:gate-call-zero-test}
 \For{\(\nu=0,\ldots,\log_2J\)}
  \State Compute \(u_\nu,u_{\nu+1}\) from \eqref{eq:gate-interval-start}.
         \label{line:gate-call-interval-endpoints}
  \If{\(u_{\nu+1}>u_\nu\)}
   \State Apply, in order, \(X_{u_\nu},\ (W_{2^\nu}^{\rm in})^\dagger,\
          \operatorname{Rot}(d^{2^\nu})_{\mathsf P},\
          W_{2^\nu}^{\rm out},\ X_{u_\nu}\),
          with the controls in \eqref{eq:gate-interval-controls}; use the work
          qubit for the test \(\mathsf A_H=\mathsf A_B=\mathsf F=0\).
          \label{line:gate-call-interval}
  \EndIf
  \State Reverse the computations of \(u_\nu,u_{\nu+1}\).
         \label{line:gate-call-interval-uncompute}
 \EndFor
 \State Apply \(V\) conditioned on the work qubit and \(\mathsf P=1\), clear
        the work qubit, and reverse the
        computation of \(j_w+1\).\label{line:gate-call-cleanup}
\end{algorithmic}
\end{algorithm}

We now verify the action of Algorithm~\ref{alg:gate-call}.
Lines~\ref{line:gate-call-query}--\ref{line:gate-call-D} apply \(Q\) followed
by the operation \(D\) in \eqref{eq:gate-private-update} on every private
summand, and hence implement \(Q'\) as specified in
\eqref{eq:gate-query-action}.  Line~\ref{line:gate-call-exchange} implements
\(E_{J-1}\cdots E_0\), and
lines~\ref{line:gate-call-left-endpoint}--\ref{line:gate-call-cleanup}
implement \(R_{J-1}\cdots R_0\) by \eqref{eq:gate-active-rotations} and
\eqref{eq:gate-translated-interval}.  Thus the algorithm applies
\(Q'\), \(E_{J-1}\cdots E_0\), and \(R_{J-1}\cdots R_0\) in this order.
On valid encodings with \(w\leq60q\), the factorization
\eqref{eq:gate-global-factorization} shows that \(\widehat S\) implements
\(S\) exactly.  Reversing Algorithm~\ref{alg:gate-call} implements
\(S^\dagger\).

\paragraph{Cost of one call to \(S\).}
The work qubit prepared in line~\ref{line:gate-call-zero-test} controls
\(V^\dagger\) in that line, \(V\) in line~\ref{line:gate-call-cleanup}, and
\((W_{2^\nu}^{\rm in})^\dagger\), \(W_{2^\nu}^{\rm out}\), and
\(\operatorname{Rot}(d^{2^\nu})_{\mathsf P}\) in
line~\ref{line:gate-call-interval}.  Preparing and clearing this qubit costs
\(O(a+\log(m+1))\) gates.  The circuits for \(W_L^{\rm in/out}\) in
line~\ref{line:gate-call-interval} have \(O(\nu)\) gates
\cite[Proposition~11]{ChenEtAl2026Gate}; controlling them changes this cost by
only a constant factor.  For each interval, computing and uncomputing the
endpoints in lines~\ref{line:gate-call-interval-endpoints}
and~\ref{line:gate-call-interval-uncompute}, together with the tests on
\(\mathsf T\) and the controlled XORs in line~\ref{line:gate-call-interval},
costs \(O(\log J)\) gates.  Computing and uncomputing \(j_w+1\) in
lines~\ref{line:gate-call-left-endpoint}
and~\ref{line:gate-call-cleanup} costs \(O(\log J+\log(m+1))\) gates, and the
controlled \(V,V^\dagger\) have constant cost.
Summing over at most
\(1+\log_2J\) intervals gives
\begin{equation}
 O\!\left(a+\log(m+1)+
           \sum_{\nu=0}^{\log_2J}(\nu+\log J)\right)
 =O\!\left(a+\log(m+1)+(\log J)^2\right)
 \label{eq:gate-rotation-cost}
\end{equation}
gates for \(R_{J-1}\cdots R_0\).

The only oracle-dependent operation in Algorithm~\ref{alg:gate-call} is the
query in line~\ref{line:gate-call-query}.  Thus a controlled \(\widehat S\) or
\(\widehat S^\dagger\) uses one
query.  Adding \eqref{eq:gate-label-cost},
\eqref{eq:gate-rotation-cost}, and the constant-size implementation of \(D\),
the gate count per call is
\begin{equation}
 O\!\left(a+q[\log J+\log(m+1)]+(\log J)^2\right).
 \label{eq:gate-one-call}
\end{equation}

\paragraph{Implementing \(\mathcal V\) and \(\mathcal U_{\rm amp}\).}
We next implement the controlled calls to \(S\) in the SELECT circuit of
Lemma~\ref{lem:lcu-reuse-maps} using \(\widehat S\).
In the reuse space \((\cH_{\mathsf M}\otimes\cH_J)\oplus\cK\),
\(\mathsf M\) labels the public copies of \(\cH_J\); on the shared private
space \(\cK\), this register is initialized to zero.  For
\(0\leq\ell<20q\), conditioned on a coefficient label
\(N\in\{1,\ldots,20q\}\) in \(\mathsf N\) and on \(\ell<N\), implement the
\(\ell\)-th call by
\begin{equation}
 \begin{aligned}
 &\mathsf M\leftarrow\mathsf M\mathbin{\mathrm{xor}}((1-\mathsf P)\ell),\\
 &\widehat S\text{ conditioned on }\mathsf M=0,\\
 &\mathsf M\leftarrow\mathsf M\mathbin{\mathrm{xor}}((1-\mathsf P)\ell).
 \end{aligned}
 \label{eq:gate-reuse-control}
\end{equation}
Before this call, the selected public component has
\((\mathsf P,\mathsf M)=(0,\ell)\), while the private component has
\((\mathsf P,\mathsf M)=(1,0)\).  The first XOR brings both components to
\(\mathsf M=0\), where \(\widehat S\) acts.  The second XOR uses the output
value of \(\mathsf P\) to return the public component to \(\mathsf M=\ell\)
and keep the private component at \(\mathsf M=0\).  Thus
\eqref{eq:gate-reuse-control} implements the action of \(S\) on
\((\ket\ell_{\mathsf M}\otimes\cH_J)\oplus\cK\) in the reuse circuit, without
changing the other public components.
Each call uses \(O(\log(q+1))\) gates for the comparison \(\ell<N\), the
zero test on \(\mathsf M\), and the two controlled XORs.

Conditioned on \(\mathsf N=N\) and \(\mathsf P=0\), SELECT applies \(F_N\)
before these calls and \(F_N^\dagger\) afterward.
The phase in \eqref{eq:lcu-select} is \(-1\) precisely when
\(N\) is even and \(4q\leq N\leq8q\), by \eqref{eq:lambda}
and the proof of Lemma~\ref{lem:reuse-combination}.
Implement \(F_N\) and \(F_N^\dagger\) using the exact uniform-superposition
state-preparation circuit of \cite[Algorithm~1]{ShuklaVedula2024}, controlled
by \(\mathsf N=N\) and \(\mathsf P=0\).
For \(F_\lambda\) and \(F_\lambda^\dagger\) in
\eqref{eq:lcu-unitary}, use the general state-preparation construction
of~\cite{MottonenEtAl2005}.

For amplification, we implement the reflections about the ranges of
\(\iota_{\rm in},\iota_{\rm out}\) in \eqref{eq:lcu-embeddings}.
Both isometries initialize \(\mathsf N=\mathsf M=0\), and
\(\varphi\) in the definition of \(\iota_{\rm out}\) ranges over all of
\(\cH_J\).  The public encoding \eqref{eq:gate-public-encoding} therefore
shows that
\(\operatorname{ran}\iota_{\rm out}\) is characterized by
\(\mathsf P=\mathsf T=\mathsf C=\mathsf A_H=\mathsf A_B=\mathsf F=0\),
with the position--label registers and \(\mathsf S\) unrestricted.
For \(\iota_{\rm in}\), all Kraus labels are also initialized to zero, which
is equivalent to \(\mathsf L_1=0\) in
\eqref{eq:gate-label-encoding}.  Thus, on valid encodings,
\begin{equation}
 \begin{array}{c|l}
 \operatorname{ran}\iota_{\rm out}
 &\mathsf N=\mathsf M=\mathsf P=\mathsf T=
   \mathsf A_H=\mathsf A_B=\mathsf F=0,\quad\mathsf C=00,\\
 \operatorname{ran}\iota_{\rm in}
 &\mathsf N=\mathsf M=\mathsf P=\mathsf T=
   \mathsf A_H=\mathsf A_B=\mathsf F=\mathsf L_1=0,\quad\mathsf C=00.
 \end{array}
 \label{eq:gate-label-projectors}
\end{equation}  Each reflection
\(2\Pi_{\rm in/out}-I\) applies phase \(+1\) on the corresponding range and
phase \(-1\) on its orthogonal complement.  It can therefore be implemented
using zero tests on
\(O(a+\log J+\log(m+1)+\log(q+1))\) qubits.
Since \(J\geq8q\), each reflection costs
\(O(a+\log J+\log(m+1))\) gates.

Starting with all Kraus labels zero,
\eqref{eq:gate-label-bound} ensures that the input to every call to \(S\) or
\(S^\dagger\) is supported on valid encodings with \(w\leq60q\).
Algorithm~\ref{alg:gate-call} and its reverse therefore implement every such
call exactly.  Since the register encodings
\eqref{eq:gate-public-encoding}--\eqref{eq:gate-private-encoding} leave
\(\mathsf S\) unchanged, the resulting OAA circuit induces the channel
\(\widetilde{\cE}\) analyzed in Section~\ref{sec:oaa}, whose error satisfies
\eqref{eq:oaa-channel-error}.

\paragraph{Gate count.}
The oracle-independent operations outside the calls to \(\widehat S\) and
\(\widehat S^\dagger\) have the following costs:
\begin{equation}
 \begin{array}{ll}
 F_\lambda,F_\lambda^\dagger\text{ in }\eqref{eq:lcu-unitary}&O(q),\\
 \text{controlled }F_N,F_N^\dagger\text{ in SELECT},\quad 1\leq N\leq20q
     &O(q\log(q+1)),\\
 \text{sign phase in SELECT and controls in }\eqref{eq:gate-reuse-control}
     &O(q\log(q+1)),\\
 \text{each reflection }2\Pi_{\rm in/out}-I
     &O(a+\log J+\log(m+1)).
 \end{array}
 \label{eq:gate-lcu-costs}
\end{equation}
The two calls to \(\mathcal V\) and one call to \(\mathcal V^\dagger\) in
\eqref{eq:oaa} use at most \(60q\) calls to \(\widehat S\) or
\(\widehat S^\dagger\) in total.
Multiplying \eqref{eq:gate-one-call} by \(60q\) and adding
\eqref{eq:gate-lcu-costs}, and using \(\log(q+1)=O(\log J)\) because
\(J\geq8q\), gives the gate count for one \(J\)-step time segment:
\begin{equation}
 O\!\left(q\bigl[a+q(\log J+\log(m+1))+(\log J)^2\bigr]\right).
 \label{eq:gate-segment-cost}
\end{equation}

\paragraph{Parameter choice and total cost.}
We now choose \(n,q,J\), apply the \(J\)-step circuit on each of \(n\) equal
time segments, and sum the costs to prove
Theorem~\ref{thm:gate-complexity}.

\Needspace{10\baselineskip}
\begin{proof}[Proof of Theorem~\ref{thm:gate-complexity}]
Recall \(\ell_\varepsilon=\log(e(1+\tau+L_{\cL}T^2)/\varepsilon)\).
For \(2\tau\leq\varepsilon\), the identity-channel bound
\eqref{eq:td-identity-error} shows that no gates are required.  Suppose that
\(2\tau>\varepsilon\).  To reduce the \(q^2\) contribution in
\eqref{eq:gate-segment-cost}, divide \([0,T]\) into \(n\) time segments
\([vT/n,(v+1)T/n]\), \(0\leq v<n\).
We allocate error \(\varepsilon/n\) to each segment and set
\begin{equation}
 \begin{aligned}
 n&=\max\!\left\{1,
              \left\lceil\frac{\tau}{\ell_\varepsilon}\right\rceil\right\},\\
 q&=\left\lceil c\left[
 \frac{\tau}{n}+
 \frac{\log(20n/\varepsilon)}
 {\log\!\bigl(e+(n/\tau)\log(20n/\varepsilon)\bigr)}
 \right]\right\rceil,
 \end{aligned}
 \label{eq:gate-segment-parameters}
\end{equation}
where \(c\) is the constant in Lemma~\ref{lem:factorial-inversion}.  This is
the choice in that lemma with \(\tau\) and \(\varepsilon\) replaced by
\(\tau/n\) and \(\varepsilon/n\), respectively.  Hence
\[
 \sqrt{\frac{\tau}{n}}
 \left(\frac{C_0\tau}{nq}\right)^q\leq\frac{\varepsilon}{20n},
\]
and \eqref{eq:uniform-product-error} bounds the amplification error by
\(\varepsilon/(2n)\).  To make the discretization error at most
\(\varepsilon/(2n)\), take \(J\) to be the smallest power
of two satisfying
\begin{equation}
 J\geq\max\!\left\{
 \frac{2\tau}{n},\ 8q,\
 \frac{20\tau^2+L_{\cL}T^2}{n\varepsilon}\right\},
 \qquad \delta=\frac{T}{nJ}.
 \label{eq:gate-grid}
\end{equation}

The first two bounds on \(J\) ensure \(\alpha\delta\leq1/2\) and
\(J\geq8q\), as required in Section~\ref{sec:gate-factorization}.  Applying
the time-dependent error bound \eqref{eq:td-ledger} on one segment, with
\(T\) and \(\tau\) replaced by \(T/n\) and \(\tau/n\), gives
\[
 10\sqrt{\frac{\tau}{n}}
       \left(\frac{C_0\tau}{nq}\right)^q
 +\frac{10\tau^2+L_{\cL}T^2/2}{n^2J}
 \leq\frac{\varepsilon}{2n}+
       \frac{\varepsilon}{2n}=\frac{\varepsilon}{n}.
\]

Composing the circuits for the \(n\) time segments in increasing order of
\(v\) gives the full simulation.  Since both the implemented and target
channels are CPTP, telescoping their compositions bounds the total error by
\(\sum_{v=0}^{n-1}\varepsilon/n=\varepsilon\).

We next estimate \(nq\), which determines the total query count.  Since
\(\ell_\varepsilon\geq1\), the definition of \(n\) gives
\[
 \tau/n\leq\ell_\varepsilon,\qquad n\leq1+\tau.
\]
It follows that
\(\log(20n/\varepsilon)\leq\log(20(1+\tau)/\varepsilon)
=O(\ell_\varepsilon)\).
The logarithm in the denominator of \(q\) in
\eqref{eq:gate-segment-parameters} is at least one, so
\[
 q\leq1+c\bigl[\tau/n+\log(20n/\varepsilon)\bigr]
 =O(\ell_\varepsilon).
\]

To bound \(nq\), consider separately \(n=1\) and \(n>1\).  If \(n=1\),
the asymptotic estimate \eqref{eq:q-no-additive-constant} gives
\[
 nq=q=O\!\left(\tau+
 \frac{\log(1/\varepsilon)}
 {\log\!\bigl(e+\log(1/\varepsilon)/\tau\bigr)}\right).
\]

If \(n>1\), then \(\tau/\ell_\varepsilon>1\), and the ceiling in
the definition of \(n\) gives
\[
 n<2\tau/\ell_\varepsilon,\qquad
 nq=O(n\ell_\varepsilon)=O(\tau).
\]
The two cases yield
\begin{equation}
 nq=O\!\left(\tau+
 \frac{\log(1/\varepsilon)}
 {\log\!\bigl(e+\log(1/\varepsilon)/\tau\bigr)}\right).
 \label{eq:gate-parameter-bounds}
\end{equation}

For the gate count, we also need \(\log J=O(\ell_\varepsilon)\).
By the minimal choice of the power of two in \eqref{eq:gate-grid},
\[
 \begin{aligned}
 J&<2\max\!\left\{
 \frac{2\tau}{n},\ 8q,\
 \frac{20\tau^2+L_{\cL}T^2}{n\varepsilon}\right\},\\
 \log J&=O\!\left(\log(q+1)+
       \log\!\frac{1+\tau^2+L_{\cL}T^2}{\varepsilon}\right)
       =O(\ell_\varepsilon).
 \end{aligned}
\]
Here \(q=O(\ell_\varepsilon)\), and
\(1+\tau^2+L_{\cL}T^2\leq(1+\tau+L_{\cL}T^2)^2\) bounds the second logarithm
by \(O(\ell_\varepsilon)\).

To simulate these \(n\) time segments, use the supplied oracles
\(U_H^{(nJ)}\) and \(U_B^{(nJ)}\) on the grid of \(nJ\) sample times.
On the \(v\)-th segment, step \(j\) corresponds to the oracle index
\(vJ+j\), since
\[
 \frac{vT}{n}+j\delta=\frac{(vJ+j)T}{nJ},
 \qquad 0\leq v<n,\quad0\leq j<J.
\]
Because \(J\) is a power of two, use the high-order bits \(\mathsf V\) of the
oracle's time-index register to store \(v\) and the low-order bits
\(\mathsf T\) to store \(j\):
\[
 \ket v_{\mathsf V}\ket j_{\mathsf T}
 =\ket{vJ+j}_{\mathsf V\mathsf T}.
\]
Within the circuit for the \(v\)-th segment, \(\mathsf T\) already stores
\(j\).  Initialize \(\mathsf V\) to \(\ket v\) before this circuit and clear
it afterward, at a cost of \(O(\log n)\) gates per segment.

The \(n\) time segments use at most \(60nq\) oracle queries, which satisfies
the query bound in Theorem~\ref{thm:time-dependent} by
\eqref{eq:gate-parameter-bounds}.  Writing \(N_{\rm gates}\) for the total
number of one- and two-qubit gates, summing \eqref{eq:gate-segment-cost} over
these segments and adding \(O(n\log(n+1))\) gates to initialize and clear
\(\mathsf V\) gives
\begin{equation}
 \begin{aligned}
 N_{\rm gates}
 &=O\!\left(nq\bigl[a+q(\log J+\log(m+1))+(\log J)^2\bigr]
                  +n\log(n+1)\right)\\
 &=O\!\left(nq\bigl[a+q(\log J+\log(m+1))
                  +(\log J)^2+\log(n+1)\bigr]\right)\\
 &=O\!\left(nq\bigl[a+\ell_\varepsilon
                   (\log(m+1)+\ell_\varepsilon)\bigr]\right).
 \end{aligned}
 \label{eq:gate-total-cost}
\end{equation}
Since \(q\geq1\),
\(n\log(n+1)=O(nq\log(n+1))\), which gives the second line of
\eqref{eq:gate-total-cost}.  The bounds \(q=O(\ell_\varepsilon)\),
\(\log J=O(\ell_\varepsilon)\), and
\(\log(n+1)=O(\ell_\varepsilon)\) then give the last line of
\eqref{eq:gate-total-cost}.  Substituting
\eqref{eq:gate-parameter-bounds} into \eqref{eq:gate-total-cost} gives
\begin{equation}
 N_{\rm gates}=O\!\left[
 \left(\tau+
 \frac{\log(1/\varepsilon)}
 {\log\!\bigl(e+\log(1/\varepsilon)/\tau\bigr)}\right)
 \bigl[a+\ell_\varepsilon(\log(m+1)+\ell_\varepsilon)\bigr]
 \right].
 \label{eq:gate-sharper-bound}
\end{equation}
Equation~\eqref{eq:gate-sharper-bound} is the gate bound
\eqref{eq:gate-main-bound} in Theorem~\ref{thm:gate-complexity}.
Setting \(T=t\) and \(L_{\cL}=0\) gives
\eqref{eq:main-gate-complexity} in Theorem~\ref{thm:main}.
\end{proof}

\section{Conclusion}

We have shown that general Lindbladian evolution can be simulated with a
query complexity that matches the Hamiltonian joint time--precision lower
bound in \eqref{eq:intro-ham-bound}.  Thus adding dissipative terms does not
worsen the optimal dependence on normalized time and precision in the oracle
model considered here.  The same query bound holds for Lipschitz-continuous
time-dependent Lindbladians under coherent time-indexed oracle access; the
Lipschitz constant affects the discretization and the gate complexity, but
not the query complexity.

It remains open whether locality, small commutators, or other properties of
the generator can reduce the simulation cost further while retaining the
optimal dependence on time and precision.

\section*{Acknowledgments}

We thank Jacob Kitchen for sharing his manuscript and for helpful discussions
on the relationship between the two works.

\paragraph{Statement on AI use.} The project was initiated and motivated by the authors. The main ideas underlying the work, including the central proof strategies, were generated by large language models. The human authors carefully studied, refined, and verified these ideas, and substantially rewrote and reorganized their presentation, adding the motivation and exposition necessary to communicate the arguments clearly. The final paper reflects the authors' own understanding of the results, and the authors take full responsibility for every claim, proof, and citation in it.

\clearpage
\appendix

\section{Proof of the polynomial bound for the private block}
\label{sec:private-block-proof}

This appendix proves Proposition~\ref{prop:private-polynomial}.
The proof compares \(S_{11}\) with its value at \(\delta=0\).
Recall that \(S_{11}:\cK\to\cK\) is the private-to-private
block of \(S\), where \(\cK=\bigoplus_{j=0}^{J-1}\cK_j\).
Define \(S_{11}(0)\) by setting \(\delta=0\), equivalently
\(\kappa=\beta=0\), in every \(G_j\), while keeping
\(\Omega=U_H\oplus U_B\oplus U_B^\dagger\) unchanged.
Write \(G_j(0)\) for \(G_j\) at these parameter values.
Set
\[
 \Delta:=S_{11}-S_{11}(0).
\]
Recall from \eqref{eq:Fq} that
\[
 p(\zeta)=\zeta^2(1+\zeta^2),\qquad
 \mathcal F_q(\zeta)=\left(\frac{p(\zeta)}2\right)^{2q}.
\]
The key property of \(S_{11}(0)\) is
\(p(S_{11}(0))=0\).  Once this identity is established, substituting
\(S_{11}=S_{11}(0)+\Delta\) into
\(p(S_{11})-p(S_{11}(0))\) expresses \(p(S_{11})\) as a sum of ordered
products, each containing at least one factor \(\Delta\).  Hence every
product in the expansion of
\(p(S_{11})^{2q}\) contains at least \(2q\) such factors.
We split \(\Delta\) into a no-jump part \(\Delta_0\) and a jump part
\(\Delta_1\).  Lemma~\ref{lem:private-block-bounds} proves
\(p(S_{11}(0))=0\) and bounds the norms of the maps
\((\Delta_s)_{kj}:\cK_j\to\cK_k\) for \(s=0,1\).
Proposition~\ref{prop:ordered-product-bound} then bounds ordered products of
\(S_{11}(0),\Delta_0,\Delta_1\).  Finally, we sum these bounds to prove
Proposition~\ref{prop:private-polynomial}.

We use the orthogonal projections associated with
\(\cK=\bigoplus_{j=0}^{J-1}\cK_j\).  For each \(j\), let \(\Pi_j\) be the
orthogonal projection from \(\cH_J\oplus\cK\) onto \(\cK_j\).  On \(\cK\),
these projections satisfy
\[
 \Pi_j\Pi_k=\delta_{jk}\Pi_j,\qquad
 \sum_{j=0}^{J-1}\Pi_j=I_{\cK}.
\]
For \(k>j\), we split the map
\(\Pi_k\Delta\Pi_j:\cK_j\to\cK_k\) according to whether the Kraus-label
register \(\mathsf E_j\) contains zero or a nonzero label.  Define
\((\Delta_0)_{kj}\) and \((\Delta_1)_{kj}\) by projecting \(\mathsf E_j\) onto
\(\mathbb C\ket0\) and \(\cH_{\rm L}\), respectively:
\begin{align}
 &(\Delta_0)_{kj}:=(\proj{0}_{\mathsf E_j}\otimes I)\Pi_k\Delta\Pi_j,
 \quad
 (\Delta_1)_{kj}:=[(I-\proj{0}_{\mathsf E_j})\otimes I]\Pi_k\Delta\Pi_j
 \quad(k>j),\nonumber\\
 \intertext{When \(k=j\), the time index is unchanged, and we set}
 &(\Delta_0)_{jj}:=\Pi_j\Delta\Pi_j,
 \qquad(\Delta_1)_{jj}:=0.
 \label{eq:correction-blocks}
\end{align}
For \(k<j\), set \((\Delta_0)_{kj}=(\Delta_1)_{kj}:=0\).
Set
\(\Delta_s:=\sum_{j,k=0}^{J-1}(\Delta_s)_{kj}\) for \(s=0,1\).

The norm bounds for \((\Delta_s)_{kj}\) will be expressed in terms of
\begin{equation}
 d=\frac{1-\mu}{1+\mu},\qquad
 \theta=1-d^2=\frac{4\mu}{(1+\mu)^2}.
 \label{eq:d-theta}
\end{equation}

The off-diagonal bounds below contain a factor \(\theta\) for
\(\Delta_0\) and \(\sqrt\theta\) for \(\Delta_1\).  We keep the two maps
separate because, for \(k>j\), \((\Delta_1)_{kj}\) creates a nonzero label in
\(\mathsf E_j\), whereas the output of \((\Delta_0)_{kj}\) has label zero in
\(\mathsf E_j\).  The resulting orthogonality between different
nonzero-label positions is used in Proposition~\ref{prop:ordered-product-bound}.

\begin{lemma}[Bounds for \(S_{11}(0)\), \(\Delta_0\), and \(\Delta_1\)]
\label{lem:private-block-bounds}
If \(\alpha\delta\leq1/2\), then
\begin{equation}
 p(S_{11}(0))=0,\qquad \norm{S_{11}(0)}\leq1,
 \label{eq:S11-zero-polynomial}
\end{equation}
\begin{equation}
 \Delta=\Delta_1+\Delta_0,
 \label{eq:private-decomposition}
\end{equation}
and, for \(k>j\),
\begin{align}
 \norm{(\Delta_1)_{kj}}
 &\leq\sqrt\theta\,d^{k-j-1},
 \label{eq:delta-one-block-bound}\\
 \norm{(\Delta_0)_{kj}}
 &\leq\theta d^{k-j-1},
 \qquad
 \norm{(\Delta_0)_{jj}}\leq1-d\leq\theta.
 \label{eq:delta-zero-block-bound}
\end{align}
\end{lemma}

\begin{proof}
\textbf{The action of \(G_r\).}
We first write the action of each
update \(G_r\) on its public and private components.  The definition of
\(S_{11}\) gives \(\Pi_kS_{11}\Pi_j=\Pi_kS\Pi_j\).  Recall the definition
of \(S\) in \eqref{eq:global-transducer-factorization}:
\begin{equation}
\label{eq:appendix-S-factorization}
 S=(G_{J-1}\cdots G_0)Q.
\end{equation}
Moreover, \(G_r\) acts as the identity outside
\(\cH_{r+1}\oplus\cK_r\).  Since \(\cK_j\) is orthogonal to this space for
\(r\ne j\),
\begin{equation}
 G_r\Pi_j=\Pi_jG_r=\Pi_j\qquad(r\ne j).
 \label{eq:inactive-update-on-Kj}
\end{equation}
For these calculations, we use the representation in
\eqref{eq:step-active-space}.  As illustrated in
Figure~\ref{fig:gj-register-spaces}, after
suppressing the fixed zero states of
\(\mathsf E_{r+1},\ldots,\mathsf E_{J-1}\),
and writing \(\mathsf F\) for the private copy of \(\mathsf E\),
\begin{align}
 \cH_{r+1}\oplus\cK_r
 &\cong\mathcal R_{<r}\otimes\left[
 (\cH_{\mathsf E_r}\otimes\cH_{\mathsf S})
 \oplus\cH_{{\rm priv},r}\right],\nonumber\\
 \cH_{{\rm priv},r}
 &\cong
 (\cH_{\mathsf A_H}\otimes\cH_{\mathsf S})
 \oplus
 (\cH_{\mathsf A_B}\otimes\cH_{\mathsf F}\otimes
  \cH_{\mathsf S})^{\oplus2},
 \qquad \cH_{\mathsf F}\cong\cH_{\mathsf E}.
 \label{eq:Gr-space-decomposition}
\end{align}
In the representation \eqref{eq:Gr-space-decomposition}, \(G_r\) acts as
\(I_{\mathcal R_{<r}}\otimes G\).  The factor \(\mathcal R_{<r}\) is
therefore unchanged and it is enough to specify the
action of \(G\) on
\((\cH_{\mathsf E_r}\otimes\cH_{\mathsf S})
 \oplus\cH_{{\rm priv},r}\).  We write a vector in this space as
\begin{equation}
 \underbrace{\begin{pmatrix}\ket0_{\mathsf E_r}\otimes s\\\ell\end{pmatrix}}
 _{\cH_{\mathsf E_r}\otimes\cH_{\mathsf S}}
 \oplus
 \underbrace{\begin{pmatrix}h\\z\\b\end{pmatrix}}
 _{\cH_{{\rm priv},r}},
 \label{eq:local-five-components}
\end{equation}
where \(s\in\cH_{\mathsf S}\),
\(\ell\in\cH_{\rm L}\otimes\cH_{\mathsf S}\),
\(h\in\cH_{\mathsf A_H}\otimes\cH_{\mathsf S}\), and
\(z,b\in\cH_{\mathsf A_B}\otimes\cH_{\mathsf F}\otimes\cH_{\mathsf S}\).
In \eqref{eq:local-five-components}, the two-entry vector is the public
component, split into its no-jump and jump parts.  The three-entry vector is
the private component, whose entries \(h,z,b\) belong, in order, to the
three private summands in \eqref{eq:Gr-space-decomposition}.

To apply \(G\) to \eqref{eq:local-five-components}, decompose \(h\) and
\(b\) into their zero-ancilla components and their orthogonal complements:
\[
 \begin{aligned}
 h_0&=(\bra0_{\mathsf A_H}\otimes I_{\mathsf S})h,
 &h^\perp&=h-(\ket0_{\mathsf A_H}\otimes I_{\mathsf S})h_0,\\
 b_0&=(\bra0_{\mathsf A_B}\bra0_{\mathsf F}\otimes I_{\mathsf S})b,
 &b^\perp&=b-(\ket0_{\mathsf A_B}\ket0_{\mathsf F}\otimes I_{\mathsf S})b_0.
 \end{aligned}
\]
The unitary matrix \(M\) from \eqref{eq:reflection-matrix} defines
\(s',h_0',b_0'\) by
\begin{equation}
 \begin{pmatrix}s'\\h_0'\\b_0'\end{pmatrix}
 =\left[\frac1{1+\mu}
 \begin{pmatrix}
  1-\mu&-2i\kappa&2\beta\\
  2\kappa&i(1-\kappa^2+\beta^2)&2\kappa\beta\\
  2\beta&-2i\kappa\beta&-1-\kappa^2+\beta^2
 \end{pmatrix}\otimes I_{\mathsf S}\right]
 \begin{pmatrix}s\\h_0\\b_0\end{pmatrix}.
 \label{eq:local-block-coefficients}
\end{equation}
To write the terms depending on \(z\), set
\begin{equation}
 P_\perp:=I-\proj0_{\mathsf A_B}\otimes P_{\rm L}.
 \label{eq:B-output-complement}
\end{equation}
The action of \(G\) in \eqref{eq:G-full-action} on the vector
\eqref{eq:local-five-components} is
\begin{equation}
 G\left[
 \begin{pmatrix}\ket0_{\mathsf E_r}\otimes s\\\ell\end{pmatrix}
 \oplus\begin{pmatrix}h\\z\\b\end{pmatrix}
 \right]
 =\begin{pmatrix}
  \ket0_{\mathsf E_r}\otimes s'\\
  (\bra0_{\mathsf A_B}\otimes P_{\rm L})z
 \end{pmatrix}\oplus\begin{pmatrix}
  (\ket0_{\mathsf A_H}\otimes I_{\mathsf S})h_0'+ih^\perp\\
  (\ket0_{\mathsf A_B}\ket0_{\mathsf F}\otimes I_{\mathsf S})b_0'-b^\perp\\
  \ket0_{\mathsf A_B}\otimes\ell+P_\perp z
 \end{pmatrix}.
 \label{eq:local-full-action}
\end{equation}
After restoring \(\mathcal R_{<r}\) and the suppressed zero states of
\(\mathsf E_{r+1},\ldots,\mathsf E_{J-1}\), the two direct-sum terms on the
right-hand side of \eqref{eq:local-full-action} lie in \(\cH_{r+1}\) and
\(\cK_r\), respectively.  Since \(\Pi_r\) projects onto \(\cK_r\), they are
the public and private components of the \(G_r\) output, respectively.

\smallskip
\noindent\textbf{The restrictions to each \(\cK_j\).}
We first compute \(\Pi_jS_{11}\Pi_j\), \(\Pi_jS_{11}(0)\Pi_j\), and
\((\Delta_0)_{jj}\).  We then show that \(S_{11}(0)\) preserves each
\(\cK_j\), so these restrictions determine its full action.
Let \(G_{11}\) denote the block of \(G\) from \(\cH_{\rm priv}\) to itself,
and let \(G_{11}(0)\) denote the corresponding block with
\(\kappa=\beta=0\).
In the representation \eqref{eq:Gr-space-decomposition} with \(r=j\), the
definition of the global query in \eqref{eq:global-query} gives
\begin{equation}
 \left.Q\right|_{\cK_j}=I_{\mathcal R_{<j}}\otimes\Omega.
 \label{eq:query-on-Kj}
\end{equation}
In particular, \(Q\cK_j\subseteq\cK_j\).  Equation
\eqref{eq:inactive-update-on-Kj} shows that any consecutive product not
containing \(G_j\) acts as the identity on \(\cK_j\).  Hence
\begin{equation}
 \begin{aligned}
 \Pi_j(G_bG_{b-1}\cdots G_a)
 &=(G_bG_{b-1}\cdots G_a)\Pi_j=\Pi_j
 &&\bigl(0\leq a\leq b<j\ \text{or}\ j<a\leq b<J\bigr),\\
 (G_rG_{r-1}\cdots G_0)Q\Pi_j
 &=(G_rG_{r-1}\cdots G_0)\Pi_jQ\Pi_j=Q\Pi_j
 &&(0\leq r<j).
 \end{aligned}
 \label{eq:remove-inactive-updates}
\end{equation}
Substituting \eqref{eq:remove-inactive-updates} into the
factorization \eqref{eq:appendix-S-factorization} removes every
\(G_r\) with \(r\ne j\).  In the representation
\eqref{eq:Gr-space-decomposition} with \(r=j\), the private-to-private block
of \(G_j\) is \(I_{\mathcal R_{<j}}\otimes G_{11}\).  Together with
\eqref{eq:query-on-Kj}, this gives
\begin{equation}
 \begin{aligned}
 \Pi_jS_{11}\Pi_j
 &=\Pi_j(G_{J-1}\cdots G_{j+1})G_j
   (G_{j-1}\cdots G_0)Q\Pi_j\\
 &=\Pi_jG_jQ\Pi_j
 =I_{\mathcal R_{<j}}\otimes G_{11}\Omega.
 \end{aligned}
 \label{eq:S11-diagonal-block}
\end{equation}
Setting \(\kappa=\beta=0\) in \eqref{eq:S11-diagonal-block} gives
\(\Pi_jS_{11}(0)\Pi_j
=I_{\mathcal R_{<j}}\otimes G_{11}(0)\Omega\).  By the definition of
\((\Delta_0)_{jj}\) in \eqref{eq:correction-blocks},
\begin{equation}
 \begin{aligned}
 (\Delta_0)_{jj}
 &=\Pi_j[S_{11}-S_{11}(0)]\Pi_j\\
 &=I_{\mathcal R_{<j}}\otimes[(G_{11}-G_{11}(0))\Omega].
 \end{aligned}
 \label{eq:delta-zero-diagonal-form}
\end{equation}
At \(\kappa=\beta=0\), the matrix in
\eqref{eq:local-block-coefficients} is \(\operatorname{diag}(1,i,-1)\).
The \(\cH_{\rm priv}\to\cH_{\rm priv}\) block of
\eqref{eq:local-full-action}, obtained by setting \(s=\ell=0\), is therefore
\begin{equation}
 G_{11}(0)
 =(iI_{\cH_{\mathsf A_H}\otimes\cH_{\mathsf S}})\oplus
 \begin{pmatrix}0&-I\\P_\perp&0\end{pmatrix}.
 \label{eq:local-zero-private}
\end{equation}
We next prove that \(\Pi_kS_{11}(0)\Pi_j=0\) whenever \(k\ne j\).
With \(\kappa=\beta=0\) and \(\ell=h=z=b=0\),
\eqref{eq:local-full-action} gives
\(s'=s\), \(h_0'=b_0'=0\), and zero output in the private space.
In the representation \eqref{eq:Gr-space-decomposition}, the subspace
\(\cH_r\) corresponds to
\(\mathcal R_{<r}\otimes
(\mathbb C\ket0_{\mathsf E_r}\otimes\cH_{\mathsf S})\), which is the first
public row in \eqref{eq:local-five-components}.
Therefore \(G_r(0)\zeta=\zeta\) for every \(\zeta\in\cH_r\).
For \(r>j\), the inclusion \(\cH_{j+1}\subseteq\cH_r\) shows that
\(G_r(0)\) fixes \(\cH_{j+1}\).  Specializing
\eqref{eq:inactive-update-on-Kj} to \(G_r(0)\) shows that it also fixes
\(\cK_j\).  Hence
\[
 \left.G_{J-1}(0)\cdots G_{j+1}(0)
 \right|_{\cH_{j+1}\oplus\cK_j}=I.
\]
Specializing the second relation in \eqref{eq:remove-inactive-updates} to
\(G_r(0)\) gives
\[
 [G_{j-1}(0)\cdots G_0(0)]Q\Pi_j=Q\Pi_j.
\]
Moreover, \(G_j(0)Q\cK_j\subseteq\cH_{j+1}\oplus\cK_j\), while \(\Pi_k\)
vanishes on this space for \(k\ne j\).  It follows that
\begin{equation}
 \begin{aligned}
 \Pi_kS_{11}(0)\Pi_j
 &=\Pi_k[G_{J-1}(0)\cdots G_{j+1}(0)]G_j(0)
 [G_{j-1}(0)\cdots G_0(0)]Q\Pi_j\\
 &=\Pi_kG_j(0)Q\Pi_j
 =0\qquad(k\ne j).
 \end{aligned}
 \label{eq:zero-private-offdiagonal}
\end{equation}
Equation~\eqref{eq:zero-private-offdiagonal} implies
\(S_{11}(0)\cK_j\subseteq\cK_j\) for every \(j\).  Its action on \(\cK_j\)
is obtained by substituting \eqref{eq:local-zero-private} and
\(\Omega=U_H\oplus U_B\oplus U_B^\dagger\) into
\eqref{eq:S11-diagonal-block}:
\begin{equation}
 \begin{aligned}
 \left.S_{11}(0)\right|_{\cK_j}
 &=\left.(\Pi_jS_{11}(0)\Pi_j)\right|_{\cK_j}
 =I_{\mathcal R_{<j}}\otimes G_{11}(0)\Omega\\
 &=I_{\mathcal R_{<j}}\otimes
 \left[(iU_H)\oplus
 \begin{pmatrix}0&-U_B^\dagger\\P_\perp U_B&0\end{pmatrix}\right],
 \quad 0\leq j<J.
 \end{aligned}
 \label{eq:zero-private}
\end{equation}
To prove \(p(S_{11}(0))=0\), we compute the squares of the two
direct-sum components in \eqref{eq:zero-private}:
\begin{align}
 (iU_H)^2&=-I,\nonumber\\
 \begin{pmatrix}0&-U_B^\dagger\\P_\perp U_B&0\end{pmatrix}^{\!2}
 &=-\operatorname{diag}\!\left(
 U_B^\dagger P_\perp U_B,P_\perp\right).
 \label{eq:zero-private-square}
\end{align}
Both \(P_\perp\) and \(U_B^\dagger P_\perp U_B\) are projections.  Thus
\begin{align*}
 p(iU_H)&=(-I)(I-I)=0,\\
 p\!\left(\begin{pmatrix}0&-U_B^\dagger\\P_\perp U_B&0\end{pmatrix}\right)
 &=-\operatorname{diag}(U_B^\dagger P_\perp U_B,P_\perp)
 \operatorname{diag}(I-U_B^\dagger P_\perp U_B,I-P_\perp)=0.
\end{align*}
Applying \(p\) to each direct-sum component in \eqref{eq:zero-private},
\[
 p(S_{11}(0))
 =\bigoplus_{j=0}^{J-1}I_{\mathcal R_{<j}}\otimes(0\oplus0)=0.
\]
To bound \(\norm{S_{11}(0)}\), compute
\[
 \begin{pmatrix}0&-U_B^\dagger\\P_\perp U_B&0\end{pmatrix}^{\!\dagger}
 \begin{pmatrix}0&-U_B^\dagger\\P_\perp U_B&0\end{pmatrix}
 =
 \operatorname{diag}\!\left(
 U_B^\dagger P_\perp U_B,I\right)\preceq I.
\]
Together with \(\norm{iU_H}=1\), the direct-sum form
\eqref{eq:zero-private} gives
\begin{equation}
 \norm{S_{11}(0)}\leq1.
 \label{eq:A0-contraction}
\end{equation}

We next bound \(\norm{(\Delta_0)_{jj}}\).
By \eqref{eq:local-scalars},
\(\mu\leq\alpha\delta/2\leq1/4\) under the assumption
\(\alpha\delta\leq1/2\).  Hence
\begin{equation}
 0<d\leq1,
 \qquad
 1-d=\frac{2\mu}{1+\mu}\leq\theta
 \leq4\mu\leq2\alpha\delta.
 \label{eq:d-theta-bounds}
\end{equation}
For \(s=\ell=0\), the private output in \eqref{eq:local-full-action} is
\(G_{11}(h,z,b)^\top\).  By
\eqref{eq:local-zero-private}, subtracting
\(G_{11}(0)(h,z,b)^\top=(ih,-b,P_\perp z)^\top\) gives
\[
 (G_{11}-G_{11}(0))\begin{pmatrix}h\\z\\b\end{pmatrix}
 =\begin{pmatrix}
  \ket0_{\mathsf A_H}\otimes(h_0'-ih_0)\\
  \ket0_{\mathsf A_B}\ket0_{\mathsf F}\otimes(b_0'+b_0)\\0
 \end{pmatrix}.
\]
The differences \(h_0'-ih_0\) and \(b_0'+b_0\) follow from the last
two rows of \eqref{eq:local-block-coefficients} with \(s=0\):
\begin{align}
 \begin{pmatrix}h_0'-ih_0\\b_0'+b_0\end{pmatrix}
 &=\left\{\left[
 \frac1{1+\mu}
 \begin{pmatrix}
  i(1-\kappa^2+\beta^2)&2\kappa\beta\\
  -2i\kappa\beta&-1-\kappa^2+\beta^2
 \end{pmatrix}
 -\begin{pmatrix}i&0\\0&-1\end{pmatrix}\right]\otimes I_{\mathsf S}\right\}
 \begin{pmatrix}h_0\\b_0\end{pmatrix}
 \nonumber\\
 &=\frac{2}{1+\mu}\left[
 \begin{pmatrix}\kappa\\\beta\end{pmatrix}
 \begin{pmatrix}-i\kappa&\beta\end{pmatrix}\otimes I_{\mathsf S}\right]
 \begin{pmatrix}h_0\\b_0\end{pmatrix}.
 \label{eq:diagonal-correction}
\end{align}
The embeddings \(h_0\mapsto\ket0_{\mathsf A_H}\otimes h_0\) and
\(b_0\mapsto\ket0_{\mathsf A_B}\ket0_{\mathsf F}\otimes b_0\)
are isometries.
By \eqref{eq:delta-zero-diagonal-form} and the unitarity of \(\Omega\),
\[
 \begin{aligned}
 \norm{(\Delta_0)_{jj}}
 &=\norm{G_{11}-G_{11}(0)}
 =\frac{2}{1+\mu}\left\|
 \begin{pmatrix}\kappa\\\beta\end{pmatrix}
 \begin{pmatrix}-i\kappa&\beta\end{pmatrix}\right\|\\
 &=\frac{2}{1+\mu}
   \sqrt{\kappa^2+\beta^2}\sqrt{\kappa^2+\beta^2}
 =\frac{2\mu}{1+\mu}=1-d\leq\theta.
 \end{aligned}
\]

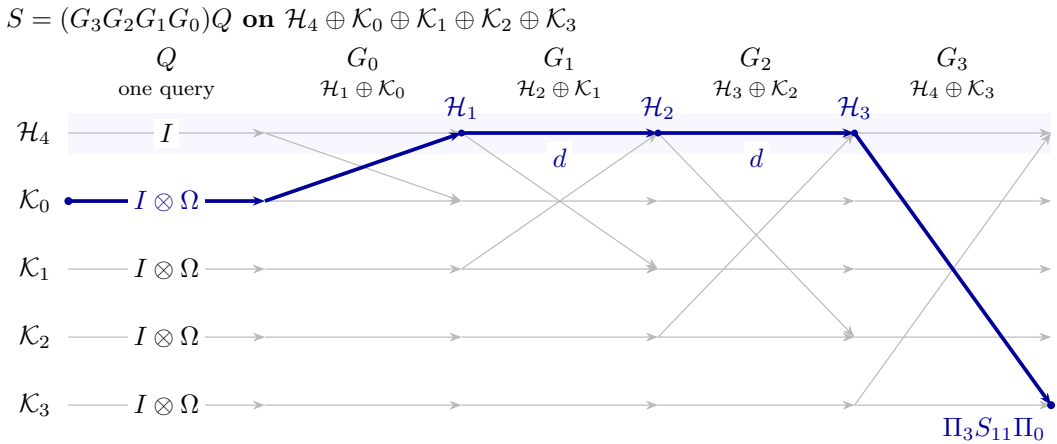
\begin{figure}[!b]
 \centering
 \begin{tikzpicture}[
  x=1cm,y=0.9cm,font=\small,
  >={Stealth[length=1.6mm,width=1.2mm]},
  every node/.style={inner sep=2pt},
  background path/.style={draw=black!27,line width=0.45pt,->},
  selected path/.style={draw=blue!60!black,line width=1.4pt,->},
  stage label/.style={font=\small,align=center}
 ]
 % Rows are direct-sum subspaces, not tensor-product wires.
 \node[anchor=west,font=\small\bfseries] at (-0.9,1.65)
  {$S=(G_3G_2G_1G_0)Q$ on
   $\cH_4\oplus\cK_0\oplus\cK_1\oplus\cK_2\oplus\cK_3$};
 \fill[blue!3] (0,-0.30) rectangle (13,0.30);
 \foreach \y/\name in {0/{\cH_4},-1/{\cK_0},-2/{\cK_1},-3/{\cK_2},-4/{\cK_3}} {
  \node[anchor=east] at (-0.15,\y) {$\name$};
  \foreach \x/\xx in {0/2.6,2.6/5.2,5.2/7.8,7.8/10.4,10.4/13} {
   \draw[background path] (\x,\y)--(\xx,\y);
  }
 }
 \node[stage label] at (1.3,0.88) {$Q$\\[-1pt]\scriptsize one query};
 \foreach \x/\r/\h in {3.9/0/1,6.5/1/2,9.1/2/3,11.7/3/4} {
  \node[stage label] at (\x,0.88)
   {$G_{\r}$\\[-1pt]\scriptsize $\cH_{\h}\oplus\cK_{\r}$};
 }
 \node[fill=white] at (1.3,0) {$I$};
 \foreach \y in {-1,-2,-3,-4} {
  \node[fill=white] at (1.3,\y) {$I\otimes\Omega$};
 }
 % Each G_r couples only the indicated H subspace and K_r.
 \foreach \x/\xx/\y in {2.6/5.2/-1,5.2/7.8/-2,7.8/10.4/-3,10.4/13/-4} {
  \draw[background path] (\x,0)--(\xx,\y);
  \draw[background path] (\x,\y)--(\xx,0);
 }
 % The component selected by Pi_3 S_11 Pi_0.
 \draw[selected path] (0,-1)--(2.6,-1);
 \node[fill=white,text=blue!60!black] at (1.3,-1) {$I\otimes\Omega$};
 \draw[selected path] (2.6,-1)--(5.2,0);
 \draw[selected path] (5.2,0)--(7.8,0)
  node[midway,below=3pt,fill=white,text=blue!60!black] {$d$};
 \draw[selected path] (7.8,0)--(10.4,0)
  node[midway,below=3pt,fill=white,text=blue!60!black] {$d$};
 \draw[selected path] (10.4,0)--(13,-4);
 \foreach \x/\h in {5.2/1,7.8/2,10.4/3} {
  \fill[blue!60!black] (\x,0) circle (1.3pt);
  \node[above=3pt,text=blue!60!black] at (\x,0) {$\cH_{\h}$};
 }
 \fill[blue!60!black] (0,-1) circle (1.5pt);
 \fill[blue!60!black] (13,-4) circle (1.5pt);
 \node[anchor=north east,text=blue!60!black] at (13,-4.10)
  {$\Pi_3S_{11}\Pi_0$};
 \end{tikzpicture}
 \caption{The sequential product for $J=4$, with rows representing
 direct-sum subspaces.  The highlighted arrows show the only sequence of
 subspaces contributing to $\Pi_3S_{11}\Pi_0$: it passes through
 $\cH_1\subseteq\cH_2\subseteq\cH_3$, where $G_1$ and $G_2$ each act by
 multiplication by $d$.  The updates $G_{r+1},\ldots,G_{J-1}$ act as the
 identity on $\cK_r$.}
 \label{fig:sequential-blocks}
\end{figure}

\smallskip
\noindent\textbf{Maps between different \(\cK_j\).}
First suppose that \(k<j\).  The factorization
\eqref{eq:appendix-S-factorization} gives
\begin{equation}
 \begin{aligned}
 \Pi_kS_{11}\Pi_j
 &=\Pi_k(G_{J-1}\cdots G_j)(G_{j-1}\cdots G_0)Q\Pi_j\\
 &=\Pi_k(G_{J-1}\cdots G_j)Q\Pi_j\\
 &=\Pi_kQ\Pi_j\\
 &=0.
 \end{aligned}
 \label{eq:S11-lower-blocks-zero}
\end{equation}
The second and third equalities use the second and first identities in
\eqref{eq:remove-inactive-updates}, respectively.  The last equality follows
from \eqref{eq:query-on-Kj} and \(\cK_k\perp\cK_j\).

For \(k>j\), the same factorization gives
\[
 \begin{aligned}
 \Pi_kS_{11}\Pi_j
 &=\Pi_k(G_{J-1}\cdots G_{k+1})G_k\cdots G_j
   (G_{j-1}\cdots G_0)Q\Pi_j\\
 &=\Pi_kG_kG_{k-1}\cdots G_jQ\Pi_j.
 \end{aligned}
\]
Here the second equality applies the two identities in
\eqref{eq:remove-inactive-updates} with indices \(k\) and \(j\), respectively.
For \(r<k\), the first identity in
\eqref{eq:remove-inactive-updates}, applied with \(j=r\), shows that a
component in \(\cK_r\) cannot subsequently reach \(\cK_k\):
\[
 \Pi_kG_kG_{k-1}\cdots G_{r+1}\Pi_r
 =\Pi_k\Pi_r=0.
\]
For each \(j\leq r<k\), inserting \(I=(I-\Pi_r)+\Pi_r\) therefore gives
\[
 \begin{aligned}
 \Pi_kG_k\cdots G_jQ\Pi_j
 &=\Pi_kG_k\cdots G_{r+1}\bigl[(I-\Pi_r)+\Pi_r\bigr]
 G_r\cdots G_jQ\Pi_j\\
 &=\Pi_kG_k\cdots G_{r+1}(I-\Pi_r)G_r\cdots G_jQ\Pi_j.
 \end{aligned}
\]
At this point the output of \(G_r\) lies in
\(\cH_{r+1}\oplus\cK_r\).  Hence \(I-\Pi_r\) removes its \(\cK_r\)
component and leaves its \(\cH_{r+1}\) component.
Repeating the insertion for \(r=j,\ldots,k-1\) gives
\begin{equation}
 \begin{aligned}
 \Pi_kS_{11}\Pi_j
 &=\Pi_kG_k\cdots G_jQ\Pi_j\\
 &=\underbrace{\Pi_kG_k}_{\cH_k\to\cK_k}
 \underbrace{\bigl[(I-\Pi_{k-1})G_{k-1}\bigr]\cdots
 \bigl[(I-\Pi_{j+1})G_{j+1}\bigr]}_{\cH_{j+1}\to\cH_k}
 \underbrace{(I-\Pi_j)G_jQ\Pi_j}_{\cK_j\to\cH_{j+1}}.
 \end{aligned}
 \label{eq:sequential-block-factorization}
\end{equation}
Figure~\ref{fig:sequential-blocks} illustrates
\eqref{eq:sequential-block-factorization} for \(J=4\).

First consider \((I-\Pi_j)G_jQ\Pi_j:\cK_j\to\cH_{j+1}\), the rightmost
factor in \eqref{eq:sequential-block-factorization}.  For
\(\xi\in\mathcal R_{<j}\) and \(\chi\in\cH_{{\rm priv},j}\),
\eqref{eq:query-on-Kj} maps \(\xi\otimes\chi\) to
\(\xi\otimes\Omega\chi\); write \(\Omega\chi=(h,z,b)^\top\).  This input
has no public component.
Setting \(s=\ell=0\) in \eqref{eq:local-full-action} and using the first row of
\eqref{eq:local-block-coefficients} therefore gives
\begin{equation}
 (I-\Pi_j)G_jQ(\xi\otimes\chi)
 =\xi\otimes\begin{pmatrix}
  \dfrac{2}{1+\mu}\ket0_{\mathsf E_j}\otimes(-i\kappa h_0+\beta b_0)\\
  (\bra0_{\mathsf A_B}\otimes P_{\rm L})z
 \end{pmatrix}.
 \label{eq:local-from-K}
\end{equation}
By the definitions of \(h_0,b_0\) and \(\Omega\chi\),
\[
 -i\kappa h_0+\beta b_0
 =\begin{pmatrix}
  -i\kappa(\bra0_{\mathsf A_H}\otimes I_{\mathsf S})&0&
  \beta(\bra0_{\mathsf A_B}\bra0_{\mathsf F}\otimes I_{\mathsf S})
 \end{pmatrix}\Omega\chi.
\]

For the other two factors in
\eqref{eq:sequential-block-factorization}, we need
\((I-\Pi_r)G_r:\cH_r\to\cH_{r+1}\) for \(j<r<k\) and
\(\Pi_kG_k:\cH_k\to\cK_k\).
An input in \(\cH_r\) has zero Kraus label in \(\mathsf E_r\) and no
\(\cK_r\) component, so set \(\ell=h=z=b=0\) in
\eqref{eq:local-full-action} and use the first column of
\eqref{eq:local-block-coefficients}.  For
\(\xi\in\mathcal R_{<r}\)
and \(j<r\leq k\), this gives
\begin{equation}
 G_r\left\{\xi\otimes\left[
 \begin{pmatrix}\ket0_{\mathsf E_r}\otimes s\\0\end{pmatrix}
 \oplus\begin{pmatrix}0\\0\\0\end{pmatrix}\right]\right\}
 =\xi\otimes\left[
 d\begin{pmatrix}\ket0_{\mathsf E_r}\otimes s\\0\end{pmatrix}
 \oplus\frac{2}{1+\mu}\begin{pmatrix}
  \kappa(\ket0_{\mathsf A_H}\otimes s)\\
  \beta(\ket0_{\mathsf A_B}\ket0_{\mathsf F}\otimes s)\\0
 \end{pmatrix}\right].
 \label{eq:local-H-action}
\end{equation}
Projecting \eqref{eq:local-H-action} onto \(\cH_{r+1}\) gives the middle
factors of \eqref{eq:sequential-block-factorization}.  Since the displayed
inputs span \(\cH_r\),
\begin{equation}
 (I-\Pi_r)G_r\zeta=d\zeta,
 \qquad \zeta\in\cH_r.
 \label{eq:local-from-H}
\end{equation}
Iterating \eqref{eq:local-from-H} along
\(\cH_{j+1}\subseteq\cdots\subseteq\cH_k\) gives
\[
 \bigl[(I-\Pi_{k-1})G_{k-1}\bigr]\cdots
 \bigl[(I-\Pi_{j+1})G_{j+1}\bigr]\zeta
 =d^{k-j-1}\zeta,
 \qquad \zeta\in\cH_{j+1}.
\]
Taking the \(\cK_k\) component of \eqref{eq:local-H-action} with \(r=k\)
gives the leftmost factor \(\Pi_kG_k:\cH_k\to\cK_k\).

\par\Needspace{13\baselineskip}
We now separate the zero- and nonzero-label components in \(\mathsf E_j\),
which define \((\Delta_0)_{kj}\) and \((\Delta_1)_{kj}\).  For
\(j<r\leq k\), both \(G_r\) and \(\Pi_r\) leave \(\mathsf E_j\) unchanged,
so
\[
 G_r(\proj0_{\mathsf E_j}\otimes I)=(\proj0_{\mathsf E_j}\otimes I)G_r,
 \qquad
 \Pi_r(\proj0_{\mathsf E_j}\otimes I)=(\proj0_{\mathsf E_j}\otimes I)\Pi_r.
\]
For \(k>j\), \eqref{eq:zero-private-offdiagonal} gives
\(\Pi_k\Delta\Pi_j=\Pi_kS_{11}\Pi_j\).  Combining the definition
\eqref{eq:correction-blocks} with
\eqref{eq:sequential-block-factorization}, and then commuting
\(\proj0_{\mathsf E_j}\otimes I\) through the factors between \(G_j\) and
the final projection \(\Pi_k\),
gives
\begin{equation}
 \begin{aligned}
 (\Delta_0)_{kj}
 &=(\proj0_{\mathsf E_j}\otimes I)\Pi_kG_k
 \bigl[(I-\Pi_{k-1})G_{k-1}\bigr]\cdots
 \bigl[(I-\Pi_j)G_j\bigr]Q\Pi_j\\
 &=\Pi_kG_k\bigl[(I-\Pi_{k-1})G_{k-1}\bigr]\cdots
 \bigl[(I-\Pi_{j+1})G_{j+1}\bigr]
 (\proj0_{\mathsf E_j}\otimes I)(I-\Pi_j)G_jQ\Pi_j.
 \end{aligned}
 \label{eq:output-projection}
\end{equation}
Replacing \(\proj0_{\mathsf E_j}\) by \(I-\proj0_{\mathsf E_j}\) gives
the corresponding expression for \((\Delta_1)_{kj}\).  Thus the upper and
lower entries in \eqref{eq:local-from-K} contribute to
\((\Delta_0)_{kj}\) and \((\Delta_1)_{kj}\), respectively.

To write the product in \eqref{eq:output-projection} explicitly, we also
record how the output of the rightmost factor is embedded into \(\cH_k\).
As this output propagates from \(\cH_{j+1}\) to \(\cH_k\),
\eqref{eq:local-from-H} shows that the intervening registers
\(\mathsf E_{j+1},\ldots,\mathsf E_{k-1}\) remain in \(\ket0\).  Define
\[
 \iota_{kj}:\cH_{\mathsf E_j}\otimes\cH_{\mathsf S}
 \longrightarrow
 \left(\bigotimes_{r=j}^{k-1}\cH_{\mathsf E_r}\right)\otimes\cH_{\mathsf S}
\]
by
\begin{equation}
 \iota_{kj}(\ket{\ell}_{\mathsf E_j}\otimes\psi)
 =\ket{\ell}_{\mathsf E_j}\otimes
   \ket{0}_{\mathsf E_{j+1}\cdots\mathsf E_{k-1}}\otimes\psi,
   \qquad 0\leq\ell\leq m.
 \label{eq:zero-label-embedding}
\end{equation}
When \(k=j+1\), \(\iota_{kj}=I\).  Since \(\iota_{kj}\) leaves
\(\mathsf E_j\) unchanged,
\begin{equation}
 (\proj{0}_{\mathsf E_j}\otimes I)\iota_{kj}
 =\iota_{kj}(\proj{0}_{\mathsf E_j}\otimes I).
 \label{eq:zero-label-projection}
\end{equation}

Substituting \eqref{eq:local-from-K} and \eqref{eq:local-H-action} into
\eqref{eq:sequential-block-factorization}, writing the middle product as
\(d^{k-j-1}\iota_{kj}\), and applying the \(\mathsf E_j\) projections in
\eqref{eq:output-projection} and \eqref{eq:zero-label-projection} gives the
following factorization.  The common factor
\(I_{\mathcal R_{<j}}\) is omitted.
\begin{samepage}
\begin{align}
 (\Delta_0)_{kj}
 &=\underbrace{\frac{2}{1+\mu}
 \begin{pmatrix}
  \kappa(\ket0_{\mathsf A_H}\otimes I_{\mathsf S})\\
  \beta(\ket0_{\mathsf A_B}\ket0_{\mathsf F}\otimes I_{\mathsf S})\\0
 \end{pmatrix}}_{\eqref{eq:local-H-action}}
 \underbrace{d^{k-j-1}\iota_{kj}}
 _{\eqref{eq:local-from-H},\,\eqref{eq:zero-label-embedding}}
 \nonumber\\[-0.2em]
 &\times
 \underbrace{\frac{2}{1+\mu}(\ket0_{\mathsf E_j}\otimes I_{\mathsf S})
 \begin{pmatrix}
  -i\kappa(\bra0_{\mathsf A_H}\otimes I_{\mathsf S})&0&
  \beta(\bra0_{\mathsf A_B}\bra0_{\mathsf F}\otimes I_{\mathsf S})
 \end{pmatrix}\Omega}_{\eqref{eq:local-from-K}},
 \nonumber\\
 (\Delta_1)_{kj}
 &=\underbrace{\frac{2}{1+\mu}
 \begin{pmatrix}
  \kappa(\ket0_{\mathsf A_H}\otimes I_{\mathsf S})\\
  \beta(\ket0_{\mathsf A_B}\ket0_{\mathsf F}\otimes I_{\mathsf S})\\0
 \end{pmatrix}}_{\eqref{eq:local-H-action}}
 \underbrace{d^{k-j-1}\iota_{kj}}
 _{\eqref{eq:local-from-H},\,\eqref{eq:zero-label-embedding}}
 \underbrace{(\bra0_{\mathsf A_B}\otimes P_{\rm L})
 \begin{pmatrix}0&I&0\end{pmatrix}\Omega}_{\eqref{eq:local-from-K}}.
 \label{eq:offdiagonal-block-factorization}
\end{align}
\end{samepage}
To bound \((\Delta_s)_{kj}\) for \(s=0,1\), use
\(\iota_{kj}^\dagger\iota_{kj}=I\) and
the following column and row norms:
\begin{equation}
\begin{aligned}
 \left\|
 \begin{pmatrix}\kappa(\ket0_{\mathsf A_H}\otimes I_{\mathsf S})\\
 \beta(\ket0_{\mathsf A_B}\ket0_{\mathsf F}\otimes I_{\mathsf S})\\0\end{pmatrix}
 \right\|&=\sqrt\mu,
 \\[-0.1em]
 \left\|
 \begin{pmatrix}-i\kappa(\bra0_{\mathsf A_H}\otimes I_{\mathsf S})&0&
 \beta(\bra0_{\mathsf A_B}\bra0_{\mathsf F}\otimes I_{\mathsf S})\end{pmatrix}
 \right\|&=\sqrt\mu,
 \\[-0.1em]
 \left\|(\bra0_{\mathsf A_B}\otimes P_{\rm L})
 \begin{pmatrix}0&I&0\end{pmatrix}\Omega\right\|&\leq1.
\end{aligned}
 \label{eq:factor-norms}
\end{equation}
The first two equalities follow from
\(\kappa^2+\beta^2=\mu\) in \eqref{eq:local-scalars}.  The map
\(\ket0_{\mathsf E_j}\otimes I_{\mathsf S}\) is also an isometry.  For the
last bound,
\(\norm{\bra0_{\mathsf A_B}\otimes P_{\rm L}}\leq1\),
\(\left\|\begin{pmatrix}0&I&0\end{pmatrix}\right\|=1\), and
\(\norm\Omega=1\).
Equations \eqref{eq:offdiagonal-block-factorization} and
\eqref{eq:factor-norms} give, for \(k>j\),
\[
 \norm{(\Delta_1)_{kj}}
 \leq\frac{2\sqrt\mu}{1+\mu}d^{k-j-1}
 =\sqrt\theta\,d^{k-j-1},
 \qquad
 \norm{(\Delta_0)_{kj}}
 \leq\frac{4\mu}{(1+\mu)^2}d^{k-j-1}
 =\theta d^{k-j-1}.
\]
Combining \eqref{eq:S11-lower-blocks-zero},
\eqref{eq:zero-private-offdiagonal}, and the definitions in
\eqref{eq:correction-blocks} gives
\begin{equation}
 \Pi_k\Delta\Pi_j=
 \begin{cases}
  0,&k<j,\\
  (\Delta_0)_{jj},&k=j,\\
  (\Delta_0)_{kj}+(\Delta_1)_{kj},&k>j.
 \end{cases}
 \label{eq:blockwise-decomposition}
\end{equation}
Summing over \(j,k\) gives \eqref{eq:private-decomposition}.
\end{proof}

\Needspace{10\baselineskip}
The expansion of \(\mathcal F_q(S_{11})\) will contain ordered products of
\(S_{11}(0),\Delta_1,\Delta_0\).  We next establish a bound for every such
product.

\begin{proposition}[Bound for an ordered product]
\label{prop:ordered-product-bound}
Assume \(\alpha\delta\leq1/2\).  Fix an ordered product
\(T:=X_L\cdots X_1\) with
\(X_i\in\{S_{11}(0),\Delta_1,\Delta_0\}\), containing \(r\) copies of
\(\Delta_1\) and \(u\) copies of \(\Delta_0\).  Then
\begin{equation}
 \norm T\leq
 \theta^{r/2+u}
 \sqrt{\binom{J+r}{r}}\binom{J+u}{u}.
 \label{eq:ordered-product-bound}
\end{equation}
\end{proposition}

\begin{proof}
To expand \(T\) with respect to \(\cK=\bigoplus_j\cK_j\), first note from
\eqref{eq:zero-private-offdiagonal} and \eqref{eq:correction-blocks} that
\[
 \begin{aligned}
 \Pi_kS_{11}(0)\Pi_j&=0 &&(k\ne j),\\
 \Pi_k\Delta_1\Pi_j&=0 &&(k\leq j),\\
 \Pi_k\Delta_0\Pi_j&=0 &&(k<j).
 \end{aligned}
\]
Thus inserting \(I_{\cK}=\sum_j\Pi_j\) between consecutive factors and on
both sides of the product restricts the nonzero terms to
nondecreasing indices:
\begin{equation}
 T=\sum_{0\leq j_0\leq\cdots\leq j_L<J}
 \Pi_{j_L}X_L\Pi_{j_{L-1}}\cdots\Pi_{j_1}X_1\Pi_{j_0}.
 \label{eq:private-summand-expansion}
\end{equation}
For each nonzero term, collect the increments \(j_i-j_{i-1}\) at the
\(\Delta_1\) factors into \(a\) and those at the \(\Delta_0\) factors into
\(b\), retaining their order:
\[
 a:=(j_i-j_{i-1})_{i:X_i=\Delta_1},\qquad
 b:=(j_i-j_{i-1})_{i:X_i=\Delta_0},\qquad
 a\in\mathbb N^r,\quad b\in\mathbb N_0^u.
\]
When \(X_i=S_{11}(0)\), the factor
\(\Pi_{j_i}S_{11}(0)\Pi_{j_{i-1}}\) can be nonzero only if
\(j_i=j_{i-1}\).
Since the order of \(X_1,\ldots,X_L\) is fixed, \(a,b\), and \(j_0\)
determine every \(j_i\), with
\(j_L=j_0+|a|+|b|\), where \(|a|:=\sum_i a_i\) and
\(|b|:=\sum_i b_i\).  For fixed \(a,b\), define \(T_{a,b}\) by summing the
terms in \eqref{eq:private-summand-expansion} over
\(0\leq j_0<J-|a|-|b|\), with all increments fixed by \(a,b\).  Then
\begin{equation}
 T=\sum_b\sum_aT_{a,b}.
 \label{eq:increment-decomposition}
\end{equation}
Only pairs with \(|a|+|b|\leq J-1\) are included.

We first bound each \(T_{a,b}\).
For fixed \(a,b\) and \(j_0\), \(T_{a,b}\Pi_{j_0}\) is a single term in
\eqref{eq:private-summand-expansion}, with input space \(\cK_{j_0}\) and
output space \(\cK_{j_0+|a|+|b|}\).  Hence
\begin{equation}
 T_{a,b}\Pi_{j_0}
 =\Pi_{j_0+|a|+|b|}T_{a,b}\Pi_{j_0},
 \qquad
 T_{a,b}=\bigoplus_{0\leq j_0<J-|a|-|b|}T_{a,b}\Pi_{j_0}.
 \label{eq:fixed-increment-blocks}
\end{equation}
Different \(j_0\) give orthogonal input and output spaces.  Using
submultiplicativity and the contraction bound \eqref{eq:A0-contraction} for
the \(S_{11}(0)\) factors, and then applying
\eqref{eq:delta-one-block-bound}--\eqref{eq:delta-zero-block-bound}, gives
\begin{equation}
 \begin{aligned}
 \norm{T_{a,b}}
 &=\max_{0\leq j_0<J-|a|-|b|}\norm{T_{a,b}\Pi_{j_0}}\\
 &\leq\max_{0\leq j_0<J-|a|-|b|}
 \left[
 \prod_{i:X_i=\Delta_1}\norm{(\Delta_1)_{j_i j_{i-1}}}
 \prod_{i:X_i=\Delta_0}\norm{(\Delta_0)_{j_i j_{i-1}}}
 \right]\\
 &\leq
 \prod_{i=1}^r\bigl(\sqrt\theta\,d^{a_i-1}\bigr)
 \prod_{i:b_i=0}\theta
 \prod_{i:b_i>0}\bigl(\theta d^{b_i-1}\bigr)\\
 &=\theta^{r/2+u}
 d^{\sum_i(a_i-1)+\sum_{i:b_i>0}(b_i-1)}
 \leq\theta^{r/2+u}.
 \end{aligned}
 \label{eq:fixed-increment-bound}
\end{equation}
For fixed \(b\), we next show that distinct \(a\) give orthogonal ranges.
The key is that \(a\) can be recovered from the positions carrying nonzero
output labels.  If \(r=0\),
the list \(a\) is empty and the sum over \(a\) has one term.  For \(r\geq1\),
let \(\Pi_{s_i}\) be the projector immediately to the right of the \(i\)-th
\(\Delta_1\) factor in \eqref{eq:private-summand-expansion}; equivalently,
the input to this factor lies in \(\cK_{s_i}\).  The indices in
\eqref{eq:private-summand-expansion} are nondecreasing, and every
\(\Delta_1\) factor has a positive increment.  Therefore
\(s_1<\cdots<s_r\), and the \(i\)-th such factor is
\((\Delta_1)_{s_i+a_i,s_i}:\cK_{s_i}\to\cK_{s_i+a_i}\).

We now track the Kraus labels through these factors.  Equations
\eqref{eq:zero-private} and
\eqref{eq:delta-zero-diagonal-form} show that both
\(\left.S_{11}(0)\right|_{\cK_j}\) and \((\Delta_0)_{jj}\) have the form
\(I_{\mathcal R_{<j}}\otimes(\cdot)\).  Hence they leave the Kraus labels in
\(\mathsf E_0,\ldots,\mathsf E_{j-1}\) unchanged.
For \(k>j\), the labels added by \((\Delta_s)_{kj}\) can be read from
\eqref{eq:offdiagonal-block-factorization}.  The omitted factor
\(I_{\mathcal R_{<j}}\) leaves the input labels
\((\ell_0,\ldots,\ell_{j-1})\) unchanged.  The upper entry in
\eqref{eq:local-from-K} is supported on \(\ket0_{\mathsf E_j}\), whereas the
lower entry is supported on labels
\(\ell\in\{1,\ldots,m\}\).  The embedding \(\iota_{kj}\)
in \eqref{eq:zero-label-embedding} sets the remaining labels
\(\mathsf E_{j+1},\ldots,\mathsf E_{k-1}\) to zero.  Therefore
\[
 \begin{array}{c|l}
 \text{map}&\text{possible output Kraus labels}\\[0.3em]
 (\Delta_0)_{kj}
   &(\ell_0,\ldots,\ell_{j-1},\underbrace{0,\ldots,0}_{k-j})\\[0.8em]
 (\Delta_1)_{kj}
   &(\ell_0,\ldots,\ell_{j-1},\ell,
     \underbrace{0,\ldots,0}_{k-j-1}),\quad1\leq\ell\leq m.
 \end{array}
\]
Consequently, for every Kraus-label string
\((\ell_0,\ldots,\ell_{j_L-1})\) appearing in the output of
\(T_{a,b}\Pi_{j_0}\), the positions carrying nonzero labels satisfy
\[
 \{0\leq j<j_L:\ell_j\ne0\}
 =\underbrace{\{0\leq j<j_0:\ell_j\ne0\}}_{\text{from the input}}
 \cup\underbrace{\{s_1,\ldots,s_r\}}_{\text{from the \(\Delta_1\) factors}},
 \qquad j_0\leq s_1<\cdots<s_r.
\]
All nonzero input labels occur at positions below \(j_0\leq s_1\).
Therefore the output determines \(s_1,\ldots,s_r\) as the \(r\) largest
positions carrying nonzero labels.  Figure~\ref{fig:output-labels}
illustrates a product containing two \(\Delta_1\) factors.

For \(1\leq i<r\), let \(c_i\) be the sum of the \(\Delta_0\)
increments between the \(i\)-th and \((i+1)\)-st \(\Delta_1\) factors;
let \(c_r\) be the sum after the last \(\Delta_1\) factor.  These sums
are determined by \(b\) and the fixed order of \(X_1,\ldots,X_L\).
The \(i\)-th \(\Delta_1\) factor maps into \(\cK_{s_i+a_i}\).  Thus, setting
\(s_{r+1}:=j_L\), we obtain
\begin{equation}
 s_{i+1}=(s_i+a_i)+c_i,
 \qquad
 a_i=s_{i+1}-s_i-c_i
 \quad(1\leq i\leq r).
 \label{eq:recover-increments}
\end{equation}
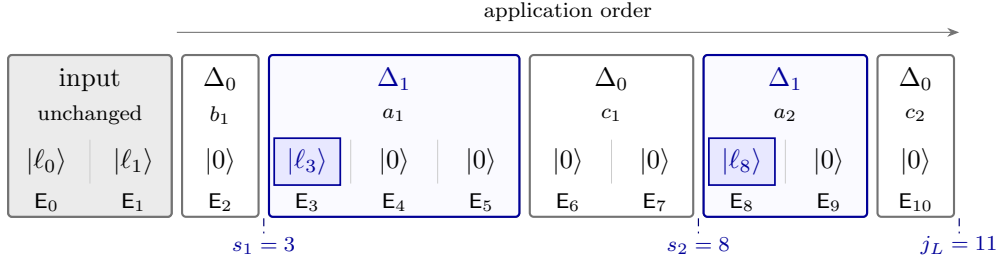
\begin{figure}[!tb]
 \centering
 \begin{tikzpicture}[
  x=1.15cm,y=1cm,font=\small,
  >={Stealth[length=1.6mm,width=1.2mm]},
  every node/.style={inner sep=2pt},
  block group/.style={draw=black!55,fill=white,line width=0.8pt,
                      rounded corners=1.5pt},
  nonzero label/.style={draw=blue!60!black,fill=blue!9,line width=0.65pt}
 ]
 % Block widths equal the numbers of new registers in this example.
 \foreach \l/\r in {0/2,2/3,3/6,6/8,8/10,10/11} {
  \draw[block group] ({\l+0.06},-0.42) rectangle ({\r-0.06},1.72);
 }
 \draw[block group,fill=black!8] (0.06,-0.42) rectangle (1.94,1.72);
 \foreach \l/\r in {3/6,8/10} {
  \draw[block group,draw=blue!60!black,fill=blue!2]
   ({\l+0.06},-0.42) rectangle ({\r-0.06},1.72);
 }
 \draw[->,draw=black!55] (2,2.02)--(11,2.02)
  node[midway,above=2pt,font=\scriptsize] {application order};
 \node at (1,1.40) {input};
 \node at (2.5,1.40) {$\Delta_0$};
 \node[text=blue!60!black] at (4.5,1.40) {$\Delta_1$};
 \node at (7,1.40) {$\Delta_0$};
 \node[text=blue!60!black] at (9,1.40) {$\Delta_1$};
 \node at (10.5,1.40) {$\Delta_0$};
 \foreach \l/\r/\lbl in {
  0/2/{\text{unchanged}},2/3/{b_1},3/6/{a_1},
  6/8/{c_1},8/10/{a_2},10/11/{c_2}} {
  \node[font=\scriptsize] at ({(\l+\r)/2},0.93) {$\lbl$};
 }
 % The full row is a basis label of E_0 ... E_10 in the output K_11.
 \foreach \e in {1,4,5,7,9} {
  \draw[draw=black!20,line width=0.35pt] (\e,0.02)--(\e,0.60);
 }
 \foreach \e in {0,...,10} {
  \node[font=\scriptsize] at ({\e+0.5},-0.23) {$\mathsf E_{\e}$};
 }
 \foreach \e in {2,4,5,6,7,9,10} {
  \node at ({\e+0.5},0.31) {$\ket0$};
 }
 \foreach \e in {0,1} {
  \node at ({\e+0.5},0.31) {$\ket{\ell_{\e}}$};
 }
 \foreach \e in {3,8} {
  \draw[nonzero label] ({\e+0.12},0) rectangle ({\e+0.88},0.62);
  \node[text=blue!60!black] at ({\e+0.5},0.31) {$\ket{\ell_{\e}}$};
 }
 % The last two nonzero positions and the final block index recover a.
 \foreach \x/\lbl in {3/{s_1=3},8/{s_2=8},11/{j_L=11}} {
  \draw[densely dashed,draw=blue!45!black,line width=0.45pt]
   (\x,-0.44)--(\x,-0.64);
  \node[text=blue!60!black,font=\scriptsize] at (\x,-0.80) {$\lbl$};
 }
 \end{tikzpicture}
 \caption{Output Kraus labels for one summand of
 \eqref{eq:private-summand-expansion}.  The horizontal widths show the initial
 index \(j_0\) and the subsequent increments.
 The two \(\Delta_1\) factors create nonzero labels at
 $s_1=3$ and $s_2=8$.  The increments $c_1=b_2$ and $c_2=b_3$ are fixed
 by $b$; \eqref{eq:recover-increments} then determines $a$.}
 \label{fig:output-labels}
\end{figure}
For fixed \(b\), the final index \(j_L\) and the \(r\) largest positions
carrying nonzero output labels therefore determine \(a\).  Thus, for fixed
\(j_L\), distinct \(a,a'\) give outputs supported on disjoint sets of
Kraus-label basis states in \(\cK_{j_L}\).  Their ranges are orthogonal;
outputs in different \(\cK_{j_L}\) are also orthogonal.  Hence
\begin{equation}
 T_{a,b}^\dagger T_{a',b}=0\qquad(a\neq a').
 \label{eq:increment-orthogonality}
\end{equation}
\Needspace{7\baselineskip}
Let \(\#\{a\}\) and \(\#\{b\}\) count all values of \(a\) and \(b\)
that occur in the decomposition of \(T\).  In particular, \(\#\{a\}\)
bounds the number of terms for any fixed \(b\).  Using
\eqref{eq:increment-orthogonality} and \eqref{eq:fixed-increment-bound},
\begin{equation}
 \begin{aligned}
 \left\|\sum_aT_{a,b}\right\|^2
 &=\left\|\sum_{a,a'}T_{a,b}^\dagger T_{a',b}\right\|
 =\left\|\sum_aT_{a,b}^\dagger T_{a,b}\right\|\\
 &\leq\sum_a\norm{T_{a,b}}^2
 \leq\#\{a\}\,\theta^{r+2u}.
 \end{aligned}
 \label{eq:fixed-b-bound}
\end{equation}

We now count the possible values of \(a\) and \(b\).  Every pair that occurs has
\(|a|+|b|=j_L-j_0\leq J-1\), so \(|a|\leq J\) and \(|b|\leq J\).
Relaxing the condition \(a_i\geq1\) to \(a_i\geq0\) only enlarges the count.
After adjoining \(a_{r+1}=J-|a|\) and \(b_{u+1}=J-|b|\), we count
nonnegative integer tuples with a fixed sum:
\begin{equation}
 \begin{aligned}
 \#\{a\}
 &\leq\#\left\{a\in\mathbb N_0^{r+1}:
       \sum_{i=1}^{r+1}a_i=J\right\}
 =\binom{J+r}{r},\\
 \#\{b\}
 &\leq\#\left\{b\in\mathbb N_0^{u+1}:
       \sum_{i=1}^{u+1}b_i=J\right\}
 =\binom{J+u}{u}.
 \end{aligned}
 \label{eq:increment-count}
\end{equation}
Using the decomposition \eqref{eq:increment-decomposition}, taking the
triangle inequality over \(b\), and substituting
\eqref{eq:fixed-b-bound} and \eqref{eq:increment-count} yields
\[
 \norm T
 \leq\sum_b\left\|\sum_aT_{a,b}\right\|
 \leq\#\{b\}\sqrt{\#\{a\}}\,\theta^{r/2+u}
 \leq\theta^{r/2+u}
 \sqrt{\binom{J+r}{r}}\binom{J+u}{u}.
\]
\end{proof}

With Proposition~\ref{prop:ordered-product-bound} established, we now
complete the proof of Proposition~\ref{prop:private-polynomial}.

\begin{proof}[Proof of Proposition~\ref{prop:private-polynomial}]
Substitute \(S_{11}=S_{11}(0)+\Delta\) into
\(p(S_{11})=S_{11}^2+S_{11}^4\).  Using
\(p(S_{11}(0))=0\) from \eqref{eq:S11-zero-polynomial} yields
\begin{equation}
 \begin{aligned}
 p(S_{11})
 &=\sum_{\ell\in\{2,4\}}
 \bigl[(S_{11}(0)+\Delta)^\ell-S_{11}(0)^\ell\bigr]\\
 &=\sum_{\ell\in\{2,4\}}
 \ \sum_{\substack{X_1,\ldots,X_\ell\in\{S_{11}(0),\Delta\}\\
                    X_i=\Delta\text{ for some }i}}
 X_\ell\cdots X_1.
\end{aligned}
\label{eq:pA-expansion}
\end{equation}
For each \(\ell\), the inner sum contains \(2^\ell-1\) terms: each of the
\(\ell\) factors has two choices, and the term \(S_{11}(0)^\ell\) is
excluded.  Each term contains between one and \(\ell\) factors \(\Delta\).
Since \(\ell\in\{2,4\}\), \eqref{eq:pA-expansion} has
\((2^2-1)+(2^4-1)=18\) terms, each containing one to four factors \(\Delta\).
Expanding \(p(S_{11})^{2q}\) therefore gives
at most \(18^{2q}\) ordered products, each containing between \(2q\) and
\(8q\) factors \(\Delta\).  Using \eqref{eq:private-decomposition} to replace
each \(\Delta\) by \(\Delta_1+\Delta_0\) splits each product into at most
\(2^{8q}\)
ordered products.  Let \(T\) range over the resulting products, counted with
multiplicity.  The triangle inequality gives
\begin{equation}
 \left\|\mathcal F_q(S_{11})\right\|
 \leq2^{-2q}18^{2q}2^{8q}\max_T\norm T.
 \label{eq:expanded-product-bound}
\end{equation}

It remains to bound each \(T\).  Let \(r\) and \(u\) count its
\(\Delta_1\) and \(\Delta_0\) factors, and set \(n=r/2+u\), the
exponent of \(\theta\) in Proposition~\ref{prop:ordered-product-bound}.
Since \(2q\leq r+u\leq8q\),
\begin{equation}
 q\leq\frac{r+u}{2}\leq n\leq r+u\leq8q.
 \label{eq:correction-count}
\end{equation}
By \eqref{eq:d-theta-bounds} and
\(\alpha\delta=\tau/J\) from \eqref{eq:local-scalars},
\(J\theta\leq2\tau\).  Applying
Proposition~\ref{prop:ordered-product-bound} and
\(\binom Mk\leq(eM/k)^k\), with \(r,u\leq8q\leq J\), gives the following
estimate.  When \(r=0\) or \(u=0\), the corresponding factor is omitted:
\[
 \begin{aligned}
 \norm T
 &\leq\theta^{r/2+u}
 \left(\frac{e(J+r)}r\right)^{r/2}
 \left(\frac{e(J+u)}u\right)^u\\
 &\leq\left(\frac{2eJ\theta}{r}\right)^{r/2}
       \left(\frac{2eJ\theta}{u}\right)^u
 \leq\frac{(4e\tau)^n}{r^{r/2}u^u}.
 \end{aligned}
\]
To express the denominator
in terms of \(n=r/2+u\), convexity of \(x\log x\), with
\(0\log0=0\), gives
\begin{equation}
 \begin{aligned}
 \frac r2\log\frac r2+u\log u&\geq n\log(n/2),\\
 r^{r/2}u^u&\geq(r/2)^{r/2}u^u\geq(n/2)^n.
 \end{aligned}
 \label{eq:weighted-denominator}
\end{equation}
Consequently, with \(C=8e\), every product satisfies
\begin{equation}
 \norm T\leq\frac{(4e\tau)^n}{(n/2)^n}
 =\left(\frac{C\tau}{n}\right)^n.
 \label{eq:single-product-bound}
\end{equation}

\Needspace{9\baselineskip}
If \(\tau=0\), every product is zero by \eqref{eq:single-product-bound}.
Otherwise, choose \(C_0\geq20736C\).  Since
\(n\geq q\geq C_0\tau\geq C\tau\),
\begin{equation}
 \left(\frac{C\tau}{n}\right)^n
 \leq\left(\frac{C\tau}{q}\right)^n
 \leq\left(\frac{C\tau}{q}\right)^q.
 \label{eq:decreasing-base}
\end{equation}
Substituting \eqref{eq:single-product-bound} and \eqref{eq:decreasing-base} into
\eqref{eq:expanded-product-bound}, with \(2^{-2q}18^{2q}2^{8q}=20736^q\), gives
\[
 \left\|\mathcal F_q(S_{11})\right\|
 \leq20736^q\left(\frac{C\tau}{q}\right)^q
 \leq\left(\frac{C_0\tau}{q}\right)^q.\qedhere
\]
\end{proof}

The time-dependent construction in Section~\ref{sec:time-dependent} has the
same decomposition into the private spaces \(\cK_j\), but the query on
\(\cK_j\) is \(\Omega_j\) rather than a common \(\Omega\).  We conclude with
the corresponding polynomial bound used in that section.

\begin{proposition}[Time-dependent polynomial bound for the private block]
\label{prop:td-private-polynomial}
Under the assumptions and notation of Theorem~\ref{thm:time-dependent}, let
\(S\) be the transducer in \eqref{eq:td-transducer-factorization}, with
private-to-private block \(S_{11}\).
If \(\alpha\delta\leq1/2\), \(J\geq8q\), and \(q\geq C_0\tau\), then
\begin{equation}
 \norm{\mathcal F_q(S_{11})}
 \leq\left(\frac{C_0\tau}{q}\right)^q.
 \label{eq:td-factorial-bound}
\end{equation}
\end{proposition}

\Needspace{4\baselineskip}
\begin{proof}
The proof of Proposition~\ref{prop:private-polynomial} applies with \(\Omega\)
on \(\cK_j\) replaced by \(\Omega_j\): it uses the oracle assumptions
separately for each \(j\) and never requires \(\Omega_j=\Omega_k\) for
\(j\ne k\).
\end{proof}

\bibliographystyle{alphaurl}
\bibliography{lindbladian_references}

\end{document}